\documentclass[superscriptaddress,aps,onecolumn,11pt,tightenlines]{revtex4-2}
\usepackage{float}
\usepackage{amsmath}
\usepackage{amssymb}
\usepackage{mathtools}
\usepackage{amsfonts}
\usepackage{hyperref}
\usepackage{amsthm}
\usepackage{bm}
\usepackage{braket}
\usepackage{cases}
\usepackage{comment}
\usepackage{graphicx,color}
\usepackage{color}
\usepackage{booktabs}
\usepackage{ascmac}
\usepackage{framed}
\usepackage{xpatch}
\usepackage{stmaryrd}
\usepackage{enumitem}
\xpatchcmd{\proof}{\hskip\labelsep}{\hskip4.3\labelsep}{}{}

\newcommand{\ketbra}[2]{\ket{#1}\!\bra{#2}}
\newcommand{\tr}{\mathrm{Tr}}

\newcommand{\sysA}{A}
\newcommand{\sysB}{B}
\newcommand{\sysC}{C}
\newcommand{\probP}{\bm{P}}
\newcommand{\probrho}{\bm{P}}

\makeatletter
  \renewcommand\@upn{\itshape}
\makeatother
\newtheorem{theorem}{Theorem}
\newtheorem{lemma}{Lemma}
\newtheorem{proposition}{Proposition}

\newtheorem{definition}{Definition}
\newtheorem{remark}{Remark}

\def\({\left(}

\def\){\right)}

\begin{document}
\title{Asymptotically tight security analysis of quantum key distribution based on universal source compression} 
\author{Takaya Matsuura}
 \email{takayamatsuura@gmail.com}
 \affiliation{RIKEN Center for Quantum Computing (RQC), Hirosawa 2-1, Wako, Saitama 351-0198, Japan} 
 \author{Shinichiro Yamano}
 \affiliation{Department of Applied Physics, Graduate School of Engineering, The University of Tokyo, 7-3-1 Hongo, Bunkyo-ku, Tokyo 113-8656, Japan}
 \author{Yui Kuramochi}
 \affiliation{Graduate School and Faculty of Information Science and Electrical Engineering, Kyushu University, 744 Motooka, Nishi-ku, Fukuoka 819-0395, Japan}
\author{Toshihiko Sasaki}
\affiliation{Quantinuum K.K., 1-9-2 Otemachi, Chiyoda-ku, Tokyo 100-0004, Japan}
\author{Masato Koashi}
 \affiliation{Department of Applied Physics, Graduate School of Engineering, The University of Tokyo, 7-3-1 Hongo, Bunkyo-ku, Tokyo 113-8656, Japan} 
 \affiliation{Photon Science Center, Graduate School of Engineering, The University of Tokyo, 7-3-1 Hongo, Bunkyo-ku, Tokyo 113-8656, Japan}

\begin{abstract}
    Practical quantum key distribution (QKD) protocols require a finite-size security proof.  The phase error correction (PEC) approach is one of the general strategies for security analyses that has successfully proved finite-size security for many protocols.  However, the conventional PEC approach cannot achieve the asymptotically optimal key rate in general, as long as the failure probability of PEC is estimated through the phase error rate. In this work, we propose a new PEC-type strategy that can provably achieve the asymptotically optimal key rate. The key piece for this is a virtual protocol based on universal source compression with quantum side information, which is of independent interest. A universal source compression with quantum side information protocol is first constructed for fixed-length independent and identically distributed (i.i.d.)~setups and then extended to adaptive-length setups with the restrictions on possible states imposed by joint random variables.
    Combined with the reduction method to collective attacks, this enables us to tightly evaluate the failure probability of PEC for permutation-symmetric QKD protocols, and thus leads to asymptotically tight analyses. As a result, the security of any permutation-symmetrizable QKD protocol gets reduced to the estimation problem of a single conditional R\'enyi entropy, which can be efficiently solved by a convex optimization.
\end{abstract}

\maketitle

\section{Introduction}
One of the important aspects of information theory is to find a fundamental connection or a duality between different information-theoretic tasks.  An example of such a duality in quantum information theory is the security of quantum key distribution (QKD) and error correction.
In QKD, one can reduce the security of the key to the error correctability of a binary string in a basis complementary to the one that defines the key~\cite{Lo1999, Shor2000, Koashi2009}.  In this security-proof framework called phase error correction (PEC), the length of the key that needs to be shortened in privacy amplification corresponds to a required amount of virtual syndrome extraction to correct the error string in a complementary basis.
First used in the security proof of the Bennett-Brassard 1984 (BB84) protocol~\cite{Bennett1984}, PEC has succeeded in proving the security of various QKD protocols in realistic scenarios, especially with finite numbers of communication rounds (see e.g.~\cite{Shor2000, Lo2001, Tamaki2003, Boileau2005} for early studies).
Later, another approach to QKD security proofs appeared~\cite{Renner2005, Renner2008, Tomamichel2011}, which generalized the leftover hash lemma (LHL)~\cite{Stinson2001} against a classical adversary to that against a quantum adversary.  This purely information-theoretic approach to the QKD security gives a tight bound on a secure key rate while the evaluation of the relevant information-theoretic quantity, the smooth conditional min-entropy~\cite{Renner2008}, is often difficult.  On the other hand, the PEC approach is relatively easy to give a lower bound on the secure key rate since the construction of ``an'' error correction procedure automatically implies an amount of the extractable key.  Since PEC is virtually performed between a sender Alice and a receiver Bob, one can freely restrict the ability of Alice and Bob to simplify the problem at the cost of a worse key rate.  In fact, restricting Alice's and Bob's capabilities leads to the overestimation of an eavesdropper Eve's attack strategy, which thus results in a pessimistic but still secure key rate.  The duality between the security of QKD and PEC thus resembles that between the primal and dual problems of an optimization problem.

A major drawback of the PEC approach is that the conventional analysis~\cite{Shor2000, Koashi2009} may not achieve the asymptotically optimal key rate for a given QKD protocol~\cite{Devetak2005}, which has been pointed out by several works~\cite{Furrer2014, Tsurumaru2020, Tsurumaru2022}.  This is attributed to the fact that the conventional analysis estimates the phase error patterns through virtual round-wise measurements between Alice and Bob. This is to ease the evaluation of the failure probability of PEC, which is eventually characterized by the phase error rate defined through Alice's and Bob's virtual measurements. Thus, if Alice extracts a key bit from the Pauli-$Z$ basis of the system $K$ of the Alice's and Bob's joint system $KQ$, then the asymptotic rate of the syndrome they need to extract for the PEC is given by $H(X|Q)_{{\cal P}_{K\to X}\otimes {\cal M}_{Q}(\rho_{KQ})}$, where ${\cal P}_{K\to X}$ maps the Pauli-$X$ basis of the system $K$ to the classical register $X$, and ${\cal M}_{Q}$ denotes a measurement channel on Alice's and Bob's joint system $Q$ whose POVM element defines the phase error rate for a classical-quantum state ${\cal P}_{K\to X}(\rho_{KQ})$.  Any heuristic choice of ${\cal M}_{Q}$ works for a security proof, but it is better to choose ${\cal M}_{Q}$ such that $H(X|Q)_{{\cal P}_{K\to X}\otimes {\cal M}_{Q}(\rho_{KQ})}$ is smaller and easier to upper-bound.
However, for correcting an $n$-bit phase-error pattern, i.e., the $X$-basis error pattern, in the system $K^n$, this individual measurement ${\cal M}_{Q}^{\otimes n}$ is not the optimal strategy.
It is known that by performing a globally optimal measurement on $Q^n$ depending on the $X$-basis syndrome of the system $K^n$ with the rate $H(X|Q)_{{\cal P}_{K\to X}(\rho_{KQ})}$, one can uniquely identify the $X$-basis error pattern and thus correct it with unit probability in the limit $n\to \infty$~\cite{Devetak2003, Renes2010, Renes2012, Tsurumaru2020, Cheng2021, Tsurumaru2022, Renes2023}. This information-theoretic task has been studied under the name ``classical source compression with quantum side information'' or ``classical-quantum Slepian Wolf.'' The syndrome extraction of the rate $H(X|Q)_{{\cal P}_{K\to X}(\rho_{KQ})}$ leads to the final key rate $1 - H(X|Q)_{{\cal P}_{K \to X}(\rho_{KQ})}$, which is then equal to $H(Z|E)_{{\cal P}_{K \to Z}(\psi_{KQE})}$ from the entropic uncertainty relation~\cite{Coles2012,Coles2017,Renes2010,Tsurumaru2022,Renes2023}, where $\psi_{KQE}$ is a purification of $\rho_{KQ}$.  This final key rate is the same as that concluded from the LHL.

What is important here is that the gap $\min_{{\cal M}_Q} H(X|Q)_{{\cal P}_{K\to X} \otimes {\cal M}_{Q}(\rho_{KQ})} - H(X|Q)_{{\cal P}_{K\to X}(\rho_{KQ})}$ corresponds to the quantum discord~\cite{Ollivier2001} of the state ${\cal P}_{K\to X}(\rho_{KQ})$ and is thus strictly larger than zero except for the special case.
This means that even with the best choice of ${\cal M}_Q$ (i.e., the phase error), the conventional PEC analysis overestimates the required rate of syndrome extraction as $H(X|Q)_{{\cal P}_{K\to X}\otimes {\cal M}_{Q}(\rho_{KQ})}$ instead of $H(X|Q)_{{\cal P}_{K\to X}(\rho_{KQ})}$ due to the suboptimal measurement strategy, and thus it cannot achieve the asymptotically optimal key rate.

To overcome this issue, one needs to modify the PEC procedure so that it can incorporate a global measurement in the same way as is done in classical source compression with quantum side information.  This is exactly what is proposed in Refs.~\cite{Tsurumaru2020, Tsurumaru2022}, which tried to show the equivalence between PEC-based and LHL-based analyses.  However, this is not the end of the story; for the protocol in Refs.~\cite{Tsurumaru2020, Tsurumaru2022} to work, one needs to estimate an entropic quantity of a global $n$-body state from a few parameters of the state obtained through round-wise measurements in actual QKD protocols. 
Thus, their protocol essentially says that one can construct a PEC-based security analysis whenever one can construct an LHL-based security analysis, which eliminates the advantage of the PEC-based approach, namely, easier constructions of finite-size security proofs.  
To leverage the implication from Refs.~\cite{Tsurumaru2020, Tsurumaru2022} while preserving the advantage of the PEC-based approach, it is necessary to construct a classical source compression protocol with quantum side information that has a certified failure probability based solely on estimated parameters in actual QKD protocols.

Here, we first develop a universal decoder for classical source compression with quantum side information with an explicit bound on the failure probability.  For an i.i.d.~quantum state $\rho_{KQ}^{\otimes n}$, our universal decoder works even if one does not know the state $\rho_{KQ}$ itself---it can be constructed solely from (an upper bound on) a conditional R\'enyi entropy of the state ${\cal P}_{K\to X}(\rho_{KQ})$.  We then develop a PEC-type QKD security proof based on this universal source compression protocol.  Since an upper bound on the conditional R\'enyi entropy can be obtained through the estimated parameters in a QKD protocol by a convex optimization, one can construct a security proof that achieves the asymptotically optimal rate~\cite{Devetak2005} in combination with the reduction method to collective attacks~\cite{Tamaki2003, Christandl2009, Fawzi2015, Matsuura2024, Nahar2024}.
We thus obtain the PEC-based approach that can reproduce the results of the LHL-based approach at the finite-size security-proof level in a large block length limit.  Since our security proof can be completed solely with the state on Alice's and Bob's joint system, i.e., without an adversarial quantum system that is hard to characterize, our method is potentially more tractable to evaluate a necessary amount of privacy amplification than the LHL-based approach.  In fact, the non-necessity of taking an adversarial state into account has already led to better performance in the reduction to collective attacks (see Ref.~\cite{Matsuura2024}).  Furthermore, as a byproduct of this new approach, we generalize a PEC-based security proof to a general base-$p$ number while the original approach is limited to the binary number~\cite{Koashi2009}.

The paper is organized as follows.
In Sec.~\ref{sec:cq_slepian_wolf}, we develop a universal decoder for classical source compression with quantum side information, which is of independent interest.  After the problem setup in Sec.~\ref{sec:problem_setup} and the preliminaries in Sec.~\ref{sec:preliminaries}, the main result of this section is stated as Theorem~\ref{thm:cq_slepian_wolf} in Sec.~\ref{sec:construction_universal}.  For discussion on the (sub)optimality of the error exponent of our protocol, see Remark~\ref{rem:exponent}. Theorem~\ref{thm:cq_slepian_wolf} is then adjusted to the form that is directly applicable to QKD security proofs in Proposition~\ref{prop:modified_for_qkd}.
In Sec.~\ref{sec:security_proof_universal}, we develop a PEC-type security proof based on the (partially) universal decoding for classical source compression with quantum side information developed in Sec.~\ref{sec:cq_slepian_wolf}.  In fact, our security proof holds not only for qubit-based protocols but also for any finite dimensional protocols, which is also an extension of the original PEC~\cite{Shor2000, Koashi2009}. 
Section~\ref{sec:complementary_bases} is to explain the definition of the complementary bases in general dimension, and Sec.~\ref{sec:virtual_protocol} is to introduce a virtual protocol for this general dimensional protocol. Then, Sec.~\ref{sec:connection_partially_universal} is devoted to connecting the failure probability of the PEC to that of classical source compression with quantum side information, which is done through yet another protocol called the estimation protocol. Section~\ref{sec:iid_reduction} reduces the security against general attacks to that against collective attacks, and then Sec.~\ref{sec:estimation_failure_prob} estimates the failure probability of the PEC, completing the security proof. Section~\ref{sec:asymptotic_optimality} discusses the asymptotic optimality of our new PEC method by comparing our asymptotic key rate formula with the Devetak-Winter rate. A short section, Sec.~\ref{sec:parameter_choice}, explains a good initial guess of the optimization parameter.
We numerically demonstrate the improvement of the key rate with our new method in Sec.~\ref{sec:numerical_comparison} by applying the conventional PEC-based analysis (Sec.~\ref{sec:conventional_PEC}) and our new analysis (Sec.~\ref{sec:new_PEC}) to the Bennett1992 (B92) protocol~\cite{Bennett1992}. The numerical simulation of the comparison of the key rate can be found in Sec.~\ref{sec:numerical_simulation}.  Finally, in Sec.~\ref{sec:discussion}, we wrap up our paper with possible future directions.

\section{Universal decoder for classical source compression with quantum side information} \label{sec:cq_slepian_wolf}
In this section, we develop a universal decoder for classical source compression with quantum side information or the classical-quantum (c-q) Slepian-Wolf problem~\cite{Devetak2003}.  
The encoding function of our protocol is based on a random construction, which is natural in the subsequent application to the QKD, and therefore our protocol is not completely universal in the sense that we cannot find a fixed encoder independently of the information source.  Nevertheless, our random encoding function is independent of the information source, so our protocol can be regarded as universal classical source compression with quantum side information and shared randomness between an encoder and a decoder.

Due to the duality of source compression with quantum side information and c-q channel coding~\cite{Cheng2018}, the problem of universal classical source compression with quantum side information is closely related to the universal c-q channel coding, which has already been constructed~\cite{Hayashi2009}.  However, since the classical information source is also probabilistic in universal source compression with quantum side information, the fixed-type encoding used for the universal c-q channel coding in Ref.~\cite{Hayashi2009} is incompatible.  Here, we explicitly construct a random encoder and a universal decoder for source compression with quantum side information without the knowledge of the source c-q state.  The obtained error exponent is slightly better than the one naively expected from the result in Ref.~\cite{Hayashi2009} due to the improved decoder based on the result in Ref.~\cite{Beigi2023}, but worse than the case of a known c-q state source~\cite{Cheng2021, Renes2023}.  For later use in the security analysis of QKD, we also develop a partially universal decoding strategy, i.e., the classical probability distribution of the information source is known while the quantum state of the side information is unknown.
The decoding error decreases in this case compared to the fully universal case, but it is only a subexponential improvement.

\subsection{Problem setups} \label{sec:problem_setup}

\begin{figure}[t]
    \centering
    \includegraphics[width=0.9\textwidth]{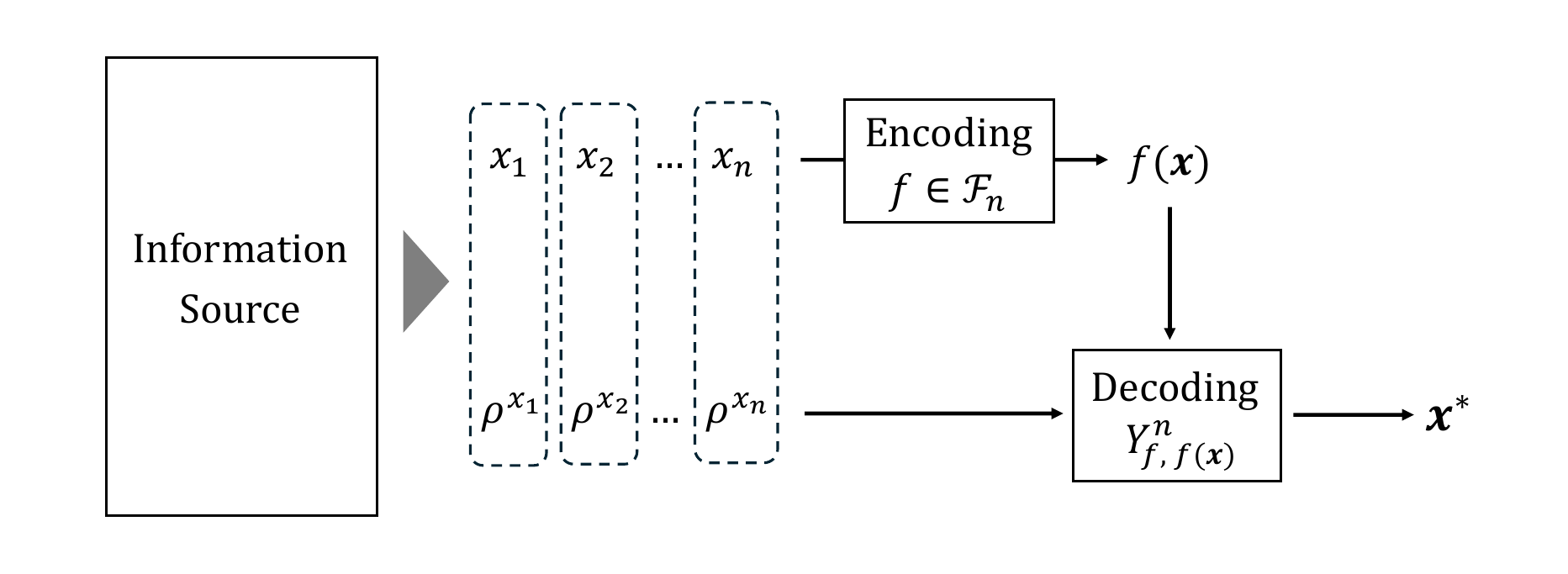}
    \caption{A schematic picture of classical source compression with quantum side information.  In our setup, the set ${\cal F}_n$ of encoding functions and the decoding POVM $Y^n$ are constructed so that they do not depend on the classical-quantum state $\rho_{XQ}$ of the information source.}
    \label{fig:cq_slepian_wolf}
\end{figure}

The information source outputs an unknown i.i.d.~c-q state $\rho_{XQ}^{\otimes n}$ with 
\begin{equation}
    \rho_{XQ}\coloneqq \sum_{x\in{\cal X}}p(x)\ketbra{x}{x}_X\otimes \rho_Q^x, \label{eq:source_state}
\end{equation}
where $X$ denotes an $|{\cal X}|$-dimensional classical system and $B$ denotes a $d$-dimensional quantum system.
An encoder has a set ${\cal F}_n$ of functions that compresses a classical random length-$n$ string $\bm{x}\in{\cal X}^n$ to a bin $b\in{\cal B}_n$.  The encoder and a decoder share randomness to specify an encoding function $f\in{\cal F}_n$.  With a randomly chosen function $f:{\cal X}^n\to{\cal B}_n$ from ${\cal F}_n$, the encoder sends the bin $f(\bm{x})$ to the decoder.  Given the bin $b\in{\cal B}_n$ and the state $\rho_{Q^n}^{\bm{x}}=\rho_Q^{x_1}\otimes\cdots\otimes\rho_Q^{x_n}$, the decoder tries to recover the length-$n$ string $\bm{x}$ in such a way that a decoding POVM does not depend on $\rho_{XB}$.  
If we denote the decoding POVM when given the encoding function $f\in{\cal F}_n$ and the bin $b\in{\cal B}_n$ as $\{Y^n_{f,b}(\bm{y})\}_{\bm{y}\in f^{-1}(b)}$, where $f^{-1}$ is the preimage of the function $f$, then the decoding error $P_{\rm err}({\cal F}_n, Y^n)_{\rho}$ of this protocol averaged over encoding functions is given by 
\begin{align}
    P_{\rm err}({\cal F}_n,Y^n)_{\rho} &= \mathbb{E}_{f\sim |{\cal F}_n|^{-1}}\!\left[\sum_{\bm{x}\in{\cal X}^n} p^n(\bm{x}) \tr\bigl[\rho_{Q^n}^{\bm{x}}(I_{Q^n} - Y^n_{f,f(\bm{x})}(\bm{x}))\bigr]\right] \label{eq:pre_def_of_error}\\
    &= \mathbb{E}_{f\sim |{\cal F}_n|^{-1}}\,\tr\left[\rho_{XQ}^{\otimes n} \left(I_{X^n Q^n} - \sum_{\bm{x}\in{\cal X}^n}\ketbra{\bm{x}}{\bm{x}}_{X^n} \otimes Y^n_{f,f(\bm{x})}(\bm{x})\right)\right], \label{eq:def_of_error_prob}
\end{align}  
where $p^n(\bm{x})\coloneqq p(x_1)\cdots p(x_n)$.
The asymptotic compression rate $R$ of this protocol is defined as 
\begin{equation}
    R\coloneqq \lim_{n\to \infty}\frac{\log|{\cal B}_n|}{n}.
\end{equation}
The protocol is said to achieve the asymptotic compression rate $R$ by the random encoders $\{{\cal F}_n\}_{n\in\mathbb{N}}$ and the universal decoder $\{Y^n\}_{n\in\mathbb{N}}$ if the error $P_{\rm err}({\cal F}_n, Y^n)_{\rho}$ is asymptotically vanishing.
The setup is illustrated in Fig.~\ref{fig:cq_slepian_wolf}.

\subsection{Preliminaries: universal symmetric state, type theory, and entropic quantities} \label{sec:preliminaries}
Throughout the paper, the base of the logarithm is taken to be $2$.
Let ${\cal H} \coloneqq (\mathbb{C}^d)^{\otimes n}$. Let ${\cal L}({\cal H})$ be the set of linear operators on ${\cal H}$, and let ${\cal D}({\cal H})$ be its subset of density operators.
Let $Y_n^d$ be the set of Young diagrams with $n$ boxes and at most $d$ rows.
An element of $Y_n^d$ can be labeled by $\bm{n}=(n_1, n_2, \ldots, n_d)$ with $n_1 \geq n_2 \geq \cdots n_d$ and $\sum_{i=1}^{d} n_i = n$.
Then, from Schur-Weyl duality, ${\cal H}$ can be decomposed into
\begin{equation}
    {\cal H} = \bigoplus_{\bm{n}\in Y_n^d} {\cal U}_{\bm{n}} \otimes {\cal V}_{\bm{n}},
\end{equation}
where ${\cal U}_{\bm{n}}$ denotes a representation space of $\mathrm{SU}(d)$ and ${\cal V}_{\bm{n}}$ denotes that of the permutation group $S_n$.  For any $U\in \mathrm{SU}(d)$, $U^{\otimes n}$ can be decomposed into 
\begin{equation}
    U^{\otimes n} = \bigoplus_{\bm{n}\in Y_n^d} \pi_{\bm{n}}(U) \otimes I_{{\cal V}_{\bm{n}}},
\end{equation}
where $\pi_{\bm{n}}$ denotes the irreducible representation of $\mathrm{SU}(d)$ on ${\cal U}_{\bm{n}}$.  Similarly, any unitary representation $V_{s}$ of $s \in S_n$ can be decomposed into 
\begin{equation}
    V_{s} = \bigoplus_{\bm{n}\in Y_n^d} I_{{\cal U}_{\bm{n}}} \otimes \zeta_{\bm{n}}(s), \label{eq:unitary_rep_perm}
\end{equation}
where $\zeta_{\bm{n}}$ denotes the irreducible representation of $S_n$ on ${\cal V}_{\bm{n}}$.  Any state that commutes with $U^{\otimes n}$ for all $U\in \mathrm{SU}(d)$ or commutes with $V_{s}$ for all $s\in S_n$ has the same block diagonal form in this Schur-Weyl basis from Schur's lemma.

Let $\Pi_{\bm{n}}$ be a projection onto the subspace labeled by $\bm{n}$.
Then, we define
\begin{align}
    \sigma_{\bm{n}} &\coloneqq \frac{\Pi_{\bm{n}}}{\dim ({\cal U}_{\bm{n}}\otimes {\cal V}_{\bm{n}})} , \label{eq:uniform_in_a_irreducible}\\
    \sigma_{U,n} &\coloneqq \sum_{\bm{n}\in Y_n^d} \frac{1}{|Y_n^d|} \sigma_{\bm{n}},
    \label{eq:uniform_state}
\end{align}
where $\sigma_{U,n}$ is called the universal symmetric state~\cite{Hayashi2009,Mosonyi2017,Hayashi2017}.  A state that commutes with both $U^{\otimes n}$ for all $U\in \mathrm{SU}(d)$ and $V_{s}$ for all $s\in S_n$ can be written as $\sum_{\bm{n}\in Y_n^d}p_{\bm{n}}\sigma_{\bm{n}}$, where $\{p_{\bm{n}}\}_{\bm{n}\in Y_n^d}$ is a probability distribution.  The state $\sigma_{\bm{n}}$ for any $\bm{n}\in Y_n^d$ commutes with any operator that has a block-diagonal form in the Schur-Weyl basis, and therefore so does $\sigma_{U, n}$.  In particular, the universal symmetric state $\sigma_{U, n}$ commutes with any operator of the form $O^{\otimes n}$, which will be used later. 
For any i.i.d.~state $\rho^{\otimes n}$, the following holds from its permutation symmetry:
\begin{equation}
    \Pi_{\bm{n}} \rho^{\otimes n} \Pi_{\bm{n}} = \rho_{{\cal U}_{\bm{n}}} \otimes \frac{I_{{\cal V}_{\bm{n}}}}{\dim {\cal V}_{\bm{n}}} \leq \frac{\Pi_{\bm{n}}}{\dim {\cal V}_{\bm{n}}}  = \dim {\cal U}_{\bm{n}} \sigma_{\bm{n}} ,\label{eq:direct_sum_decomp}
\end{equation}
where $\rho_{{\cal U}_{\bm{n}}}$ is a subnormalized density operator.
Therefore, we have
\begin{equation}
    \rho^{\otimes n} \leq \sum_{\bm{n}\in Y_n^d} \dim {\cal U}_{\bm{n}} \, \sigma_{\bm{n}} \leq \max_{\bm{n}}(\dim {\cal U}_{\bm{n}}) |Y_n^d| \sigma_{U,n}.
    \label{eq:bound_universal_symmetric}
\end{equation}
It is known that the following upper bounds hold~\cite{Itzykson1966,Hayashi2009, Mosonyi2017}:
\begin{equation}
    \forall \bm{n}\in Y_n^d,\qquad\dim {\cal U}_{\bm{n}} \leq (n+1)^{\frac{d(d-1)}{2}},
\end{equation}
and
\begin{equation}
    |Y_n^d| \leq (n+1)^{d-1}.
\end{equation}
Therefore, the coefficient of $\sigma_{U,n}$ in Eq.~\eqref{eq:bound_universal_symmetric} can be bounded from above by 
\begin{equation}
    \max_{\bm{n}}(\dim {\cal U}_{\bm{n}}) |Y_n^d|\leq (n+1)^{\frac{(d+2)(d-1)}{2}}. \label{eq:quantum_coefficient_bound}
\end{equation}

Let us now consider a string $\bm{x}\in{\cal X}^n$. For ease of discussion, let us label each element of ${\cal X}$ by an integer from $1$ to $k$, where $k=|{\cal X}|$.  
Define $P_{\bm{x}}$ as the type of the string $\bm{x}$, i.e., the probability distribution over ${\cal X}$ satisfying
\begin{equation}
    \forall y\in{\cal X},\qquad P_{\bm{x}}(y) = \frac{\left|\bigl\{i\in\{1,\ldots,n\}:x_i=y\bigr\}\right|}{n}.
\end{equation}
Let ${\cal P}_n$ be the set of types for length-$n$ strings.
It is known that the following holds~\cite{Hayashi2017}:
\begin{equation}
    |{\cal P}_n| \leq (n+1)^{k-1}. \label{eq:number_of_types}
\end{equation}
For each $P\in{\cal P}_n$, let ${\cal T}_P$ be the set of length-$n$ strings with the type $P$, i.e.,
\begin{equation}
    {\cal T}_P \coloneqq \{\bm{x}\in{\cal X}^n:P_{\bm{x}}=P\}.
\end{equation}
For each string $\bm{x}\in{\cal X}^n$, the string $\bm{\chi}(\bm{x})\in{\cal X}^n$ is defined as
\begin{equation}
    \bm{\chi}(\bm{x}) \coloneqq (\underbrace{1, \ldots, 1}_{m_1}, \underbrace{2, \ldots, 2}_{m_2}, \ldots, \underbrace{k, \ldots, k}_{m_k}),
\end{equation}
where $m_i\coloneqq n P_{\bm{x}}(i)$ for $i=1,\ldots,k$.  Then, there exists a permutation $s_{\bm{x}}\in S_n$ such that $\bm{\chi}(\bm{x})=s_{\bm{x}}(\bm{x})$. In general, $s_{\bm{x}}$ is not unique.

Let us now consider an i.i.d.~c-q state $\rho_{XQ}^{\otimes n}$ with $\rho_{XQ}$ defined in Eq.~\eqref{eq:source_state} and quantum states $\{\rho_{Q^n}^{\bm{x}}\}_{\bm{x}\in{\cal X}^n}$ therein.  Then, for each $\bm{x}$, we have 
\begin{equation}
    \rho_{Q^n}^{\bm{\chi}(\bm{x})} = (\rho_Q^{1})^{\otimes m_1}\otimes\cdots\otimes(\rho_Q^{k})^{\otimes m_k},
\end{equation}
where $m_i = n P_{\bm{x}}(i)$ for $i=1,\ldots,k$.  By applying Eqs.~\eqref{eq:bound_universal_symmetric} and~\eqref{eq:quantum_coefficient_bound} recursively~\cite{Hayashi2009,Matsuura2024}, we have 
\begin{equation}
    \rho_{Q^n}^{\bm{\chi}(\bm{x})} \leq \prod_{i=1}^k (m_i+1)^{\frac{(d+2)(d-1)}{2}} \sigma_{U,m_1}\otimes\cdots\otimes\sigma_{U,m_k} \leq (n+1)^{\frac{(d+2)(d-1)}{2}|{\cal X}|}\sigma_{U,m_1}\otimes\cdots\otimes\sigma_{U,m_k}.
\end{equation}
Furthermore, the right-hand side commutes with the left-hand side since the universal symmetric state $\sigma_{U,m_i}$ commutes with $O^{\otimes m_i}$ for any operator $O$.
We define $\sigma_{\bm{x}}$ for $\bm{x}\in{\cal X}^n$ as 
\begin{equation}
    \sigma_{\bm{x}} \coloneqq V_{s_{\bm{x}}}^{-1} \sigma_{U,m_1}\otimes\cdots\otimes\sigma_{U,m_k}V_{s_{\bm{x}}},
\end{equation}
where $V_{s}$ is the unitary representation of $s\in S_n$ as mentioned earlier.
From the relation $\bm{\chi}(\bm{x}) = s_{\bm{x}}(\bm{x})$, the following holds for any $\bm{x}\in{\cal X}^n$
\begin{equation}
    \rho_{Q^n}^{\bm{x}} = V_{s_{\bm{x}}}^{-1} \rho_{Q^n}^{\bm{\chi}(\bm{x})} V_{s_{\bm{x}}} \leq (n+1)^{\frac{(d+2)(d-1)}{2}|{\cal X}|} \sigma_{\bm{x}}, \label{eq:sequence_dependent_bound}
\end{equation}
and $\rho_{Q^n}^{\bm{x}}$ commutes with $\sigma_{\bm{x}}$.
For any type $P\in{\cal P}_n$, let $\sigma_{U,P}$ be defined as 
\begin{equation}
    \sigma_{U,P} \coloneqq \frac{1}{|{\cal T}_P|}\sum_{\bm{x}\in{\cal T}_P}  \sigma_{\bm{x}}. \label{eq:def_universal_type_symmetric}
\end{equation}
Then, the state $\sigma_{U,P}$ commutes with both $U^{\otimes n}$ for any $U\in\mathrm{SU}(d)$ and $V_{\sigma}$ for any $\sigma\in S_n$, and thus have the form $\sum_{\bm{n}\in Y_n^d}p_{\bm{n}}\sigma_{\bm{n}}$.  In particular, it commutes with $\sigma_{\bm{x}}$ for any $\bm{x}\in{\cal X}^n$.

Next, we introduce information-theoretic quantities that characterize this task.  For $\alpha\in [0, 1) \cup(1, \infty)$, let $D_{\alpha}(\rho\|\sigma)$ be the $\alpha$-R\'enyi divergence defined as 
\begin{equation}
    D_{\alpha}(\rho\|\sigma) \coloneqq \begin{cases}\frac{1}{\alpha-1}\log\tr[\rho^{\alpha}\sigma^{1-\alpha}] & 0\leq \alpha < 1 \text{ or } \mathrm{supp}(\rho)\subseteq\mathrm{supp}(\sigma), \\
        \infty & \text{Otherwise}.
    \end{cases} \label{eq:alpha_Renyi_divergence}
\end{equation}
As $\alpha\rightarrow 1$, $D_{\alpha}(\rho\|\sigma)$ reduces to the quantum relative entropy $D(\rho\|\sigma)$ given by 
\begin{equation}
    D(\rho\|\sigma)\coloneqq\begin{cases} \tr[\rho\log\rho - \rho\log\sigma] & \mathrm{supp}(\rho)\subseteq\mathrm{supp}(\sigma), \\
        \infty & \text{Otherwise}.
    \end{cases}
\end{equation}
For later use, we also introduce the relative entropy variance $V(\rho\|\sigma)$ defined as 
\begin{equation}
    V(\rho\|\sigma) \coloneqq \tr[\rho(\log\rho - \log\sigma)^2] - \left(\tr[\rho(\log\rho-\log\sigma)]\right)^2. \label{eq:relative_entropy_variance}
\end{equation}

The conditional $\alpha$-R\'enyi entropy $H_{\alpha}^{\uparrow}(A|B)$ is defined for any $\alpha\in [0, \infty)$ as 
\begin{equation}
    H_{\alpha}^{\uparrow}(A|B)_{\rho} \coloneqq \max_{\sigma_B\in{\cal D}({\cal H}_B)} -D_{\alpha}(\rho_{AB}\|I_A\otimes \sigma_B). \label{eq:conditional_alpha_Renyi}
\end{equation}
The maximum in Eq.~\eqref{eq:conditional_alpha_Renyi} is attained by the following unique maximizer $\sigma_{\alpha}^*$ for any $\alpha\in(0, \infty)$~\cite{Sharma2013}:
\begin{equation}
    \sigma_{\alpha}^* = \frac{\left(\tr_A[\rho_{AB}^{\alpha}]\right)^{\frac{1}{\alpha}}}{\tr\!\left[\left(\tr_A[\rho_{AB}^{\alpha}]\right)^{\frac{1}{\alpha}}\right]}, \label{eq:maximizer_unique}
\end{equation}
which then leads to the following analytical expression for $\alpha\in(0,1)\cup(1,\infty)$~\cite{Koenig2009, Sharma2013, Mosonyi2017, Cheng2022}:
\begin{equation}
    H_{\alpha}^{\uparrow}(A|B)_{\rho} = -\frac{\alpha}{\alpha-1} \log \tr\Bigl[\bigl(\tr_A[\rho_{AB}^{\alpha}]\bigr)^{\frac{1}{\alpha}}\Bigr]. \label{eq:Sibson_identity}
\end{equation}
As $\alpha\rightarrow 1$, $H_{\alpha}^{\uparrow}(A|B)$ reduces to the von Neumann conditional entropy $H(A|B)_{\rho}$ given by 
\begin{equation}
H(A|B)_{\rho}\coloneqq -D(\rho_{AB}\|I_A\otimes\rho_B). \label{eq:vN_conditional}
\end{equation} 
When $\alpha\in[0,1]$, the conditional $\alpha$-R\'enyi entropy $H_{\alpha}^{\uparrow}(A|B)_{\rho}$ is concave for $\rho$, which follows directly from the joint convexity of the $\alpha$-R\'enyi divergence Eq.~\eqref{eq:alpha_Renyi_divergence} for $\alpha\in[0,1]$~\cite{Mosonyi2011}.  Furthermore, the function $\alpha\mapsto H_{\alpha}^{\uparrow}(A|B)_{\rho}$ is continuous and monotonically non-increasing on $\alpha\in[0,1]$ for any $\rho_{AB}$~\cite{Cheng2021}.

For a subnormalized positive operator $\tilde{\rho}_{AB}$, we define $H^{\uparrow, \leq}_{\alpha}(A|B)_{\tilde{\rho}}$ as
\begin{equation}
    H^{\uparrow, \leq}_{\alpha}(A|B)_{\tilde{\rho}} \coloneqq H^{\uparrow}_{\alpha}(A'|B')_{\tilde{\rho}\, \oplus \,(1 - \tr[\tilde{\rho}])\bm{1}\otimes \bm{1}},
    \label{eq:subnormalized_Renyi}
\end{equation}
where ${\cal H}_{A'(B')}={\cal H}_{A(B)}\oplus\mathbb{C}$, and $\bm{1}$ is the unique state on the trivial Hilbert space $\mathbb{C}$. From Eq.~\eqref{eq:Sibson_identity}, we have 
\begin{align}
    H^{\uparrow, \leq}_{\alpha}(A|B)_{\tilde{\rho}} &= -\frac{\alpha}{\alpha-1}\log \tr\left[\bigl(\tr_{{\cal H}_A\oplus\mathbb{C}}[\tilde{\rho}_{AB}^{\alpha}\oplus (1 - \tr[\tilde{\rho}_{AB}])^{\alpha}\bm{1}\otimes \bm{1}]\bigr)^{\frac{1}{\alpha}}\right] \\
    &= -\frac{\alpha}{\alpha-1}\log \tr\left[\bigl(\tr_{{\cal H}_A}[\tilde{\rho}_{AB}^{\alpha}] \oplus (1 - \tr[\tilde{\rho}_{AB}])^{\alpha}\bm{1}\bigr)^{\frac{1}{\alpha}}\right] \\
    &= -\frac{\alpha}{\alpha-1}\log \tr\left[(\tr_{{\cal H}_A}[\tilde{\rho}_{AB}^{\alpha}])^{\frac{1}{\alpha}}\oplus (1 - \tr[\tilde{\rho}_{AB}]) \bm{1} \right] \\
    &= -\frac{\alpha}{\alpha-1}\log \left( \tr\left[(\tr_{{\cal H}_A}[\tilde{\rho}_{AB}^{\alpha}])^{\frac{1}{\alpha}}\right] + (1 - \tr[\tilde{\rho}_{AB}])\right). \label{eq:expression_subnormalized}
\end{align}
Since $\tilde{\rho}_{AB}\mapsto \tilde{\rho}_{AB}\oplus (1 - \tr[\tilde{\rho}_{AB}])\bm{1}\otimes\bm{1}$ is affine for $\tilde{\rho}_{AB}$, $H^{\uparrow, \leq}_{\alpha}(A|B)_{\tilde{\rho}}$ is concave for $\tilde{\rho}$ from the concavity of $H_{\alpha}^{\uparrow}(A|B)_{\rho}$ for $\rho\in{\cal D}({\cal H})$.
Using this definition, we can extend the definition of the von Neumann conditional entropy in Eq.~\eqref{eq:vN_conditional} to a subnormalized positive operator $\tilde{\rho}_{AB}$ as 
\begin{align}
    H(A|B)_{\tilde{\rho}} &\coloneqq \lim_{\alpha \to 1} H_{\alpha}^{\uparrow,\leq}(A|B)_{\tilde{\rho}} \\
    &= -D(\tilde{\rho}_{AB}\|I_A\otimes \tilde{\rho}_B). \label{eq:subnormalized_cond_ent}
\end{align}

For a string $\bm{x}\in{\cal X}^n$, let $H(\bm{x})$ be the empirical entropy defined as 
\begin{equation}
    H(\bm{x}) \coloneqq -\sum_{y\in{\cal X}}P_{\bm{x}}(y)\log P_{\bm{x}}(y) = H(X)_{P_{\bm{x}}}. \label{eq:def_empirical_entropy}
\end{equation}
The empirical entropy thus depends only on the type of the string $\bm{x}$.
For any type $P\in{\cal P}_n$, the following is known to hold~\cite{Cover2006}:
\begin{equation}
    |{\cal T}_P|\leq 2^{n H(X)_P}. \label{eq:type_cardinality_bound}
\end{equation}
Furthermore, for any i.i.d.~probability distribution $p^n$ over ${\cal X}^n$, the empirical entropy $H(\bm{x})$ satisfies~\cite{Cover2006} 
\begin{equation}
    \log p^n(\bm{x}) \leq -n H(\bm{x}). \label{eq:empirical_distribution_most_probable}
\end{equation}

\subsection{Construction of a random encoder and a universal decoder} \label{sec:construction_universal}
In this section, we explicitly construct a random encoder and a universal decoder for classical source compression with quantum side information described in Sec.~\ref{sec:problem_setup}.
For an encoder ${\cal F}_n$ in Sec.~\ref{sec:problem_setup}, we choose a 2-universal family of hash functions, which is defined as follows.

\begin{definition}[2-universal family of hash functions]\label{def:2-universal_hash}
    For finite sets ${\cal A}$ and ${\cal B}$, a family of hash functions ${\cal H} = \{h:{\cal A}\to{\cal B}\}$ is called 2-universal if it satisfies
    \begin{equation}
        \forall x,y\in{\cal A}, x\neq y,\quad |\{h\in{\cal H}:h(x)=h(y)\}| \leq \frac{|{\cal H}|}{|{\cal B}|}. \label{eq:universal_2}
    \end{equation} 
If we define a function $\bm{1}(E)$ of an expression $E$ as 
\begin{equation}
    \bm{1}(E)=\begin{cases}
        1 & \text{ if } E \text{ is satisfied}, \\
        0 & \text{ Otherwise}, 
    \end{cases} \label{eq:def_indicator}
\end{equation}
then, the condition~\eqref{eq:universal_2} can alternatively be written as 
\begin{equation}
    \forall x,y\in{\cal A}, x\neq y, \quad \sum_{h\in{\cal H}} \bm{1}(h(x)=h(y)) \leq \frac{|{\cal H}|}{|{\cal B}|}. \label{eq:alternative_universal_2}
\end{equation}
\end{definition}

For a decoder $Y^n$ in Sec.~\ref{sec:problem_setup}, we choose the following universal likelihood decoder:
\begin{equation}
    \forall \bm{x}\in h^{-1}(b),\qquad Y^n_{h, b}(\bm{x}) \coloneqq \frac{2^{-n H(\bm{x})}\sigma_{\bm{x}}}{\sum_{\bm{y}\in h^{-1}(b)}2^{-n H(\bm{y})}\sigma_{\bm{y}}} = \frac{2^{-n H(\bm{x})}\sigma_{\bm{x}}}{\sum_{\bm{y}\in {\cal X}^n}\bm{1}(h(\bm{y})=b) \, 2^{-n H(\bm{y})}\sigma_{\bm{y}}}, \label{eq:fully_universal_decoder}
\end{equation}
where the division $\frac{A}{B}$ of two positive operators $A$ and $B>0$ is defined as~\cite{Beigi2023} 
\begin{equation}
    \frac{A}{B} \coloneqq \int_{0}^{\infty}d\lambda\, (B+\lambda)^{-1}A(B+\lambda)^{-1}, \label{eq:operator_division}
\end{equation}
and satisfies the followings~\cite{Beigi2023}:
\begin{align}
    &\text{For any } A\geq 0 \text{ and } B > 0,  & \frac{A}{B} &\geq 0, \\
    &\text{For any } A\geq 0 \text{ and } B > 0 \text{ with } [A, B]=0, & \frac{A}{B} &= A B^{-1} = B^{-1}A, \\
    &\text{For any } A, B \geq 0 \text{ and } C > 0, &  \frac{A}{C} + \frac{B}{C} &= \frac{A+B}{C}, \label{eq:common_denominator}\\
    &\text{For any } A\geq 0 \text{ and } B > 0, &  \frac{A}{A + B} &\leq \frac{A}{B}. \label{eq:larger_denominator}
\end{align}
A similar decoding strategy has been studied in classical information theory~\cite{Merhave2017}.
In the partially universal setup in which the encoder and the decoder know the probability distribution $p(\cdot)$ of the c-q state $\rho_{XB}$, the encoder uses the same random encoding, but the decoder uses the following $\tilde{Y}^n$ instead of $Y^n$ introduced above:
\begin{equation}
    \forall \bm{x}\in h^{-1}(b),\qquad\tilde{Y}^n_{h, b}(\bm{x})\coloneqq  \frac{p^n(\bm{x})\sigma_{\bm{x}}}{\sum_{\bm{y}\in h^{-1}(b)}p^n(\bm{y})\sigma_{\bm{y}}} = \frac{p^n(\bm{x})\sigma_{\bm{x}}}{\sum_{\bm{y}\in {\cal X}^n}\bm{1}(h(\bm{y})=b) \, p^n(\bm{y})\sigma_{\bm{y}}} . \label{eq:partially_universal_decoder}
\end{equation}
For these encoding and decoding strategies, we prove the following.

\begin{theorem}\label{thm:cq_slepian_wolf}
    Using a 2-universal family ${\cal H}_n\coloneqq \{h:{\cal X}^n\to {\cal B}_n\}$ of hash functions as a random encoder and the POVM $Y^n$ in Eq.~\eqref{eq:fully_universal_decoder} as a (universal) decoder, the universal source compression with quantum side information described in Sec.~\ref{sec:problem_setup} is achievable with the non-asymptotic error exponent $-\frac{1}{n}\log P_{\rm err}({\cal H}_n, Y^n)_{\rho}$ bounded from below by 
    \begin{equation}
        -\frac{1}{n}\log P_{\rm err}({\cal H}_n,Y^n)_{\rho} \geq \max_{\alpha\in[0,1]} \alpha \left(\frac{\log|{\cal B}_n|}{n} - H_{1-\alpha}^{\uparrow}(X|Q)_{\rho} - \frac{[|{\cal X}|(d^2+d)-2]\log (n+1)}{2n}\right). \label{eq:error_exponent_universal}
    \end{equation}
    Thus, any rate $R>H(X|Q)_{\rho}$ is achievable with asymptotically vanishing errors.
    
    For the partially universal source compression with quantum side information, the following non-asymptotic error exponent $-\frac{1}{n}\log P'_{\rm err}({\cal H}_n, \tilde{Y}^n)_{\rho}$ is achievable using the 2-universal family ${\cal H}_n$ of hash functions as a random encoder and the POVM $\tilde{Y}^n$ in Eq.~\eqref{eq:partially_universal_decoder} as a (universal) decoder:
    \begin{equation}
        -\frac{1}{n}\log P'_{\rm err}({\cal H}_n,\tilde{Y}^n)_{\rho} \geq \max_{\alpha\in[0,1]} \alpha \left(\frac{\log|{\cal B}_n|}{n} - H_{1-\alpha}^{\uparrow}(X|Q)_{\rho} - \frac{|{\cal X}|(d+2)(d-1)\log (n+1)}{2n}\right). \label{eq:error_exponent_partially_universal}
    \end{equation}
\end{theorem}

To prove the above theorem, we first show the following lemma.
\begin{lemma}\label{lem:decoder_reformulation}
    Let $Y^n_{h, b}$ and $\tilde{Y}^n_{h,b}$ be defined as in Eqs.~\eqref{eq:fully_universal_decoder} and \eqref{eq:partially_universal_decoder}, respectively. Let ${\cal H}_n\coloneqq \{h:{\cal X}^n\to{\cal B}_n\}$ be 2-universal hash functions. Then, for any $\alpha\in[0,1]$, we have 
    \begin{equation}
        \mathbb{E}_{h\sim |{\cal H}_n|^{-1}}\!\left[I_{Q^n} - Y^n_{h,h(\bm{x})}(\bm{x}) \right] \leq |{\cal B}_n|^{-\alpha}\left(2^{-nH(\bm{x})}\sigma_{\bm{x}}\right)^{-\alpha} \left(\sum_{P\in{\cal P}_n}|{\cal T}_{P}| 2^{-n H(X)_P}\sigma_{U,P}\right)^{\alpha}, \label{eq:operator_ineq_for_fully_universal}
    \end{equation}
    and 
    \begin{equation}
        \mathbb{E}_{h\sim |{\cal H}_n|^{-1}}\!\left[I_{Q^n} - \tilde{Y}^n_{h,h(\bm{x})}(\bm{x}) \right] \leq |{\cal B}_n|^{-\alpha}\left(p^n(\bm{x})\sigma_{\bm{x}}\right)^{-\alpha} \left(\sum_{P\in{\cal P}_n}|{\cal T}_{P}| \prod_{x\in{\cal X}}p(x)^{nP(x)}\sigma_{U,P}\right)^{\alpha},\label{eq:operator_ineq_for_partially_universal}
    \end{equation}
    where $p(x)$ denotes the probability distribution of the classical source.
\end{lemma}
\begin{proof}
    From the fact that $T \leq T^\alpha$ holds for $0\leq T \leq I$, we have
    \begin{equation}
        I_{Q^n} - Y^n_{h,h(\bm{x})}(\bm{x}) \leq \left(I_{Q^n} - Y^n_{h,h(\bm{x})}(\bm{x})\right)^{\alpha}.
    \end{equation}
    It is known from L\"owner-Heinz theorem that the function $t\mapsto t^\alpha$ of $t\in[0, \infty)$ is operator monotone for $\alpha\in[0, 1]$. Combining this with the following operator inequality, which holds from Eqs.~\eqref{eq:common_denominator} and \eqref{eq:larger_denominator},
    \begin{equation}
        I-\frac{B}{A+B}=\frac{A}{A+B}\leq \frac{A}{B},
    \end{equation}
    we have 
    \begin{align}
        \left(I_{Q^n} - Y^n_{h,h(\bm{x})}(\bm{x})\right)^{\alpha} &= \left(I_{Q^n} - \frac{2^{-nH(\bm{x})}\sigma_{\bm{x}}}{\sum_{\bm{y}\in{\cal X}^n}\bm{1}(h(\bm{y})=h(\bm{x}))2^{-nH(\bm{y})}\sigma_{\bm{y}}}\right)^{\alpha} \\
        &\leq \left(\frac{\sum_{\bm{y}\in{\cal X}^n\setminus\{\bm{x}\}}\bm{1}(h(\bm{y})=h(\bm{x}))2^{-nH(\bm{y})}\sigma_{\bm{y}}}{2^{-nH(\bm{x})}\sigma_{\bm{x}}}\right)^{\alpha}. \label{eq:monotone_ineq}
    \end{align}
    Since the function $t\mapsto t^\alpha$ of $t\in[0, \infty)$ is also operator concave for $\alpha\in[0, 1]$ from L\"owner-Heinz theorem, we have 
    \begin{align}
        \mathbb{E}_{h\sim |{\cal H}_n|^{-1}}\!\left[I_{Q^n} - Y^n_{h,h(\bm{x})}(\bm{x}) \right] &\leq \mathbb{E}_{h\sim |{\cal H}_n|^{-1}}\!\left[\left(\frac{\sum_{\bm{y}\in{\cal X}^n\setminus\{\bm{x}\}}\bm{1}(h(\bm{y})=h(\bm{x}))2^{-nH(\bm{y})}\sigma_{\bm{y}}}{2^{-nH(\bm{x})}\sigma_{\bm{x}}}\right)^{\alpha}\right] \\
        &\leq \left(\frac{\sum_{\bm{y}\in{\cal X}^n\setminus\{\bm{x}\}}\mathbb{E}_{h\sim |{\cal H}_n|^{-1}}[\bm{1}(h(\bm{y})=h(\bm{x}))]2^{-nH(\bm{y})}\sigma_{\bm{y}}}{2^{-nH(\bm{x})}\sigma_{\bm{x}}}\right)^{\alpha} \\
        &\leq \left(\frac{\sum_{\bm{y}\in{\cal X}^n\setminus\{\bm{x}\}}|{\cal B}_n|^{-1}2^{-nH(\bm{y})}\sigma_{\bm{y}}}{2^{-nH(\bm{x})}\sigma_{\bm{x}}}\right)^{\alpha}\\
        &\leq \left(\frac{\sum_{\bm{y}\in{\cal X}^n}|{\cal B}_n|^{-1}2^{-nH(\bm{y})}\sigma_{\bm{y}}}{2^{-nH(\bm{x})}\sigma_{\bm{x}}}\right)^{\alpha}, \label{eq:before_grouped}
    \end{align}
    where we used Eq.~\eqref{eq:monotone_ineq} in the first inequality, operator concavity of $t\mapsto t^\alpha$ in the second inequality, Eq.~\eqref{eq:alternative_universal_2} in the third inequality, and the operator monotonicity of $t\mapsto t^\alpha$ in the last inequality. Now, we group strings with the same type in the numerator of Eq.~\eqref{eq:before_grouped} to have 
    \begin{align}
        \left(\frac{\sum_{\bm{y}\in{\cal X}^n}|{\cal B}_n|^{-1}2^{-nH(\bm{y})}\sigma_{\bm{y}}}{2^{-nH(\bm{x})}\sigma_{\bm{x}}}\right)^{\alpha} &= |{\cal B}_n|^{-\alpha}\left(\frac{\sum_{P\in{\cal P}_n}\sum_{\bm{y}\in{\cal T}_P}2^{-nH(X)_P}\sigma_{\bm{y}}}{2^{-nH(\bm{x})}\sigma_{\bm{x}}}\right)^{\alpha} \\
        &= |{\cal B}_n|^{-\alpha}\left(\frac{\sum_{P\in{\cal P}_n}|{\cal T}_P|2^{-nH(X)_P}\sigma_{U,P}}{2^{-nH(\bm{x})}\sigma_{\bm{x}}}\right)^{\alpha}, \label{eq:after_grouped}
    \end{align}
    where $\sigma_{U,P}$ is defined in Eq.~\eqref{eq:def_universal_type_symmetric}. Thus, Eq.~\eqref{eq:operator_ineq_for_fully_universal} is proved by combining Eqs.~\eqref{eq:before_grouped} and \eqref{eq:after_grouped} with the fact that $\sigma_{U,P}$ for any $P\in{\cal P}_n$ commutes with $\sigma_{\bm{x}}$ for any $\bm{x}\in{\cal X}^n$.

    The same derivation holds for the partially universal decoder $\tilde{Y}^n_{h,h(\bm{x})}$ by replacing $2^{-nH(\bm{x})}$ in the above with $p^n(\bm{x})$. Then, for any string $\bm{y}\in{\cal T}_P$, we have 
    \begin{equation}
        p^n(\bm{y}) = \prod_{y\in{\cal X}} p(y)^{nP(y)},
    \end{equation}
    which thus proves Eq.~\eqref{eq:operator_ineq_for_partially_universal}.
\end{proof}
Using the above lemma, we show Theorem~\ref{thm:cq_slepian_wolf} as follows.
\begin{proof}[Proof of Theorem~\ref{thm:cq_slepian_wolf}]
    By substituting ${\cal H}_n$ and $Y^n$ given in the theorem to the definition of $P_{\rm err}$ in Eq.~\eqref{eq:pre_def_of_error}, we have 
    \begin{align}
        P_{\rm err}({\cal H}_n, Y^n)_{\rho} &= \mathbb{E}_{h\sim |{\cal H}_n|^{-1}}\left[\sum_{\bm{x}\in{\cal X}^n}p^n(\bm{x})\tr\left[\rho_{Q^n}^{\bm{x}}\left(I_{Q^n} - Y^n_{h,h(\bm{x})}(\bm{x})\right)\right]\right] \\
        &= \tr\left[\sum_{\bm{x}\in{\cal X}^n}p^n(\bm{x})\rho_{Q^n}^{\bm{x}}\;\mathbb{E}_{h\sim |{\cal H}_n|^{-1}}\!\left[I_{Q^n} - Y^n_{h,h(\bm{x})}(\bm{x})\right]\right]. \label{eq:reformulation_first}
    \end{align}
    Applying Eq.~\eqref{eq:operator_ineq_for_fully_universal} in Lemma~\ref{lem:decoder_reformulation}, we have 
    \begin{equation}
        P_{\rm err}({\cal H}_n, Y^n)_{\rho} \leq  |{\cal B}_n|^{-\alpha} \mathrm{Tr}\left[\sum_{\bm{x}\in{\cal X}^n}  p^n(\bm{x})\rho_{Q^n}^{\bm{x}}\left(2^{-n H(\bm{x})}\sigma_{\bm{x}}\right)^{-\alpha}\left(\sum_{P\in {\cal P}_n} |{\cal T}_P| 2^{-n H(X)_P}\sigma_{U, P}\right)^{\alpha}\right]. \label{eq:just_before_final_eq}
    \end{equation}
    Combining Eqs.~\eqref{eq:sequence_dependent_bound}, \eqref{eq:def_empirical_entropy}, and \eqref{eq:empirical_distribution_most_probable} with the fact that $\rho_{Q^n}^{\bm{x}}$ commutes with $\sigma_{\bm{x}}$ and that the function $t\mapsto -t^{-\alpha}$ of $t\in(0,\infty)$ is operator monotone for $\alpha\in[0,1]$, we have, for any $\bm{x}\in{\cal X}^n$,
    \begin{equation}
        p^n(\bm{x})\rho_{Q^n}^{\bm{x}}\left(2^{-n H(\bm{x})}\sigma_{\bm{x}}\right)^{-\alpha} \leq (p^n(\bm{x}))^{1-\alpha} \rho_{Q^n}^{\bm{x}}\sigma_{\bm{x}}^{-\alpha} \leq (n+1)^{\alpha |{\cal X}| \frac{(d+2)(d-1)}{2}}\left(p^n(\bm{x})\rho_{Q^n}^{\bm{x}}\right)^{1-\alpha}. \label{eq:1-alpha_bound}
    \end{equation} 
    Substituting this into Eq.~\eqref{eq:just_before_final_eq}, we have 
    \begin{align}
        P_{\rm err}({\cal H}_n, Y^n)_{\rho} &\leq \left(|{\cal B}_n|^{-1} (n+1)^{|{\cal X}| \frac{(d+2)(d-1)}{2}}\right)^{\alpha}\mathrm{Tr}\!\left[\sum_{\bm{x}\in{\cal X}^n}  \left(p^n(\bm{x})\rho_{Q^n}^{\bm{x}}\right)^{1-\alpha}\left(\sum_{P\in {\cal P}_n} |{\cal T}_P| 2^{-n H(X)_P}\sigma_{U, P}\right)^{\alpha}\right] \\
        \begin{split}
        &\leq \left(|{\cal B}_n|^{-1} (n+1)^{|{\cal X}| \frac{(d+2)(d-1)}{2}}\right)^{\alpha} \left(\mathrm{Tr}\left[\sum_{P\in {\cal P}_n} |{\cal T}_P| 2^{-n H(X)_P}\sigma_{U, P}\right]\right)^{\alpha} \\
        &\hspace{5cm} \max_{\tau\in{\cal D}({\cal H}_Q^{\otimes n})}\mathrm{Tr}\left[ \left(\sum_{x\in{\cal X}} \bigl(p(x)\rho_{Q}^{x}\bigr)^{1-\alpha}\right)^{\otimes n}\tau^{\alpha}\right],
        \end{split} \label{eq:branching}
    \end{align}
    where we used $\tr[X Y^{\alpha}] \leq (\tr[Y])^{\alpha}\max_{Z\in{\cal D}({\cal H})}\tr[X Z^{\alpha}]$ for any positive operators $X$ and $Y$.
    From Eq.~\eqref{eq:type_cardinality_bound}, we have 
    \begin{equation}
        \mathrm{Tr}\left[\sum_{P\in {\cal P}_n} |{\cal T}_P| 2^{-n H(X)_P}\sigma_{U, P}\right] \leq \sum_{P\in{\cal P}_n} 1 = |{\cal P}_n|. \label{eq:num_of_type_bound}
    \end{equation}
    Furthermore, we use Lemma 2 in Ref.~\cite{Hayashi2009} stating that the following holds for any positive operator $X$ and $\alpha\in[0,1]$:
    \begin{equation}
        \max_{\tau\in{\cal D}({\cal H})} \mathrm{Tr}[X\tau^\alpha] = \left(\mathrm{Tr}X^{\frac{1}{1-\alpha}}\right)^{1-\alpha}. \label{eq:Hayashi_lemma}
    \end{equation}
    Combining Eqs.~\eqref{eq:branching}, \eqref{eq:num_of_type_bound} and \eqref{eq:Hayashi_lemma}, we have 
    \begin{align}
        P_{\rm err}({\cal H}_n, Y^n)_{\rho} \leq \left(|{\cal B}_n|^{-1} (n+1)^{|{\cal X}| \frac{(d+2)(d-1)}{2}}\right)^{\alpha} |{\cal P}_n|^{\alpha} \left(\mathrm{Tr}\left[\left(\sum_{x\in{\cal X}} \bigl(p(x)\rho_{Q}^{x}\bigr)^{1-\alpha}\right)^{\frac{1}{1-\alpha}}\right]\right)^{n(1-\alpha)}. \label{eq:Hayashi_lemma_applied}
    \end{align}
    Using Eqs.~\eqref{eq:number_of_types} and \eqref{eq:Sibson_identity}, we therefore have, for any $\alpha\in[0,1]$,
    \begin{equation}
        -\frac{1}{n}\log P_{\rm err}({\cal H}_n, Y^n)_{\rho} \geq \alpha\left(\frac{\log|{\cal B}_n|}{n} - H_{1-\alpha}^{\uparrow}(X|Q)_{\rho} - \frac{\left[|{\cal X}|(d^2+d)-2\right]\log (n+1)}{2n} \right), \label{eq:proof_error_exponent}
    \end{equation}
    which proves Eq.~\eqref{eq:error_exponent_universal}. 
    
    Since $\alpha\mapsto H^{\uparrow}_{\alpha}(X|Q)_{\rho}$ is continuous and monotonically decreasing on $\alpha\in[0,1]$, we find that for any $R>H(X|Q)_{\rho}$, there exists a sufficiently small $\alpha$ such that the right-hand side of Eq.~\eqref{eq:proof_error_exponent} is positive as $n\to\infty$.  Thus, any rate $R>H(X|Q)_{\rho}$ is achievable with asymptotically vanishing errors.
    
    For the partially universal setup, we have from Eq.~\eqref{eq:operator_ineq_for_partially_universal} in Lemma~\ref{lem:decoder_reformulation} that 
    \begin{align}
        P_{\rm err}({\cal H}_n, \tilde{Y}^n)_{\rho} &\leq  |{\cal B}_n|^{-\alpha} \mathrm{Tr}\left[\sum_{\bm{x}\in{\cal X}^n}  p^n(\bm{x})\rho_{Q^n}^{\bm{x}}\left(p^n(\bm{x})\sigma_{\bm{x}}\right)^{-\alpha} \left(\sum_{P\in{\cal P}_n}|{\cal T}_{P}| \prod_{x\in{\cal X}}p(x)^{nP(x)}\sigma_{U,P}\right)^{\alpha}\right] \\
        \begin{split}
        &\leq \left(|{\cal B}_n|^{-1} (n+1)^{|{\cal X}| \frac{(d+2)(d-1)}{2}}\right)^{\alpha} \left(\mathrm{Tr}\left[\sum_{P\in {\cal P}_n} |{\cal T}_P| \prod_{x\in{\cal X}}p(x)^{nP(x)}\sigma_{U, P}\right]\right)^{\alpha} \\
        &\hspace{5cm} \max_{\tau\in{\cal D}({\cal H}_Q^{\otimes n})}\mathrm{Tr}\left[ \left(\sum_{x\in{\cal X}} \bigl(p(x)\rho_{Q}^{x}\bigr)^{1-\alpha}\right)^{\otimes n}\tau^{\alpha}\right],
        \end{split}
    \end{align}
    where we followed the same derivation as the fully universal case.
    Here, we have 
    \begin{equation}
        \mathrm{Tr}\left[\sum_{P\in{\cal P}_n}|{\cal T}_P|\prod_{y\in{\cal X}}p(y)^{nP(y)}\sigma_{U, P}\right] = \sum_{P\in{\cal P}_n}\sum_{\bm{y}\in{\cal T}_P}p^n(\bm{y}) = \sum_{\bm{y}\in{\cal X}^n}p^n(\bm{y})=1. \label{eq:normalized_partially_universal}
    \end{equation}
    Thus, we obtain Eq.~\eqref{eq:error_exponent_partially_universal} for this case.
\end{proof}

\begin{remark}\label{rem:exponent}
    References~\cite{Renes2023} and~\cite{Renes2024} showed with a rather indirect argument that the achievable error exponent (i.e., $\lim_{n\to\infty}-\frac{1}{n}\log P_{\rm err}({\cal H}_n, Y^n)_{\rho}$) of classical source compression with quantum side information when one knows the c-q state $\rho_{XB}$ is given by 
    \begin{equation}
        \max_{\alpha\in[0,1]} \alpha (R- H_{\frac{1}{1+\alpha}}^{\uparrow}(X|Q)_{\rho}). \label{eq:random_coding_exponent}
    \end{equation}
    Since the function $\alpha\mapsto H_{\alpha}^{\uparrow}(X|Q)_{\rho}$ is monotonically decreasing for $\alpha\in[0,1]$, $H_{\frac{1}{1+\alpha}}^{\uparrow}(X|Q)_{\rho}$ is smaller than $H_{1-\alpha}^{\uparrow}(X|Q)_{\rho}$, which leads to a better error exponent.  
    Whether this error exponent is achievable even with a universal decoder is an open problem.
    An upper bound on the error exponent when the encoder and the decoder know the state $\rho_{XQ}$ (and thus also an upper bound on the error exponent for the universal setup as well) is derived in Ref.~\cite{Cheng2021} as
    \begin{equation}
        \max_{\alpha\geq 0} \alpha (R- H_{\frac{1}{1+\alpha}}^{\uparrow}(X|Q)_{\rho}).
        \label{eq:sphere_packing_bound}
    \end{equation} 
    Note that even though the obtained error exponent in Theorem~\ref{thm:cq_slepian_wolf} may not be optimal, we later focus on minimizing the rate $R$ with a fixed error probability when $n\gg 1$.  In this regime, the choice $\alpha\sim n^{-1/2}$ may be optimal, and the difference between Eq.~\eqref{eq:sphere_packing_bound} and ours is small.

    The single-shot error bound on classical source compression with quantum side information is also given in Ref.~\cite{Renes2012} in terms of the smooth conditional max entropy, and its error exponent can be obtained using the asymptotic expansion of the smooth max entropy~\cite{Tomamichel2013}. 
\end{remark}
\begin{remark}
    The results in this section can easily be generalized to the case where one uses more general hash function families such as the $\delta$-almost 2-universal hash function family~\cite{Tsurumaru2013}.  In the $\delta$-almost 2-universal hash function family, the factor $\delta\geq 1$ is multiplied to the right-hand side of Eq.~\eqref{eq:universal_2}, i.e., the usual 2-universal hash function family is the 1-almost 2-universal hash function family.  When one uses the $\delta$-almost 2-universal hash function family as a random encoder, the term $\frac{\alpha\log\delta}{n}$ is further subtracted from the right-hand sides of Eqs.~\eqref{eq:error_exponent_universal} and \eqref{eq:error_exponent_partially_universal}, which does not change the asymptotically achievable compression rate as long as $\delta={\cal O}(\mathrm{poly}(n))$.
\end{remark}

We prove the following proposition for later application to QKD security proofs.
\begin{proposition} \label{prop:modified_for_qkd}
    Let $n$ be a positive integer and $q$ be a real number satisfying $0\leq q\leq 1$. Let ${\rm Binom}[n,q](m)$ be a binomial distribution defined as  
    \begin{equation}
        {\rm Binom}[n, q](m)\coloneqq {n \choose m}q^m (1-q)^{n-m}. \label{eq:def_of_binom}
    \end{equation}
    Let $\rho_{XQ}$ be a c-q state with a uniform marginal on the system $X$, i.e., $\rho_{XQ}=\sum_{x\in{\cal X}}|{\cal X}|^{-1}\ketbra{x}{x}_X\otimes 
    \rho^x_{Q}$ with a set of normalized states $\{\rho^x_{Q}\}_{x\in{\cal X}}$.  Let $\probP$ be a joint probability distribution of two random variables $\hat{m}$ and $\hat{\omega}$ with the ranges $\{0,\ldots,n\}$ and $\Omega$, respectively, and, for each $m\in\{0,\ldots,n\}$, let $\{\rho_{Q^m}^{\bm{x},\omega}\}_{\bm{x}\in{\cal X}^m, \omega\in\Omega}$ be a set of density operators such that 
    \begin{equation}
        \sum_{\omega\in\Omega}\probP(\hat{m}=m,\hat{\omega}=\omega) \sum_{\bm{x}\in{\cal X}^m}|{\cal X}|^{-m}\ketbra{\bm{x}}{\bm{x}}_{X^m}\otimes \rho_{Q^m}^{\bm{x},\omega} = {\rm Binom}[n, q](m)\, \rho_{XQ}^{\otimes m}. \label{eq:marginal_iid}
    \end{equation}
    For each $m$ and $\omega$, we perform a partially universal source compression with quantum side information for $\sum_{\bm{x}\in{\cal X}^m}|{\cal X}|^{-m}\ketbra{\bm{x}}{\bm{x}}_{X^m}\otimes \rho_{Q^m}^{\bm{x},\omega}$ with the following: the encoder uses a 2-universal family ${\cal H}_{\hat{m},\hat{\omega}}$ of hash functions with the range ${\cal B}_{\hat{m},\hat{\omega}}$ depending on $\hat{m}$ and $\hat{\omega}$, and the decoder uses the POVM $\tilde{Y}^n$ in Eq.~\eqref{eq:partially_universal_decoder}. Let ${\cal W}\subseteq \{0,\ldots,n\}\times \Omega$ be a set of values for the pair of random variables $(\hat{m},\hat{\omega})$.  
    Then, the probability $P_{\rm err}(\probP,\{(\{\rho_{Q^m}^{\bm{x},\omega}\}_{\bm{x}\in{\cal X}^m}, {\cal H}_{m,\omega})\}_{m,\omega},{\cal W})$ that the pair $(\hat{m},\hat{\omega})$ takes a value in ${\cal W}$ according to $\probP$ and the partially universal source compression with quantum side information protocol performed on $\sum_{\bm{x}\in{\cal X}^{\hat{m}}}|{\cal X}|^{-\hat{m}}\ketbra{\bm{x}}{\bm{x}}_{X^{\hat{m}}}\otimes \rho_{Q^{\hat{m}}}^{\bm{x},\hat{\omega}}$ fails is given by 
    \begin{align}
        \begin{split}
        &P_{\rm err}(\probP,\{(\{\rho_{Q^m}^{\bm{x},\omega}\}_{\bm{x}\in{\cal X}^m}, {\cal H}_{m,\omega})\}_{(m,\omega)},{\cal W}) \\
        &\coloneqq \sum_{(m,\omega)\in{\cal W}} \probP(\hat{m}=m,\hat{\omega}=\omega)\,\mathbb{E}_{h\sim |{\cal H}_{m,\omega}|^{-1}}\!\left[\sum_{\bm{x}\in{\cal X}^m}|{\cal X}|^{-m} \tr\!\left[\rho^{\bm{x},\omega}_{Q^m}\left(I_{Q^m} - \tilde{Y}^m_{h,h(\bm{x})}(\bm{x})\right)\right]\right] 
        \end{split}\label{eq:defining_rel_err}\\
        &\leq \overline{P}_{\rm err}(q\rho_{XQ},\{{\cal B}_{m,\omega}\}_{(m,\omega)},{\cal W}), \label{eq:variable_length_renyi}
    \end{align} 
    where 
    \begin{equation}
        \overline{P}_{\rm err}(\sigma_{XQ},\{{\cal B}_{m,\omega}\}_{(m,\omega)},{\cal W})\coloneqq \min_{\alpha\in[0,1]}\max_{(m,\omega)\in{\cal W}} 2^{-\alpha \bigl(\log|{\cal B}_{m,\omega}|-n H^{\uparrow,\leq}_{1-\alpha}(X|Q)_{\sigma}-[|{\cal X}|(d+2)(d-1)\log(m+1)]/ 2\bigr)}, \label{eq:def_overline_P}
    \end{equation}
    for a subnormalized $\sigma_{XQ}$, and 
    $H_{1-\alpha}^{\uparrow, \leq}(X|Q)_{\sigma}$ is defined in Eq.~\eqref{eq:subnormalized_Renyi}. 
\end{proposition}
\begin{proof}
    From Eq.~\eqref{eq:operator_ineq_for_partially_universal} in Lemma~\ref{lem:decoder_reformulation} and the defining equation Eq.~\eqref{eq:defining_rel_err}, we have for any $\alpha\in[0,1]$, 
    \begin{equation}
        \begin{split}
        P_{\rm err}(\probP,\{(\{\rho_{Q^m}^{\bm{x},\omega}\}_{\bm{x}\in{\cal X}^m}, {\cal H}_{m,\omega})\}_{(m,\omega)},{\cal W}) &\leq \sum_{(m,\omega)\in{\cal W}} \probP(\hat{m}=m,\hat{\omega}=\omega)\,\sum_{\bm{x}\in{\cal X}^m}|{\cal X}|^{-m}|{\cal B}_{m,\omega}|^{-\alpha} \\
        &\qquad \tr\!\left[\rho^{\bm{x},\omega}_{Q^m} (|{\cal X}|^{-m}\sigma_{\bm{x}})^{-\alpha}\!\left(\sum_{P\in{\cal P}_m}|{\cal T}_P|\prod_{x\in{\cal X}}|{\cal X}|^{-m P(x)}\sigma_{U,P}\right)^{\alpha}\right].
        \end{split} \label{eq:bound_with_B}
    \end{equation}
    For any $(m,\omega)\in{\cal W}$, we trivially have
    \begin{equation}
        |{\cal B}_{m,\omega}|^{-\alpha} \leq (m+1)^{-\alpha|{\cal X}|(d+2)(d-1)/2}\max_{(m^*,\omega^*)\in{\cal W}} 2^{-\alpha\bigl(\log|{\cal B}_{m^*,\omega^*}| - [|{\cal X}|(d+2)(d-1)\log(m^*+1)]/2\bigr)}. \label{eq:bin_size_bound}
    \end{equation} 
    Furthermore, we have 
    \begin{equation}
        \sum_{P\in{\cal P}_m}|{\cal T}_P|\prod_{x\in{\cal X}}|{\cal X}|^{-m P(x)}\sigma_{U,P} = \sum_{P\in{\cal P}_m}|{\cal T}_P||{\cal X}|^{-m}\sigma_{U,P} =\sum_{\bm{x}\in{\cal X}^m}|{\cal X}|^{-m} \sigma_{\bm{x}} \eqqcolon \tau_m, \label{eq:uniform_sum_sequence}
    \end{equation}
    where $\tr[\tau_m]=1$.
    Applying Eqs.~\eqref{eq:bin_size_bound} and~\eqref{eq:uniform_sum_sequence} to Eq.~\eqref{eq:bound_with_B}, we have 
    \begin{align}
        \begin{split}
        &P_{\rm err}(\probP,\{(\{\rho_{Q^m}^{\bm{x},\omega}\}_{\bm{x}\in{\cal X}^m}, {\cal H}_{m,\omega})\}_{(m,\omega)},{\cal W}) \\ &\leq \max_{(m^*,\omega^*)\in{\cal W}} 2^{-\alpha\bigl(\log|{\cal B}_{m^*,\omega^*}| - [|{\cal X}|(d+2)(d-1)\log(m^*+1)]/2\bigr)} \\
        &\quad \sum_{m=0}^n\sum_{\omega\in\Omega}(m+1)^{-\alpha|{\cal X}|(d+2)(d-1)/2} \probP(\hat{m}=m, \hat{\omega}=\omega)\sum_{\bm{x}\in{\cal X}^m} \tr\!\left[|{\cal X}|^{-m}\rho^{\bm{x},\omega}_{Q^m} (|{\cal X}|^{-m}\sigma_{\bm{x}})^{-\alpha}\tau_m^{\alpha}\right]
        \end{split} \label{eq:before_alpha_bound}
    \end{align} 
    Now, notice that the following equalities hold:
    \begin{align}    
        &\sum_{\omega\in\Omega} \probP(\hat{m}=m, \hat{\omega}=\omega)\sum_{\bm{x}\in{\cal X}^m} \tr\!\left[|{\cal X}|^{-m}\rho^{\bm{x},\omega}_{Q^m} (|{\cal X}|^{-m}\sigma_{\bm{x}})^{-\alpha}\tau_m^{\alpha}\right] \nonumber \\
        &= \sum_{\omega\in\Omega} \probP(\hat{m}=m, \hat{\omega}=\omega) \sum_{\bm{x},\bm{y}\in{\cal X}^m} \tr\!\left[|{\cal X}|^{-m}\ketbra{\bm{x}}{\bm{x}}_{X^m}\otimes\rho^{\bm{x},\omega}_{Q^m} (|{\cal X}|^{-m}\ketbra{\bm{y}}{\bm{y}}_{X^m}\otimes \sigma_{\bm{y}})^{-\alpha}\tau_m^{\alpha}\right] \\
        &= {\rm Binom}[n,q](m)\sum_{\bm{y}\in{\cal X}^m} \tr\!\left[\rho_{XQ}^{\otimes m} \bigl(|{\cal X}|^{-m}\ketbra{\bm{y}}{\bm{y}}_{X^m}\otimes \sigma_{\bm{y}}\bigr)^{-\alpha}\tau_m^{\alpha}\right] \\
        &= {\rm Binom}[n,q](m) \sum_{\bm{y}\in{\cal X}^m} \tr\!\left[|{\cal X}|^{-m}\rho_{Q^m}^{\bm{y}} \bigl(|{\cal X}|^{-m}\sigma_{\bm{y}}\bigr)^{-\alpha}\tau_m^{\alpha}\right],
    \end{align}
    where $\rho_{Q^m}^{\bm{y}}\coloneqq \rho_Q^{y_1}\otimes\cdots\otimes\rho_Q^{y_m}$, and the second equality follows from Eq.~\eqref{eq:marginal_iid}.
    From this, the summation over $m$ and $\omega$ in Eq.~\eqref{eq:before_alpha_bound} can be taken as 
    \begin{align}
        &\sum_{m=0}^n\sum_{\omega\in\Omega}(m+1)^{-\alpha|{\cal X}|(d+2)(d-1)/2} \,\probP(\hat{m}=m, \hat{\omega}=\omega)\sum_{\bm{x}\in{\cal X}^m} \tr\!\left[|{\cal X}|^{-m}\rho^{\bm{x},\omega}_{Q^m} (|{\cal X}|^{-m}\sigma_{\bm{x}})^{-\alpha}\tau_m^{\alpha}\right] \nonumber \\
        &= \sum_{m=0}^n(m+1)^{-\alpha|{\cal X}|(d+2)(d-1)/2} {\rm Binom}[n,q](m) \sum_{\bm{y}\in{\cal X}^m} \tr\!\left[|{\cal X}|^{-m}\rho_{Q^m}^{\bm{y}} \bigl(|{\cal X}|^{-m}\sigma_{\bm{y}}\bigr)^{-\alpha}\tau_m^{\alpha}\right]\\
        &\leq \sum_{m=0}^n {\rm Binom}[n,q](m)\, \tr\!\left[\sum_{\bm{y}\in{\cal X}^m}\left(|{\cal X}|^{-m}\rho_{Q^m}^{\bm{y}}\right)^{1-\alpha}\tau_m^{\alpha}\right] \\
        &\leq \sum_{m=0}^n {\rm Binom}[n,q](m)\,\left(\tr\left[\left(\left(\tr_X[\rho_{XQ}^{1-\alpha}]\right)^{\otimes m}\right)^{\frac{1}{1-\alpha}}\right]\right)^{1-\alpha}, \label{eq:convex_combination_of_renyi}
    \end{align} 
    where the first inequality follows from Eq.~\eqref{eq:1-alpha_bound}, and the second inequality follows from Eq.~\eqref{eq:Hayashi_lemma}.
    Since the function $t\mapsto t^{1-\alpha}$ is concave, we have 
    \begin{align}
        &\sum_{m=0}^n {\rm Binom}[n,q](m)\,\left(\tr\left[\left(\left(\tr_X[\rho_{XQ}^{1-\alpha}]\right)^{\otimes m}\right)^{\frac{1}{1-\alpha}}\right]\right)^{1-\alpha} \nonumber \\
        &\leq \left(\sum_{m=0}^n {\rm Binom}[n,q](m)\,\tr\left[\left(\left(\tr_X[\rho_{XQ}^{1-\alpha}]\right)^{\frac{1}{1-\alpha}}\right)^{\otimes m}\right]\right)^{1-\alpha} \\
        &= \left(\sum_{m=0}^n {n\choose m}(1-q)^{n-m}\,\left(\tr\left[\left(\tr_X[(q\rho_{XQ})^{1-\alpha}]\right)^{\frac{1}{1-\alpha}}\right]\right)^{m}\right)^{1-\alpha} \\
        &= \left(\tr\left[\left(\tr_X[(q\rho_{XQ})^{1-\alpha}]\right)^{\frac{1}{1-\alpha}}\right] + (1-q)\right)^{n(1-\alpha)} \\
        &= 2^{n\alpha H^{\uparrow,\leq}_{1-\alpha}(X|Q)_{q\rho}}, \label{eq:bound_by_subnormalized_renyi}
    \end{align}
    where we used Eq.~\eqref{eq:expression_subnormalized} in the last equality. Combining Eqs.~\eqref{eq:before_alpha_bound}, \eqref{eq:convex_combination_of_renyi}, and \eqref{eq:bound_by_subnormalized_renyi}, we prove the statement.
\end{proof}

Proposition~\ref{prop:modified_for_qkd} states that the failure probability of the partially universal source compression with quantum side information performed on the non-i.i.d.~state $\bigoplus_{(m,\omega)\in{\cal W}}\probP(\hat{m}=m,\hat{\omega}=\omega)\sum_{\bm{x}\in{\cal X}^m}|{\cal X}|^{-m}\rho_{Q^m}^{\bm{x},\omega}$, which has an i.i.d.~marginal as in Eq.~\eqref{eq:marginal_iid}, can be bounded from above by $\overline{P}_{\rm err}(q\rho_{XQ},\{{\cal B}_{m,\omega}\}_{(m,\omega)},{\cal W})$ in Eq.~\eqref{eq:def_overline_P}, which depends only on single-copy quantity $q \rho_{XQ}$. In the QKD scenario, $\hat{m}$ corresponds to the length of a sifted key and $\hat{\omega}$ corresponds to other random variables that are used for parameter estimations. Thus, $q$ corresponds to a probability of the successful sifting, which makes a quantum state corresponding to a sifted key subnormalized. The fact that $H_{1-\alpha}^{\uparrow, \leq}(X|Q)_{\tilde{\rho}}$ is concave with a subnormalized $\tilde{\rho}$ will be used later.

\section{Security proof based on phase error correction} \label{sec:security_proof_universal}
In this section, we apply universal source compression with quantum side information developed in the previous section to the security proof of QKD.  This new PEC-type security proof is applied to a key with any base-$p$ number for a prime $p$.  Thus, we first extend the definition of the complementary bases beyond the conventional Pauli-$Z$ and $X$ bases.  We then establish the security condition for such a generalized key and introduce a virtual protocol to correct an error in the complementary basis (i.e., the phase error), which is the essence of the PEC-type security proof.  Since the universal decoder developed in the previous section can only be applied to i.i.d.~sources, we need to reduce the security against general attacks to that against collective attacks.  This reduction is achieved by the recently developed i.i.d.~reduction technique~\cite{Matsuura2024}, which builds on an old idea~\cite{Tamaki2003} and can be viewed as a PEC analogue of the post-selection technique~\cite{Christandl2009}.  As a result of this reduction, together with the partially universal decoder for classical source compression with quantum side information, the security analysis of a QKD protocol boils down to estimating a single R\'enyi-entropic quantity.  This estimation can be performed via a convex optimization technique, although the resulting problem takes the form of a nonlinear convex semidefinite programming.  
We finally comment on the asymptotic optimality of our newly developed security proof based on universal decoding for source compression with quantum side information.

\subsection{Definition of the complementary bases} \label{sec:complementary_bases}
In the method of phase error correction, the bases complementary to each other play an important role.  Here, we define the complementary bases used in this article based on a prime field.  
Let $p$ be a prime number and $\mathbb{F}_p$ be the prime field with the base $p$.  We pick up a basis $\{\ket{c}:c\in \mathbb{F}_p\}$ for a $p$-dimensional Hilbert space, and define operators $X(a)$ and $Z(b)$ for $a,b\in\mathbb{F}_{p}$ as 
\begin{align}
    X(a) &\coloneqq \sum_{c\in\mathbb{F}_{p}}\ketbra{c+a}{c}, \label{eq:generalized_X} \\
    Z(b) &\coloneqq  \sum_{c\in\mathbb{F}_{p}} \chi_p(b c)\ketbra{c}{c},\label{eq:generalized_Z}
\end{align}
where the addition and multiplication above is of $\mathbb{F}_{p}$, and the additive character $\chi_p$ is defined as 
\begin{equation}
    \chi_p(a)\coloneqq\exp\!\left(\frac{2\pi i a}{p}\right). \label{eq:additive_character}
\end{equation}
They satisfy the following commutation relation:
\begin{equation}
    X(a)Z(b) = \chi_p(-ab)Z(b)X(a). \label{eq:commutation_relation}
\end{equation}
In the case $p=2$, these operators correspond to the usual Pauli-X and Z operators.

We define the $X$ (resp.~$Z$) basis as the diagonalizing basis of the operator $X(a)$ (resp.~$Z(b)$).  We put tilde to denote the $X$ basis and put nothing to denote the $Z$ basis, i.e., $\{\ket{\widetilde{c}}:c\in\mathbb{F}_{p}\}$ is the $X$ basis and $\{\ket{c}:c\in\mathbb{F}_{p}\}$ is the $Z$ basis.  From the definitions~\eqref{eq:generalized_X} and \eqref{eq:generalized_Z}, we have 
\begin{equation}
    \ket{\widetilde{c}} = \frac{1}{\sqrt{p}}\sum_{c'\in\mathbb{F}_{p}} \chi_p(-c c') \ket{c'}. \label{eq:basis_transform_finite_field}
\end{equation}
where $\ket{\widetilde{c}}$ is the eigenstate of $X(a)$ with the eigenvalue $\chi_p(a c)$ for any $a\in\mathbb{F}_{p}$. Note that $Z(b)$ shifts $\ket{\tilde{c}}$ in the minus direction, i.e.,
\begin{equation}
    Z(b)\ket{\widetilde{c}} = \ket{\widetilde{c-b}}. \label{eq:b_subtraction}
\end{equation}
The bases $\{\ket{\widetilde{c}}:c\in\mathbb{F}_{p}\}$ and $\{\ket{c}:c\in\mathbb{F}_{p}\}$ are thus mutually unbiased~\cite{Durt2010}.
Furthermore, the character $\chi_p$ satisfies~\cite{Durt2010} 
\begin{equation}
    \sum_{b\in\mathbb{F}_{p}} \chi_p(a b) =p\delta_{a, 0}, \label{eq:Kronecker_relation}
\end{equation} 
for any $a\in\mathbb{F}_{p}$, where $\delta_{a, b}$ denotes the Kronecker delta and $0$ denotes the additive identity of $\mathbb{F}_{p}$.

The formalism above can be generalized to a composite system by considering a vector space over the field $\mathbb{F}_{p}$.  Consider an $n$-qudit system, each of which has local dimension $p$. Then, the $n$-qudit $X$ and $Z$ operators are given respectively by $X(\bm{a})\coloneqq X(a_1)\otimes\cdots\otimes X(a_{n})$ and $Z(\bm{b})\coloneqq Z(b_1)\otimes\cdots\otimes Z(b_{n})$ for row vectors $\bm{a},\bm{b}\in\mathbb{F}_{p}^{n}$.  
We then have $X(\bm{a})Z(\bm{b}) = \chi_p(-\bm{a}\bm{b}^{\top})Z(\bm{b}) X(\bm{a})$, where $\top$ denotes the transposition.  Let $\{\ket{\bm{c}}\}_{\bm{c}\in\mathbb{F}_p^{n}}$ (resp.~$\{\ket{\widetilde{\bm{c}}}\}_{\bm{c}\in\mathbb{F}_p^{n}}$) be the diagonalizing basis of $X(\bm{a})$ for all $\bm{a}\in\mathbb{F}_{p}^{n}$ (resp.~$Z(\bm{b})$ for all $\bm{b}\in \mathbb{F}_{p}^{n}$), which we called the $X$ basis (resp.~the $Z$ basis) for the $n$-qudit system.
Generalization of the relation~\eqref{eq:basis_transform_finite_field} is then given by 
\begin{equation}
    \ket{\widetilde{\bm{c}}} = p^{-n/2}\sum_{\bm{c}'\in\mathbb{F}_{p}^{n}}\chi_p(-\bm{c}'\bm{c}^{\top})\ket{\bm{c}'}.\label{eq:basis_transform_n_dit}
\end{equation}
For an invertible matrix $C\in\mathbb{F}_p^{n\times n}$, there exists a unitary $U(C)$ such that 
\begin{equation}
    U(C)X(\bm{a})U(C)^{\dagger}=X(\bm{a}C^{\top}). \label{eq:X_op_transform}
\end{equation}
From the relation~\eqref{eq:basis_transform_n_dit}, the same unitary $U(C)$ transforms $Z(\bm{b})$ as 
\begin{equation}
    U(C)Z(\bm{b})U(C)^{\dagger}=Z(\bm{b}C^{-1}). \label{eq:Z_op_transform}
\end{equation}

In the following sections, an $r$ tuple of a finite field $\mathbb{F}_{p}$ is associated with alphabets ${\cal X}$ with $|{\cal X}|= p^r$; i.e., each of $\{(p_1,\ldots,p_r):p_1,\ldots,p_r\in\mathbb{F}_p\}$ is labeled by a distinct letter in ${\cal X}$. Thus, one extracts an $r$-dit sifted key per a communication round.
For protocols that primarily extract sifted keys of a base-$d$ number with non-prime-power $d$, we embed it into an $r$-tuple of a finite field $\mathbb{F}_p$ with $p^r>d$ and perform post-processing with $\mathbb{F}_{p}$. Thus, if we embed it into an $r$-tuple of $\mathbb{F}_2$, which always exists satisfying $d<2^r<2d$, then the resulting post-processing is based on conventional binary numbers.
The details will be described in the subsequent section.

\subsection{Introduction of a virtual protocol} \label{sec:virtual_protocol}

In the following sections, we denote random variables with the symbol $\hat{\cdot}$.  
We consider here a prepare-and-measure protocol, where a sender Alice randomly generates an alphabet $\hat{a}\in{\cal X}_{\sysA}$, encodes it to a quantum state, and sends it to a receiver Bob. 
Bob measures a received quantum state and probabilistically obtains an outcome $\hat{b}\in{\cal X}_{\sysB}$.  
Alice and Bob announce some random variables $\hat{\xi}_{\rm pub}=\mathfrak{t}(\hat{a},\hat{b})$, where $\mathfrak{t}:{\cal X}_{\sysA}\times{\cal X}_{\sysB}\to \Omega_{\rm pub}$ denotes a fixed map. Then, they perform sifting with maps $\mathfrak{s}[\xi]:{\rm proj}_{\sysA}(\mathfrak{t}^{-1}(\xi))\to{\cal X}_{\rm sift}$ and $\mathfrak{s}'[\xi]:{\rm proj}_{\sysB}(\mathfrak{t}^{-1}(\xi))\to{\cal X}_{\rm sift}$ for any $\xi\in\Omega_{\rm pub}$, where $\mathfrak{t}^{-1}$ denotes a preimage and ${\rm proj}_{\sysA}:{\cal X}_{\sysA}\times{\cal X}_{\sysB}\to{\cal X}_{\sysA}$ (resp.~${\rm proj}_{\sysB}:{\cal X}_{\sysA}\times{\cal X}_{\sysB}\to{\cal X}_{\sysB}$) denotes a projection of a Cartesian product, to obtain sifted keys $\hat{\bm{k}}_A^{\rm sift}, \hat{\bm{k}}_B^{\rm sift}\in{\cal X}_{\rm sift}^{\hat{n}_{\rm sift}}$ of the length $\hat{n}_{\rm sift}$.
Our following analysis applies to permutation-symmetric protocols.  The protocol can be made permutation-symmetric as long as Alice and Bob perform jointly permutation-symmetric quantum operations (e.g., i.i.d.~state preparation and measurement) and jointly permutation-symmetrize the obtained classical data. Here, we restrict our attention to the case where they perform i.i.d.~quantum operations for simplicity. For classical data, since only the frequency of each outcome is used for parameter estimations, which is automatically permutation-symmetric, they just need to reorder their sifted keys $\hat{\bm{k}}_A^{\rm sift}, \hat{\bm{k}}_B^{\rm sift}\in{\cal X}_{\rm sift}^{\hat{n}_{\rm sift}}$ joint-randomly. This requires ${\cal O}(\hat{n}_{\rm sift}\log\hat{n}_{\rm sift})$-bit randomness source and public communication.

As mentioned in the previous section, if the cardinality $|{\cal X}_{\rm sift}|$ is not a prime power, then Alice and Bob embed their (permutation-symmetrized) sifted keys to $\mathbb{F}_p^r$, where $p$ is a prime, so that an element of their final keys is in $\mathbb{F}_p$.   
By performing further classical post-processing that consists of information reconciliation and privacy amplification, Alice and Bob obtain the final keys $\hat{\bm{k}}_A^{\rm fin},\hat{\bm{k}}_B^{\rm fin}\in\mathbb{F}_p^{\hat{n}_{\rm fin}}$ of the length $\hat{n}_{\rm fin}$.
The definition of the $\varepsilon$-security is given in terms of the final classical-classical-quantum state $\rho_{K_A^n K_B^n E|\hat{n}_{\rm fin}=n}^{\rm fin}$ between Alice's and Bob's final keys and an adversary Eve as~\cite{Mueller2009} 
\begin{equation}
    \frac{1}{2}\sum_{n\geq 1}\mathrm{Pr}(\hat{n}_{\rm fin}=n)\|\rho_{K_A^n K_B^n E|\hat{n}_{\rm fin}=n}^{\rm fin} -  \rho_{K_A^n K_B^n E|\hat{n}_{\rm fin}=n}^{\rm ideal}\| \leq \varepsilon,
\end{equation}
where $\rho_{K_A^n K_B^n E|\hat{n}_{\rm fin}=n}^{\rm ideal}$ is defined as 
\begin{equation}
    \rho_{K_A^n K_B^n E|\hat{n}_{\rm fin}=n}^{\rm ideal}\coloneqq \sum_{\bm{k}\in\mathbb{F}_p^n} |{\cal X}|^{-n}\ketbra{\bm{k}}{\bm{k}}_{K_A^n}\otimes\ketbra{\bm{k}}{\bm{k}}_{K_B^n}\otimes \tr_{K_A^n K_B^n}\left[\rho_{K_A^n K_B^n E|\hat{n}_{\rm fin}=n}^{\rm fin}\right].
\end{equation}
The security of the final key is known to be split into two conditions: correctness and secrecy~\cite{Koashi2009}.  The former requires that Alice's and Bob's final keys are the same with a high probability, i.e.,
\begin{equation}
    \mathrm{Pr}(\hat{n}_{\rm fin}\geq 1, \hat{\bm{k}}_{A}^{\rm fin}\neq \hat{\bm{k}}_{B}^{\rm fin}) \leq \varepsilon_{\rm cor},
\end{equation}  
where $\hat{n}_{\rm fin}=0$ whenever the protocol is aborted.
This should be ensured by the information reconciliation step in the actual protocol.  The latter requires that Alice's final key looks almost completely random to an eavesdropper Eve, i.e.,  
\begin{equation}
    \frac{1}{2}\sum_{n\geq 1}\mathrm{Pr}(\hat{n}_{\rm fin}=n)\left\|\rho_{K_A^n E|\hat{n}_{\rm fin}=n}^{\rm fin} - \rho_{K_A^n E|\hat{n}_{\rm fin}=n}^{\rm ideal}\right\| \leq \varepsilon_{\rm sec}. \label{eq:secrecy_condition}
\end{equation}
(In the case of reverse reconciliation, the requirement is for Bob's key instead of Alice's key.)  If these conditions are met, the protocol is $(\varepsilon_{\rm cor}+\varepsilon_{\rm sec})$-secure.  The security proof of QKD protocols mainly focuses on proving the latter condition~\eqref{eq:secrecy_condition}.

In the security proof based on PEC~\cite{Lo1999, Shor2000, Koashi2009}, the secrecy is proved through the introduction of a virtual protocol.  
The goal of a virtual protocol is to establish the existence of a family $\{\rho_{K_A^nE|\hat{n}_{\rm fin}=n}^{\rm virt}\}_{n\geq 1}$ of quantum-quantum states such that the following two conditions hold, with $\hat{n}_{\rm fin}$ required to follow the same probability distribution as in the actual protocol~\cite{Matsuura2023}:
\begin{align}
    \forall n\geq 1,\qquad \sum_{\bm{k}\in\mathbb{F}_p^{n}}\ketbra{\bm{k}}{\bm{k}}_{K_A^n}\rho_{K_A^nE|\hat{n}_{\rm fin}=n}^{\rm virt}\ketbra{\bm{k}}{\bm{k}}_{K_A^n} &= \rho_{K_A^nE|\hat{n}_{\rm fin}=n}^{\rm fin}, \label{eq:compatibility} \\
    \intertext{and}
    \sum_{n\geq 1}\mathrm{Pr}(\hat{n}_{\rm fin}=n) \left(1 - \bra{\widetilde{\bm{0}}}_{K_A^n} \rho_{K_A^n|\hat{n}_{\rm fin}=n}^{\rm virt} \ket{\widetilde{\bm{0}}}_{K_A^n}\right) &\leq \frac{\varepsilon_{\rm sec}^2}{2},
    \label{eq:phase_error_corrected}
\end{align}
where $\ket{\widetilde{\bm{0}}}$ is the $+1$ eigenstate of $X(\bm{a})$ for any $\bm{a}\in \mathbb{F}_p^n$.

To construct $\{\rho_{K_A^nE|\hat{n}_{\rm fin}=n}^{\rm virt}\}_{n\geq 1}$ that satisfies Eqs.~\eqref{eq:compatibility} and \eqref{eq:phase_error_corrected} in a virtual protocol, Alice prepares an entangled state and keeps a part of it, instead of encoding classical information into a quantum state.  During $n_{\rm tot}$ rounds of quantum communication (i.e., $n_{\rm tot}$ optical pulses sent), Alice and Bob need to give Eve the same information as that in the actual protocol, but they can perform an arbitrary quantum operation as long as this requirement is fulfilled.  Finally, Alice performs a quantum version of the post-processing on the quantum systems she keeps and obtains $\rho_{K_A^nE|\hat{n}_{\rm fin}=n}^{\rm virt}$ with a probability $\mathrm{Pr}(\hat{n}_{\rm fin}=n)$. (This implicitly imposes that the probability distribution for $\hat{n}_{\rm fin}$ in a virtual protocol needs to be the same as that in the actual protocol, as mentioned above.)  These are to ensure that the condition~\eqref{eq:compatibility} is satisfied.
As the condition~\eqref{eq:phase_error_corrected} suggests, the $X$-basis error of the state at the end of the virtual protocol needs to be corrected with a high probability when averaged over $\hat{n}_{\rm fin}$.  This can be achieved by Alice's and Bob's collaborative error correction at the stage of the quantum version of the post-processing, which is the origin of the name ``phase error correction''.  For more details, see Refs.~\cite{Koashi2009, Matsuura2023}.  
It is important to point out here that Eve's system does not explicitly appear in the condition~\eqref{eq:phase_error_corrected}.  We only need to evaluate the failure probability of the phase error correction with Alice and Bob cooperating.

For simplicity, here we only consider a variable-key-length prepare-and-measure QKD protocol.
A procedure of such a QKD protocol and its corresponding virtual protocol can be described as follows.

\bigskip 
\noindent --- Actual protocol ---

Prior to the protocol, Alice and Bob agree on the protocol parameters and the total number $n_{\rm tot}$ of quantum communication rounds.

\begin{enumerate}
    \item Alice generates a random alphabet $\hat{a}\in{\cal X}_{\sysA}$ with a probability distribution $\{q_a\}_{a\in{\cal X}_{\sysA}}$ and encodes it to a quantum state $\rho_{\hat{a}}\in{\cal D}({\cal H}_{\sysC})$ on a system $\sysC$. The system $\sysC$ is sent to a quantum channel and received by Bob as a system $\sysB$. Bob performs a measurement on the system $\sysB$ to obtain an outcome $\hat{b}\in{\cal X}_{\sysB}$, with its POVM written as $\{M_{\sysB}^b\}_{b\in{\cal X}_{\sysB}}$.
    Alice and Bob repeat this quantum communication for $n_{\rm tot}$ rounds, where Alice may initiate each round before the previous round is completed. 
    \item After $n_{\rm tot}$ rounds of quantum communication, Alice and Bob publicly exchange messages based on their respective random variables. The messages for each round are collectively represented by a random variable $\hat{\xi}^{(i)}_{\rm pub}\coloneqq \mathfrak{t}(\hat{a}^{(i)},\hat{b}^{(i)})$ for $i=1,\ldots,n_{\rm tot}$ with a map $\mathfrak{t}:{\cal X}_{\sysA}\times{\cal X}_{\sysB}\to\Omega_{\rm pub}$, where $\Omega_{\rm pub}$ is a (finite) set. From $\hat{\xi}_{\rm pub}^{(i)}$, Alice and Bob determines a random variable $\hat{s}^{(i)}$ which represents the success ($\hat{s}^{(i)}=1$) and failure ($\hat{s}^{(i)}=0$) of sifting in this round. When $\hat{s}^{(i)}=1$, Alice determines her sifted value in ${\cal X}_{\rm sift}$ from $\hat{a}^{(i)}$ and $\hat{\xi}_{\rm pub}^{(i)}$, and Bob does the same from $\hat{b}^{(i)}$ and $\hat{\xi}_{\rm pub}^{(i)}$. Here, Alice's sifting procedure is represented by a partition 
    \begin{equation}
        \Omega_{\rm pub} = \Omega_{\rm pub}^{(0)} \sqcup \Omega_{\rm pub}^{(1)}
    \end{equation} 
    and a collection $\{\mathfrak{s}[\xi]\}_{\xi\in\Omega_{\rm pub}^{(1)}}$ of maps $\mathfrak{s}[\xi]:{\rm proj}_{\sysA}(\mathfrak{t}^{-1}(\xi))\to{\cal X}_{\rm sift}$, where $\mathfrak{t}^{-1}(\xi)$ denotes a preimage of $\xi$ and ${\rm proj}_{\sysA}:{\cal X}_A\times{\cal X}_B\to{\cal X}_A$ denotes a projection of a Cartesian product, such that $\hat{s}^{(i)}=1$ iff $\hat{\xi}_{\rm pub}^{(i)}=\mathfrak{t}(\hat{a}^{(i)},\hat{b}^{(i)})\in \Omega_{\rm pub}^{(1)}$ and then Alice's sifted value is $\mathfrak{s}[\hat{\xi}_{\rm pub}^{(i)}](\hat{a}^{(i)})$.
    Similarly, Bob's sifted value is represented as $\mathfrak{s}'[\hat{\xi}_{\rm pub}^{(i)}](\hat{b}^{(i)})$ with a collection $\{\mathfrak{s}'[\xi]\}_{\xi\in\Omega_{\rm pub}^{(1)}}$ of maps $\mathfrak{s}'[\xi]:{\rm proj}_{\sysB}(\mathfrak{t}^{-1}(\xi))\to{\cal X}_{\rm sift}$.
    They compute the frequency of public announcements $\hat{\Xi}_{\rm pub}\coloneqq n_{\rm tot} P_{\hat{\bm{\xi}}_{\rm pub}}$, where $\hat{\bm{\xi}}_{\rm pub}\coloneqq (\hat{\xi}_{\rm pub}^{(1)},\ldots,\hat{\xi}_{\rm pub}^{(n_{\rm tot})})$, and $P_{\bm{x}}$ is the type of $\bm{x}$.
    \item (Permutation symmetrization) By consuming ${\cal O}(\hat{n}_{\rm sift}\log\hat{n}_{\rm sift})$-bit local randomness and public communication with $\hat{n}_{\rm sift}\coloneqq \sum_{i=1}^{n_{\rm tot}}\hat{s}^{(i)}$, Alice and Bob joint-randomly reorder their sifted values they generated for the rounds with $\hat{s}^{(i)}=1$ and obtain sifted keys $\hat{\bm{k}}^{\rm sift}_A$ and $\hat{\bm{k}}^{\rm sift}_B$, respectively, where $\hat{\bm{k}}_A^{\rm sift}, \hat{\bm{k}}_B^{\rm sift}\in{\cal X}_{\rm sift}^{\hat{n}_{\rm sift}}$.
    \item (Information reconciliation) Depending on the the frequency $\hat{\Xi}_{\rm pub}$ of announcements and the length $\hat{n}_{\rm sift}$ of the sifted key, Alice computes the required number $K_{\rm EC}(\hat{n}_{\rm sift}, \hat{\Xi}_{\rm pub})$ of syndrome dits for information reconciliation and the required number $K_{\rm PA}(\hat{n}_{\rm sift}, \hat{\Xi}_{\rm pub})$ of sacrificed dits for privacy amplification.
    Alice then defines $\hat{n}_{\rm fin}\coloneqq n_{\rm fin}(\hat{n}_{\rm sift},\hat{\Xi}_{\rm pub})$ as 
    \begin{equation}
        n_{\rm fin}(\hat{n}_{\rm sift},\hat{\Xi}_{\rm pub}) = 
        \begin{cases}
            \hat{n}_{\rm sift} r - K_{\rm PA}(\hat{n}_{\rm sift}, \hat{\Xi}_{\rm pub}) & \text{if } \hat{n}_{\rm sift} r - K_{\rm PA}(\hat{n}_{\rm sift}, \hat{\Xi}_{\rm pub}) > K_{\rm EC}(\hat{n}_{\rm sift}, \hat{\Xi}_{\rm pub}), \\
            0 & \text{otherwise},
        \end{cases} \label{eq:final_key_as_function}
    \end{equation}
    which is a deterministic function of $\hat{n}_{\rm sift}$ and $\hat{\Xi}_{\rm pub}$.
    If $\hat{n}_{\rm fin}=0$, Alice aborts the protocol. Otherwise, Alice sends Bob a $K_{\rm EC}(\hat{n}_{\rm sift}, \hat{\Xi}_{\rm pub})$-dit syndrome by consuming $K_{\rm EC}(\hat{n}_{\rm sift}, \hat{\Xi}_{\rm pub})$-dit preshared secret key, and Bob performs error correction on his sifted key $\hat{\bm{k}}^{\rm sift}_B$ according to the received syndrome to obtain a reconciled key $\hat{\bm{k}}^{\rm rec}_B$. 
    \item (Privacy amplification) Alice and Bob determine an injective map $\mathfrak{e}:{\cal X}_{\rm sift}\hookrightarrow\mathbb{F}_p^r$, where $p$ is a prime number.
    Alice randomly chooses a hash function $H_{\hat{i}}$ from a dual 2-universal family ${\cal H}_d(\hat{n}_{\rm sift}r,\hat{n}_{\rm fin})\eqqcolon \{H_i\}_i$ of surjective linear hash functions over the field $\mathbb{F}_{p}$, where $\hat{i}\in\{1,\ldots,|{\cal H}_d(\hat{n}_{\rm sift}r,\hat{n}_{\rm fin})|\}$. Alice then acts it on her sifted key to obtain the final key $\hat{\bm{k}}_A^{\rm fin}\coloneqq \mathfrak{e}^{\hat{n}_{\rm sift}}(\hat{\bm{k}}_A^{\rm sift})H_{\hat{i}}^{\top} \in \mathbb{F}_p^{\hat{n}_{\rm fin}}$, where we regard $H_{\hat{i}}$ as an $\hat{n}_{\rm fin}\times \hat{n}_{\rm sift}r$ matrix over $\mathbb{F}_p$.  Alice sends the index $\hat{i}$ to Bob, who applies the hash function $H_{\hat{i}}$ to his reconciled key to obtain the final key $\hat{\bm{k}}_B^{\rm fin}\coloneqq \mathfrak{e}^{\hat{n}_{\rm sift}}(\hat{\bm{k}}_B^{\rm rec})H_{\hat{i}}^{\top}$.
\end{enumerate}

In the above actual protocol, the random variables $\hat{a}$ and $\hat{b}$ should be chosen so that the sifted value and publicly exchanged information are deterministically defined through them. For example, when Alice (resp.~Bob) randomly switches the encoding basis (resp.~the measurement basis), then $a\in{\cal X}_{\sysA}$ (resp.~$b\in{\cal X}_{\sysB}$) should contain the basis information as well. For a protocol in which new random variables, independent of Eve, are generated during the sifting, we reformulate it so that Alice or Bob generates them prior to the communication, which amounts to ${\cal X}_{\sysA}$ or ${\cal X}_{\sysB}$ containing such random variables. This also implies that some of the encoded states $\{\rho_a\}_{a\in{\cal X}_{\sysA}}$ may be the same state for different values of $a$. Note that the following analyses can be straightforwardly adapted to the reverse reconciliation scenario by simply exchanging the role of $\sysA$ and $\sysB$.

To introduce a virtual protocol corresponding to the above actual protocol, we first define the following.
For each element $H_i \in \{H_i\}_i = {\cal H}_d(\hat{n}_{\rm sift}r,\hat{n}_{\rm fin})$, let $\overline{H}_i$ be an $(\hat{n}_{\rm sift}r - \hat{n}_{\rm fin})\times \hat{n}_{\rm sift}r$ matrix such that $(H_i^{\top}~\overline{H}_i^{\top})$ is an invertible matrix. Such a matrix $\overline{H}_i$ is not uniquely determined from $H_i$, but any choice that satisfies the above condition is allowed for our purpose. Then, let $G_i$ and $\overline{G}_i$ be $(\hat{n}_{\rm sift}r - \hat{n}_{\rm fin})\times \hat{n}_{\rm sift}r$ and $\hat{n}_{\rm fin} \times \hat{n}_{\rm sift}r$ matrices, respectively, such that $(\overline{G}_i^{\top}~G_i^{\top})$ is an inverse transpose of $(H_i^{\top}~\overline{H}_i^{\top})$. Then, we have $\overline{G}_i H_i^{\top}= I_{\hat{n}_{\rm fin}}$, $G_i \overline{H}_i^{\top}= I_{\hat{n}_{\rm sift}r - \hat{n}_{\rm fin}}$, and $\overline{G}_i \overline{H}_i^{\top}=0=G_i H_i^{\top}$. From the definition of dual 2-universal family~\cite{Tsurumaru2013}, the set $\{G_i\}_i \eqqcolon {\cal H}(\hat{n}_{\rm sift}r,\hat{n}_{\rm sift}r-\hat{n}_{\rm fin})$ is nothing but a 2-universal family of surjective linear hash functions with the field $\mathbb{F}_{p}$. Note that we can use an almost dual 2-universal hash family for $\{H_i\}_i$ instead of a 2-universal hash family, which results in $\{G_i\}_i$ being an almost 2-universal family, at the price of an additional overhead in the privacy amplification. See Ref.~\cite{Tsurumaru2013} for the details of an almost (dual) 2-universal family. 

We then introduce an auxiliary quantum system and a unitary acting on it that together simulate the privacy amplification in the actual protocol. Consider an $\hat{n}_{\rm sift} r$-qudit system $K^{\hat{n}_{\rm sift}r}$ with $\dim{\cal H}_K = p$. Furthermore, consider a unitary $U(H_i, \overline{H}_i)$ such that  
\begin{equation}
    U(H_i, \overline{H}_i)\ket{\bm{z}}_{K^{\hat{n}_{\rm sift}r}} = \ket{\bm{z}(H_i^{\top}\ \overline{H}_i^{\top})}_{K^{\hat{n}_{\rm sift}r}} \label{eq:action_on_Z_basis}
\end{equation}
holds for any $\bm{z}\in\mathbb{F}_p^{\hat{n}_{\rm sift}r}$, which then implies from Eqs.~\eqref{eq:X_op_transform} and \eqref{eq:Z_op_transform} that for any $\bm{x}\in\mathbb{F}_p^{\hat{n}_{\rm sift}r}$,
\begin{equation}
    U(H_i, \overline{H}_i)\ket{\widetilde{\bm{x}}}_{K^{\hat{n}_{\rm sift}r}}= \ket{\widetilde{\bm{x}(\overline{G}_i^{\top}\ G_i^{\top})}}_{K^{\hat{n}_{\rm sift}r}}. \label{eq:action_on_X_basis}
\end{equation}

In a virtual protocol corresponding to the actual protocol mentioned earlier, there remains a quantum system $Q$ after a quantum version of the sifting operation performed at each of $n_{\rm tot}$ rounds in the actual protocol.

\bigskip 
\noindent --- Virtual protocol ---

The functions $K_{\rm EC}(\hat{n}_{\rm sift}, \hat{\Xi}_{\rm pub})$, $K_{\rm PA}(\hat{n}_{\rm sift}, \hat{\Xi}_{\rm pub})$, and $ n_{\rm fin}(\hat{n}_{\rm sift}, \hat{\Xi}_{\rm pub})$ appearing during the protocol is the same as those in the actual protocol.
\begin{enumerate}
    \item Alice prepares an entangled state $\ket{\Psi}_{\sysA\sysC}$ such that there exists a POVM $\{M_{\sysA}^a\}_{a\in{\cal X}_{\sysA}}$ that satisfies 
    \begin{equation}
        \forall a\in{\cal X}_{\sysA},\qquad q_a \rho_a = \tr_{\sysA}[M_{\sysA}^a \ketbra{\Psi}{\Psi}_{\sysA\sysC}], \label{eq:purified_state_compatible}
    \end{equation}
    where $q_a$ and $\{\rho_a\}_{a\in{\cal X}_{\sysA}}$ are as defined in the actual protocol. The system $\sysC$ of $\ket{\Psi}_{\sysA\sysC}$ is sent to a quantum channel and received by Bob as a system $\sysB$. Bob keeps a quantum state on a system $\sysB$.
    Alice and Bob repeat this quantum communication for $n_{\rm tot}$ rounds, where Alice may initiate each round before the previous round is completed.
    \item After $n_{\rm tot}$ rounds of quantum communication, for every $i\in\{1,\ldots,n_{\rm tot}\}$, Alice and Bob jointly access systems $\sysA$ and $\sysB$ at the $i$-th round for obtaining random variables $(\hat{s}^{(i)}, \hat{\xi}_{\rm pub}^{(i)})\in\{0,1\}\times \Omega_{\rm pub}$ while performing a joint operation specified by CP maps ${\cal M}_{\sysA\sysB}^{(\xi)}$ and ${\cal J}_{\sysA\sysB\to Q}^{(\xi)}$ that correspond to $(\hat{s}^{(i)}, \hat{\xi}_{\rm pub}^{(i)})=(0,\xi)$ and $(\hat{s}^{(i)}, \hat{\xi}_{\rm pub}^{(i)})=(1,\xi)$, respectively, which are defined to satisfy:
    \begin{align}
        & \forall \xi \in \Omega_{\rm pub}^{(0)}, \forall\rho_{\sysA\sysB}\in{\cal D}({\cal H}_{\sysA\sysB}),  & {\cal M}_{\sysA\sysB}^{(\xi)}(\rho_{\sysA\sysB}) &=  \tr\!\left[\rho_{\sysA\sysB}\, \sum_{(a,b): \xi=\mathfrak{t}(a,b)}M_{\sysA}^a\otimes M_{\sysB}^b\right], \label{eq:failure_map}\\
        \intertext{and}
        & \forall \xi \in \Omega_{\rm pub}^{(1)}, \forall z\in{\cal X}_{\rm sift}, & \left({\cal J}_{\sysA\sysB\to Q}^{(\xi)}\right)^{\dagger}(\Pi_Q[z]) &= \sum_{(a,b):\xi=\mathfrak{t}(a,b), z=\mathfrak{s}[\xi](a)}M_{\sysA}^a\otimes M_{\sysB}^b, \label{eq:success_map}
    \end{align}
    where $Q$ is a quantum system at the $i$-th round, $\left({\cal J}_{\sysA\sysB\to Q}^{(\xi)}\right)^{\dagger}$ denotes the adjoint map of ${\cal J}_{\sysA\sysB\to Q}^{(\xi)}$, and $\{\Pi_Q[z]\}_{z\in{\cal X}_{\rm sift}}$ denotes a projection-valued measure (PVM).
    (Note that $\sum_{\xi\in\Omega_{\rm pub}^{(0)}}{\cal M}_{\sysA\sysB}^{(\xi)}$ and $\sum_{\xi\in\Omega_{\rm pub}^{(1)}}{\cal J}_{\sysA\sysB\to Q}^{(\xi)}$ combined are trace-preserving.)
    Alice and Bob announce $\hat{\xi}_{\rm pub}^{(i)}$ for each $i=1,\ldots,n_{\rm tot}$. Let $\hat{\Xi}_{\rm pub}\coloneqq n_{\rm tot}P_{\hat{\bm{\xi}}_{\rm pub}}$.
    \item (Permutation symmetrization) Let $\hat{n}_{\rm sift}$ be the number of rounds with $\hat{s}^{(i)}=1$ as in the actual protocol. Alice and Bob randomly permute the $\hat{n}_{\rm sift}$ systems $Q^{\hat{n}_{\rm sift}}$ by consuming ${\cal O}(\hat{n}_{\rm sift}\log\hat{n}_{\rm sift})$-bit local randomness and public communication.  
    \item (Information reconciliation)  Depending on the announcement $\hat{\Xi}_{\rm pub}$ and the length $\hat{n}_{\rm sift}$ of the sifted key, Alice computes $K_{\rm EC}(\hat{n}_{\rm sift}, \hat{\Xi}_{\rm pub})$, $K_{\rm PA}(\hat{n}_{\rm sift}, \hat{\Xi}_{\rm pub})$, and $\hat{n}_{\rm fin}\coloneqq n_{\rm fin}(\hat{n}_{\rm sift}, \hat{\Xi}_{\rm pub})$.
    If $\hat{n}_{\rm fin}=0$, Alice aborts the protocol.  Otherwise, Alice sends Bob $K_{\rm EC}(\hat{n}_{\rm sift}, \hat{\Xi}_{\rm pub})$ random bits.
    \item (Privacy amplification) Alice obtains the $\hat{n}_{\rm sift}r$-qudit system $K^{\hat{n}_{\rm sift}r}$ by performing an isometry $C_{Q\to K^r Q}^{\otimes \hat{n}_{\rm sift}}$ that coherently applies $\mathfrak{e}:{\cal X}_{\rm sift}\hookrightarrow \mathbb{F}_p^r$ given in the actual protocol to the sifted value encoded in the system $Q$ and embeds it into the $Z$ bases of ${\cal H}_K^{\otimes r}$; i.e., $C_{Q\to K^r Q}$ is defined as
    \begin{equation}
        C_{Q\to K^r Q}\coloneqq \sum_{z\in{\cal X}_{\rm sift}} 
        \ket{\mathfrak{e}(z)}_{K^r}\otimes \Pi_{Q}[z]. \label{eq:key_extraction_isometry}
    \end{equation}
    Alice randomly chooses $H_{\hat{i}}\in{\cal H}_d(\hat{n}_{\rm sift}r,\hat{n}_{\rm fin})$, performs a unitary $U(H_{\hat{i}}, \overline{H}_{\hat{i}})$ on the system $K^{\hat{n}_{\rm sift} r}$, and measures the last $(\hat{n}_{\rm sift}r - \hat{n}_{\rm fin})$ qudits in the $X$ basis to obtain an outcome string $\hat{\bm{c}}\in\mathbb{F}_p^{\hat{n}_{\rm sift}r - \hat{n}_{\rm fin}}$. The unmeasured qudits are named as the system $K_A^{\hat{n}_{\rm fin}}$. 
    Alice and Bob then perform a measurement $\{\tilde{Y}^{\hat{n}_{\rm sift}}_{G_{\hat{i}},\hat{\bm{c}}}(\bm{x}^{*})\}_{\bm{x}^{*}\in (\mathbb{F}_p^r)^{\hat{n}_{\rm sift}}:\bm{x}^* G_{\hat{i}}^{\top}=\hat{\bm{c}}}$ on the system $Q^{\hat{n}_{\rm sift}}$ to obtain an estimate $\hat{\bm{x}}^{*}\in(\mathbb{F}_p^r)^{\hat{n}_{\rm sift}}$ and set $\hat{\bm{b}}=\hat{\bm{x}}^{*}\overline{G}_{\hat{i}}^{\top}\in \mathbb{F}_p^{\hat{n}_{\rm fin}}$.  Alice then performs a phase error correction $Z(\hat{\bm{b}})$ on the system $K_A^{\hat{n}_{\rm fin}}$, which leads to the final state $\rho_{K_A^{\hat{n}_{\rm fin}} E}^{\rm virt}$. See also Fig.~\ref{fig:relation_between_hashings}.
\end{enumerate}

\begin{figure}[t]
    \centering 
    \includegraphics[width=0.95\linewidth]{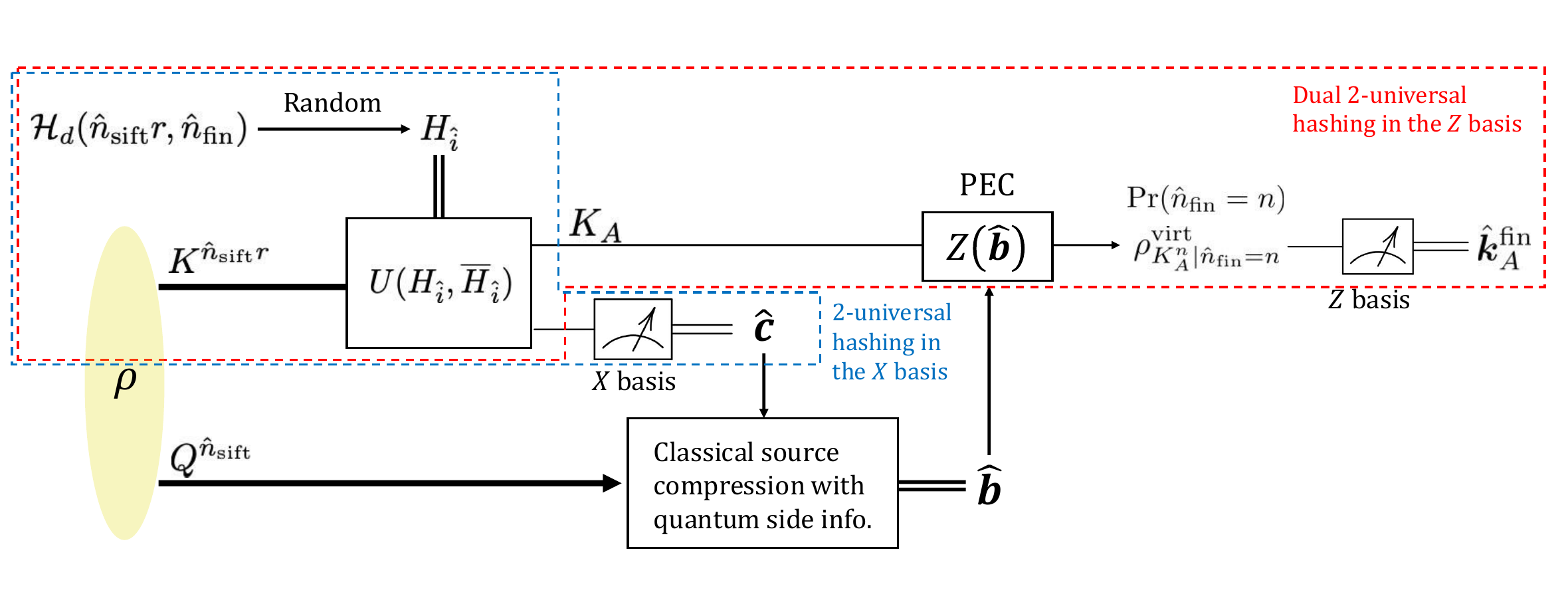}
    \caption{The schematics of how the phase error correction (PEC) protocol works and the state $\rho_{K_A^{n} E | \hat{n}_{\rm fin}=n}^{\rm virt}$ is produced with the probability $\mathrm{Pr}(\hat{n}_{\rm fin}=n)$ in the virtual protocol. The dual 2-universal hashing to obtain the final key $\hat{\bm{k}}_{A}^{\rm fin}$ from the sifted key with the length $\hat{n}_{\rm sift}$ in the actual protocol corresponds to the dual 2-universal hashing in the $Z$ basis in the virtual protocol realized through the unitary $U(H_{\hat{i}},\overline{H}_{\hat{i}})$ acting on the system $K^{\hat{n}_{\rm sift}r}$, where $H_{\hat{i}}$ is chosen randomly from a dual 2-universal family ${\cal H}_d(\hat{n}_{\rm sift}r,\hat{n}_{\rm fin})$.  At the same time, this unitary action followed by the $X$-basis measurement on the last $\hat{n}_{\rm sift}r-\hat{n}_{\rm fin}$ qudit to obtain an outcome string $\hat{\bm{c}}\in\mathbb{F}_p^{\hat{n}_{\rm sift}-\hat{n}_{\rm fin}}$ corresponds to 2-universal hashing in the $X$ basis, which enables to perform the protocol of universal classical source compression with quantum side information in Sec.~\ref{sec:cq_slepian_wolf}.}
    \label{fig:relation_between_hashings}
\end{figure}

For this virtual protocol, we can show Eq.~\eqref{eq:compatibility} by combining Eqs.~\eqref{eq:action_on_Z_basis} and \eqref{eq:purified_state_compatible}--\eqref{eq:success_map} with the fact that the announcements made by Alice and Bob at the post-processing steps (Steps~3--5) in the virtual protocol are the same as those in the actual protocol. For the detailed proof, see Appendix~\ref{sec:equivalence}.
The remaining task is to show another condition~\eqref{eq:phase_error_corrected}, which will be carried out in the subsequent sections.

In addition to the choice $\{\overline{H}_i\}_i$ for the dual 2-universal family $\{H_i\}_i$, the virtual protocol has arbitrariness in the construction of $\{{\cal J}^{(\xi)}_{\sysA\sysB\to Q}\}_{\xi\in\Omega_{\rm pub}^{(1)}}$. How the choice of $\{{\cal J}^{(\xi)}_{\sysA\sysB\to Q}\}_{\xi\in\Omega_{\rm pub}^{(1)}}$ affects the security proof and the overall performance will be revealed in the subsequent sections.

\subsection{From PEC to partially universal source compression through estimation protocol} \label{sec:connection_partially_universal}
To see the explicit relation between the failure probability of PEC given in the form of the left-hand side of Eq.~\eqref{eq:phase_error_corrected} and that of the source compression with quantum side information studied in Sec.~\ref{sec:cq_slepian_wolf}, we introduce yet another protocol called an estimation protocol. An estimation protocol introduces an auxiliary random variable $\hat{\theta}^{(i)}_{\rm aux}\in\Omega_{\rm aux}$, where $\Omega_{\rm aux}$ is a (finite) set. As will be explained later, the introduction of the variable $\hat{\theta}^{(i)}_{\rm aux}$ is often essential for establishing the security. Furthermore, instead of Steps~2 and 5 in the virtual protocol, an estimation protocol performs alternative procedures Step~2' and 5' as follows.

\bigskip 
\noindent --- Estimation protocol ---
\begin{enumerate}
    \item[1.] The same as that in the virtual protocol.
    \item[2'.] After $n_{\rm tot}$ rounds of quantum communication, for every $i\in\{1,\ldots,n_{\rm tot}\}$, Alice and Bob jointly access systems $\sysA$ and $\sysB$ at the $i$-th round for obtaining random variables $(\hat{s}^{(i)}, \hat{\xi}_{\rm pub}^{(i)}, \hat{\theta}_{\rm aux}^{(i)})\in\{0,1\}\times \Omega_{\rm pub} \times \Omega_{\rm aux}$ while performing a joint operation specified by CP maps ${\cal M}_{\sysA\sysB}^{(\xi, \theta)}$ and ${\cal J}_{\sysA\sysB\to Q}^{(\xi, \theta)}$ that correspond to $(\hat{s}^{(i)}, \hat{\xi}_{\rm pub}^{(i)}, \hat{\theta}_{\rm aux}^{(i)})=(0, \xi, \theta)$ and $(\hat{s}^{(i)}, \hat{\xi}_{\rm pub}^{(i)}, \hat{\theta}_{\rm aux}^{(i)})=(1, \xi, \theta)$, respectively, which are defined to satisfy:
    \begin{align}
        & \forall \xi \in \Omega_{\rm pub}^{(0)},  & \sum_{\theta\in\Omega_{\rm aux}}{\cal M}_{\sysA\sysB}^{(\xi, \theta)} &=  {\cal M}_{\sysA\sysB}^{(\xi)}, \label{eq:failure_map_estimation}\\
        \intertext{and}
        & \forall \xi \in \Omega_{\rm pub}^{(1)},  & \sum_{\theta\in\Omega_{\rm aux}} {\cal J}_{\sysA\sysB\to Q}^{(\xi, \theta)} &= {\cal J}_{\sysA\sysB\to Q}^{(\xi)}. \label{eq:success_map_estimation}
    \end{align}
    That is, the maps ${\cal M}_{\sysA\sysB}^{(\xi, \theta)}$ and ${\cal J}_{\sysA\sysB\to Q}^{(\xi, \theta)}$ are subdivisions of the maps ${\cal M}_{\sysA\sysB}^{(\xi)}$ and ${\cal J}_{\sysA\sysB\to Q}^{(\xi)}$ introduced in the virtual protocol. Alice and Bob announce $\hat{\xi}_{\rm pub}^{(i)}$ for each $i=1,\ldots,n_{\rm tot}$. Let $\hat{\Xi}_{\rm pub}\coloneqq n_{\rm tot}P_{\hat{\bm{\xi}}_{\rm pub}}$ and $\hat{\Theta}_{\rm aux}\coloneqq n_{\rm tot}P_{\hat{\bm{\theta}}_{\rm aux}}$, where $\hat{\bm{\theta}}_{\rm aux} \coloneqq (\hat{\theta}_{\rm aux}^{(1)},\ldots,\hat{\theta}_{\rm aux}^{(n_{\rm tot})})$.
    \item[3--4.] The same as those in the virtual protocol.
    \item[5'.] Alice performs the isometry $C_{Q\to K^r Q}^{\otimes \hat{n}_{\rm sift}}$ defined in Eq.~\eqref{eq:key_extraction_isometry}. Alice successively performs $X$-basis measurements on her qudit to obtain $\hat{\bm{x}}$. She randomly chooses $H_{\hat{i}}\in{\cal H}_d(\hat{n}_{\rm sift}r,\hat{n}_{\rm fin})$, which amounts to determining $(G_{\hat{i}}, \overline{G}_{\hat{i}})$. She then computes $\hat{\bm{c}}=\hat{\bm{x}}G_{\hat{i}}^{\top}$. Depending on $\hat{\bm{c}}$, Alice and Bob perform the partially universal decoder $\{\tilde{Y}^{\hat{n}_{\rm sift}}_{G_{\hat{i}},\hat{\bm{c}}}(\bm{x}^*)\}_{\bm{x}^*\in (\mathbb{F}_p^r)^{\hat{n}_{\rm sift}}:\bm{x}^*G_{\hat{i}}^{\top}=\hat{\bm{c}}}$ defined in Eq.~\eqref{eq:partially_universal_decoder} on a quantum state in $Q^{\hat{n}_{\rm sift}}$ to obtain an estimate $\hat{\bm{x}}^*$ of $\hat{\bm{x}}$.  Alice then sets $\hat{\bm{b}}=\hat{\bm{x}}^* \overline{G}_{\hat{i}}^{\top}$ and outputs $\hat{\bm{x}}_{\rm fin}\coloneqq \hat{\bm{x}}\overline{G}_{\hat{i}}^{\top}-\hat{\bm{b}}$.
\end{enumerate}

From Eqs.~\eqref{eq:b_subtraction} and~\eqref{eq:action_on_X_basis}, if a state $U(H_i,\overline{H}_i)\ket{\widetilde{\bm{x}}}$ is measured in the $X$ basis, then it deterministically yields the outcome $(\bm{x}\overline{G}_i^{\top}, \bm{x} G_i^{\top})$. Thus, taking Eqs.~\eqref{eq:failure_map_estimation} and \eqref{eq:success_map_estimation} into account, the probability distribution of the output $\hat{\bm{x}}_{\rm fin}$ of the estimation protocol when $\hat{\Theta}_{\rm aux}$ is ignored is the same as that of the outcome when the state at the end of the virtual protocol is measured in the $X$ basis. Thus, from Eq.~\eqref{eq:phase_error_corrected}, we have 
\begin{align}
    \sum_{n\geq 1}\mathrm{Pr}(\hat{n}_{\rm fin}=n) \left(1 - \bra{\widetilde{\bm{0}}}_{K_A^n} \rho_{K_A^n|\hat{n}_{\rm fin}=n}^{\rm virt} \ket{\widetilde{\bm{0}}}_{K_A^n}\right) &= \probrho_{\rho_{A^{n_{\rm tot}}B^{n_{\rm tot}}}}(\hat{n}_{\rm fin}\geq 1, \hat{\bm{x}}_{\rm fin}\neq \bm{0})\\ 
    &\leq \probrho_{\rho_{A^{n_{\rm tot}}B^{n_{\rm tot}}}}(\hat{n}_{\rm fin}\geq 1, \hat{\bm{x}}\neq\hat{\bm{x}}^*), \label{eq:PEC_bounded_by_cqSW}
\end{align}
where $\probrho_{\rho_{A^{n_{\rm tot}}B^{n_{\rm tot}}}}$ refers to the probability distribution of the random variables in the estimation protocol when the quantum state $\rho_{A^{n_{\rm tot}}B^{n_{\rm tot}}}$ is shared between Alice and Bob at the end of Step~1. The last inequality follows from the definition of $\hat{\bm{x}}_{\rm fin}$.
The procedure of obtaining the estimate $\hat{\bm{x}}^*$ of $\hat{\bm{x}}$ from $\hat{\bm{c}}$ and the remaining quantum state in the system $Q^{\hat{n}_{\rm sift}}$ at Step~5' is the same as the procedure of classical source compression with quantum side information studied in the previous section, and the right-hand side of Eq.~\eqref{eq:PEC_bounded_by_cqSW} corresponds to a failure probability of this estimation procedure. There is, however, a crucial difference in these two scenarios; the underlying quantum state in the estimation protocol is not an i.i.d.~state in general.  Thus, we need a reduction to an i.i.d.~scenario, which is carried out in the subsequent sections.

It is now clear that the success/failure of the PEC is directly tied to the success/failure of source compression with quantum side information. The virtual protocol can thus be interpreted as coherently performing source compression with quantum side information and a corresponding correction operation.
See also Fig.~\ref{fig:relation_between_hashings} for an overview of how the two bases $Z$ and $X$, or equivalently, how the actual and the estimation protocol are related through the virtual protocol.

In the rest of this section, we clarify how the right-hand side of Eq.~\eqref{eq:PEC_bounded_by_cqSW}, and more generally the probability distributions of other random variables appearing in the estimation protocol, depend on the state $\rho_{\sysA^{n_{\rm tot}}\sysB^{n_{\rm tot}}}$ shared between Alice and Bob at the end of Step~1. These analyses will enable us to make a reduction to an i.i.d.~scenario, where we can apply the results of Sec.~\ref{sec:construction_universal}.
Let us first define a set $\Delta_m(\Xi)$ for a frequency $\Xi$ as 
\begin{equation}
    \begin{split}
    \Delta_m(\Xi)&\coloneqq \Bigl\{\bm{\xi}=(\xi_1,\ldots,\xi_{n_{\rm tot}}):(\xi_1,\ldots,\xi_{m})\in(\Omega_{\rm pub}^{(1)})^{\times m}, (\xi_{m+1},\ldots,\xi_{n_{\rm tot}})\in(\Omega_{\rm pub}^{(0)})^{\times (n_{\rm tot} - m)}, \\
    &\hspace{5cm} n_{\rm tot}P_{\bm{\xi}}=\Xi\Bigr\}.
    \end{split}
\end{equation}
Then, the subnormalized state of the system $Q^{n_{\rm tot}}$ at the end of Step~3 of the estimation protocol with $(\hat{n}_{\rm sift},\hat{\Xi}_{\rm pub},\hat{\Theta}_{\rm aux}) = (m, \Xi, \Theta)$ is given in terms of a CP map $\mathfrak{J}_{\sysA^{n_{\rm tot}}\sysB^{n_{\rm tot}}\to Q^m}^{(m,\Xi,\Theta)}$ defined as 
\begin{align}
    \begin{split}
    &\mathfrak{J}_{\sysA^{n_{\rm tot}}\sysB^{n_{\rm tot}}\to Q^m}^{(m,\Xi,\Theta)} (\rho_{\sysA^{n_{\rm tot}}\sysB^{n_{\rm tot}}}) \\
    &= \sum_{\substack{{\cal I}\subseteq[n_{\rm tot}]\\ :|{\cal I}|=m}} \frac{1}{|S_{m}|}\sum_{\sigma\in S_{m}} \sum_{(\xi_1,\ldots,\xi_{n_{\rm tot}})\in \Delta_m(\Xi)}\,\sum_{\substack{(\theta_1,\ldots,\theta_{n_{\rm tot}})\\:n_{\rm tot}P_{\bm{\theta}}=\Theta}}\\
    &\qquad {\cal V}_{\sigma}\circ\left({\cal J}^{(\xi_1,\theta_1)}_{\sysA\sysB\to Q}\otimes \cdots \otimes {\cal J}^{(\xi_m,\theta_m)}_{\sysA\sysB\to Q} \otimes {\cal M}^{(\xi_{m+1},\theta_{m+1})}_{\sysA\sysB} \otimes \cdots \otimes {\cal M}^{(\xi_{n_{\rm tot}},\theta_{n_{\rm tot}})}_{\sysA\sysB}\right)\circ{\cal T}_{{\cal I}\to [m]}(\rho_{\sysA^{n_{\rm tot}}\sysB^{n_{\rm tot}}}),
    \end{split}
    \label{eq:instrument_global}
\end{align}
where $[n]\coloneqq \{1,\ldots,n\}$ is an ordered set, ${\cal T}_{{\cal I}\to [m]}$ denotes a unitary map that rearranges systems in ${\cal I}$ to $[m]$ (i.e., the first $m$ systems) preserving the relative order in both ${\cal I}$ and $[n_{\rm tot}]\setminus{\cal I}$, which thus rearranges $[n_{\rm tot}]\setminus{\cal I}$ to $\{m+1,\ldots,n_{\rm tot}\}$, and ${\cal V}_{\sigma}$ with $\sigma\in S_{m}$ denotes a unitary channel representation of the permutation group $S_{m}$ on the operators on ${\cal H}_{Q}^{\otimes m}$ such that ${\cal V}_{\sigma}\circ{\cal V}_{\sigma'}={\cal V}_{\sigma\circ\sigma'}$ for any $\sigma,\sigma'\in S_{m}$. Note that ${\cal V}_{\sigma}$ reflects the random permutation performed on the system $Q^m$ at Step~3 in the estimation protocol. 

As long as $\hat{n}_{\rm fin}\geq 1$ is obtained at Step~4 of the estimation protocol, a quantum state of the system $Q^m$ is further mapped by the channel ${\cal C}^{\otimes m}_{Q\to K^rQ}$ at Step~5' of the estimation protocol, where ${\cal C}_{Q\to K^rQ}$ denotes the action of an isometry $C_{Q\to K^rQ}$, as given in Eq.~\eqref{eq:key_extraction_isometry}, and then each system $K^{r}$ is measured in the $X$ basis. The resulting subnormalized state $\mathfrak{X}_{X^{m}Q^{m}}^{(m,\Xi,\Theta)}(\rho_{\sysA^{n_{\rm tot}}\sysB^{n_{\rm tot}}})$ is given by 
\begin{align}
    \mathfrak{X}_{X^{m}Q^{m}}^{(m,\Xi,\Theta)}(\rho_{\sysA^{n_{\rm tot}}\sysB^{n_{\rm tot}}})
    \coloneqq {\cal P}_{K^r\to X}^{\otimes m}\circ {\cal C}^{\otimes m}_{Q\to K^rQ}\circ\mathfrak{J}_{\sysA^{n_{\rm tot}}\sysB^{n_{\rm tot}}\to Q^m}^{(m,\Xi,\Theta)}(\rho_{\sysA^{n_{\rm tot}}\sysB^{n_{\rm tot}}}), \label{eq:def_of_projected_state}
\end{align}
where the pinching channel ${\cal P}_{K^r\to X}$ is defined as 
\begin{equation}
    {\cal P}_{K^r\to X}(\rho_{K^r})=\sum_{x\in \mathbb{F}_p^r}\ket{\widetilde{x}}_X\bra{\widetilde{x}}_{K^r}\rho_{K^r}\ket{\widetilde{x}}_{K^r}\bra{\widetilde{x}}_{X}. \label{eq:pinching_channel}
\end{equation}
We can rewrite $\mathfrak{X}_{X^{m}Q^{m}}^{(m,\Xi,\Theta)}(\rho_{\sysA^{n_{\rm tot}}\sysB^{n_{\rm tot}}})$ by using Eqs.~\eqref{eq:generalized_Z}, \eqref{eq:basis_transform_n_dit}, \eqref{eq:key_extraction_isometry}, and \eqref{eq:pinching_channel} as
\begin{equation}
    \begin{split}
    \mathfrak{X}_{X^{m}Q^{m}}^{(m,\Xi,\Theta)}(\rho_{\sysA^{n_{\rm tot}}\sysB^{n_{\rm tot}}})
    &= \sum_{\bm{x}\in (\mathbb{F}_p^r)^{m}}\! p^{-mr} \ketbra{\widetilde{\bm{x}}}{\widetilde{\bm{x}}}_{X^{m}} \otimes \mathfrak{Z}_{Q^m}(\bm{x})\,\mathfrak{J}_{\sysA^{n_{\rm tot}}\sysB^{n_{\rm tot}}\to Q^m}^{(m,\Xi,\Theta)}(\rho_{\sysA^{n_{\rm tot}}\sysB^{n_{\rm tot}}})\,\mathfrak{Z}_{Q^m}(\bm{x})^{\dagger}, 
    \end{split}\label{eq:rewritten_by_X}
\end{equation}
where the unitary $\mathfrak{Z}_{Q^m}(\bm{x})$ for $\bm{x}=(x_1,\ldots,x_m)$ with $x_i\in \mathbb{F}_p^r$ for $i=1,\ldots,m$ is defined as 
\begin{align}
    \mathfrak{Z}_{Q^m}(\bm{x})&\coloneqq \mathfrak{Z}'_{Q}(x_{1})\otimes\cdots\otimes \mathfrak{Z}'_{Q}(x_m), \label{eq:Z_list}\\
    \forall x\in \mathbb{F}_p^r, \qquad \mathfrak{Z}'_{Q}(x)&\coloneqq p^{r/2}\bra{\widetilde{x}}_{K^r}C_{Q\to K^rQ} \\
    & = \sum_{z\in{\cal X}_{\rm sift}}p^{r/2}\braket{\widetilde{x}|\mathfrak{e}(z)}\Pi_Q[z], 
    \label{eq:Z_restriction}
\end{align}
where the phase factor $p^{r/2}\braket{\widetilde{x}|\mathfrak{e}(z)}$ in Eq.~\eqref{eq:Z_restriction} can be given by Eq.~\eqref{eq:basis_transform_n_dit}.
Notice that we have 
\begin{equation}
    \tr_{Q^m}\left[\mathfrak{X}_{X^{m}Q^{m}}^{(m,\Xi,\Theta)}(\rho_{\sysA^{n_{\rm tot}}\sysB^{n_{\rm tot}}})\right] = p^{-mr}\,\tr\left[\mathfrak{J}_{\sysA^{n_{\rm tot}}\sysB^{n_{\rm tot}}\to Q^m}^{(m,\Xi,\Theta)}(\rho_{\sysA^{n_{\rm tot}}\sysB^{n_{\rm tot}}})\right] I_{X^{m}}, \label{eq:uniform_marginal}
\end{equation}
which is why we can use the partially universal decoder in the following analysis.

Now, the probability distributions of the random variables in the estimation protocol can be related to the subnormalized state $\mathfrak{X}_{X^{m}Q^{m}}^{(m,\Xi,\Theta)}(\rho_{\sysA^{n_{\rm tot}}\sysB^{n_{\rm tot}}})$ as follows. By definition, we have 
\begin{equation}
    \begin{split}
    \probrho_{\rho_{A^{n_{\rm tot}}B^{n_{\rm tot}}}}(\hat{n}_{\rm sift}=m, \hat{\Xi}_{\rm pub}=\Xi,\hat{\Theta}_{\rm aux}=\Theta) = \tr\!\left[\mathfrak{X}_{X^{m}Q^{m}}^{(m,\Xi,\Theta)}(\rho_{\sysA^{n_{\rm tot}}\sysB^{n_{\rm tot}}}) \right].
    \end{split} \label{eq:event_prob_estimation}
\end{equation} 
As for the failure of the partially universal decoder $\{\tilde{Y}^{\hat{n}_{\rm sift}}_{G_{\hat{i}},\hat{\bm{c}}}(\bm{x}^*)\}_{\bm{x}^*\in(\mathbb{F}_p^r)^{\hat{n}_{\rm sift}}:\bm{x}^*G_{\hat{i}}^{\top}=\hat{\bm{c}}}$, we have, for any $n\geq 1$,
\begin{equation}
    \begin{split}
    &\probrho_{\rho_{A^{n_{\rm tot}}B^{n_{\rm tot}}}}(\hat{n}_{\rm sift}=m, \hat{\Xi}_{\rm pub}=\Xi,\hat{\Theta}_{\rm aux}=\Theta, \hat{n}_{\rm fin}=n, \hat{\bm{x}}\neq\hat{\bm{x}}^*) \\
    &= \bm{1}(n_{\rm fin}(m,\Xi)=n)\, \tr\!\left[M[m,\Xi] \,\mathfrak{X}_{X^{mr}Q^{m}}^{(m,\Xi,\Theta)}(\rho_{\sysA^{n_{\rm tot}}\sysB^{n_{\rm tot}}}) \right],
    \end{split} \label{eq:failure_estimation_protocol}
\end{equation} 
where $\bm{1}(\cdot)$ is defined in Eq.~\eqref{eq:def_indicator}, and $M[m,\Xi]$ is defined as
\begin{align}
    M[m,\Xi] &\coloneqq \mathbb{E}_{\hat{i}\sim |{\cal H}_d(m r, n_{\rm fin}(m,\Xi))|^{-1}}\left[\sum_{\bm{x},\bm{x}^*\in (\mathbb{F}_p^r)^{m}} \bm{1}(\bm{x}\neq\bm{x}^*)\ketbra{\bm{x}}{\bm{x}}_{X^{m}} \otimes \tilde{Y}^{m}_{G_{\hat{i}},\bm{x}G_{\hat{i}}^{\top}}(\bm{x}^*)\right] \\
    &=\mathbb{E}_{\hat{i}\sim |{\cal H}_d(m r, n_{\rm fin}(m,\Xi))|^{-1}}\left[I_{X^{m}Q^m} - \sum_{\bm{x}\in(\mathbb{F}_p^r)^{m}} \ketbra{\bm{x}}{\bm{x}}_{X^{m}} \otimes \tilde{Y}^{m}_{G_{\hat{i}},\bm{x}G_{\hat{i}}^{\top}}(\bm{x})\right]. \label{eq:def_M_m_xi}
\end{align}
Equations~\eqref{eq:event_prob_estimation} and~\eqref{eq:failure_estimation_protocol} imply that the statistical behavior of random variables appeared in the estimation protocol, $(\hat{n}_{\rm sift}, \hat{\Xi}_{\rm pub},\hat{\Theta}_{\rm aux}, \hat{n}_{\rm fin}, \hat{\bm{x}},\hat{\bm{x}}^*)$, is fully characterized by the collection of subnormalized states $\{\mathfrak{X}_{X^{m}Q^{m}}^{(m,\Xi,\Theta)}(\rho_{\sysA^{n_{\rm tot}}\sysB^{n_{\rm tot}}})\}_{(m,\Xi,\Theta)}$.

\subsection{Reduction to collective attacks} \label{sec:iid_reduction}

To obtain an upper bound on the failure probability of Eq.~\eqref{eq:PEC_bounded_by_cqSW}, we apply the technique of reducing the analysis on the permutation-symmetric global quantum states to that of a mixture of i.i.d.~states, which is so-called quantum de Finetti reduction~\cite{Tamaki2003, Christandl2009, Fawzi2015, Matsuura2024, Nahar2024}. There is a caveat for this technique: it only gives us statistical statements concerning $(\hat{n}_{\rm sift}, \hat{\Xi}_{\rm pub},\hat{\Theta}_{\rm aux}, \hat{n}_{\rm fin}, \hat{\bm{x}},\hat{\bm{x}}^*)$ for `any' i.i.d.~state, which unfortunately cannot directly restrict the form of a single-round density operator characterizing an i.i.d.~state. 
The random variables $\{\hat{\theta}_{\rm aux}^{(i)}\}_{i=1}^{n_{\rm tot}}$ in the estimation protocol play an important role here.  They enable us to use restrictions on a single-round density operator for the system $\sysA\sysB$, which typically stems from the nonorthogonality among the states transmitted from Alice in the actual protocol. That is, $\hat{\theta}_{\rm aux}^{(i)}$ is chosen such that it is sensitive to the restriction of a single-round density operator stemming from the protocol itself and the protocol parameters. Then, through the statistics of the frequency $\hat{\Theta}_{\rm aux}$, we can indirectly reflect the restriction on the form of the single-round density operator. 
More explicitly, by exploiting the restriction from the protocol, we may find a subset $\Upsilon_{\epsilon_t}$ of the range of $(\hat{\Xi}_{\rm pub}, \hat{\Theta}_{\rm aux})$ such that 
\begin{align}
    \probrho_{\rho_{A^{n_{\rm tot}}B^{n_{\rm tot}}}}\left((\hat{\Xi}_{\rm pub}, \hat{\Theta}_{\rm aux})\notin \Upsilon_{\epsilon_t}\right) \leq \epsilon_t \label{eq:bound_out_of_Upsilon}
\end{align}
should hold regardless of Eve's attack strategy that determines $\rho_{A^{n_{\rm tot}}B^{n_{\rm tot}}}$.
A simple example that we can exploit to have such a nontrivial $\Upsilon_{\epsilon_t}$ is to use the relation $\tr_{\sysB^{n_{\rm tot}}}[\rho_{\sysA^{n_{\rm tot}}\sysB^{n_{\rm tot}}}] = \bigl(\tr_C[\ketbra{\Psi}{\Psi}_{\sysA\sysC}]\bigr)^{\otimes n_{\rm tot}}$ from the definition of $\ket{\Psi}_{\sysA\sysC}$ in the virtual protocol. By enforcing $\hat{\theta}_{\rm aux}^{(i)}$ to depend only on Alice's marginal state on the system $\sysA$, we can write, for any $\theta\in\Omega_{\rm aux}$,
\begin{equation}
\sum_{\xi\in\Omega_{\rm pub}^{(0)}}{\cal M}^{(\xi,\theta)}_{\sysA\sysB}(\cdot) + \sum_{\xi\in\Omega_{\rm pub}^{(1)}}\tr\bigl[{\cal J}^{(\xi,\theta)}_{\sysA\sysB\to Q}(\cdot)\bigr] = {\cal L}^{(\theta)}_{\sysA}\circ\tr_{\sysB}(\cdot),
\end{equation}
where ${\cal L}^{(\theta)}_{\sysA}:{\cal D}({\cal H}_{\sysA})\to[0,1]$ is a CP map.
Then, $\hat{\Theta}_{\rm aux}$ obeys a multinomial distribution with underlying probabilities given by $\{{\cal L}_A^{(\theta)}(\Psi_{\sysA})\}_{\theta\in\Omega_{\rm aux}}$, where $\Psi_{\sysA}=\tr_{\sysC}[\ketbra{\Psi}{\Psi}_{\sysA\sysC}]$. 
In this example, $\Upsilon_{\epsilon_t}$ is a nontrivial subset only for $\hat{\Theta}_{\rm aux}$.

Equation~\eqref{eq:bound_out_of_Upsilon} can be used to bound the right-hand side of Eq.~\eqref{eq:PEC_bounded_by_cqSW}. Since it follows that $\probrho_{\rho_{\sysA^{n_{\rm tot}}\sysB^{n_{\rm tot}}}}\left((\hat{\Xi}_{\rm pub}, \hat{\Theta}_{\rm aux})\notin \Upsilon_{\epsilon_t}, \hat{n}_{\rm fin}\geq 1, \hat{\bm{x}}\neq\hat{\bm{x}}^*\right) \leq \epsilon_t$, we have 
\begin{equation}
    \probrho_{\rho_{A^{n_{\rm tot}}B^{n_{\rm tot}}}}(\hat{n}_{\rm fin}\geq 1, \hat{\bm{x}}\neq\hat{\bm{x}}^*) \leq \epsilon_t + \probrho_{\rho_{\sysA^{n_{\rm tot}}\sysB^{n_{\rm tot}}}}\left((\hat{\Xi}_{\rm pub}, \hat{\Theta}_{\rm aux})\in \Upsilon_{\epsilon_t}, \hat{n}_{\rm fin}\geq 1, \hat{\bm{x}}\neq\hat{\bm{x}}^*\right). \label{eq:failure_phase_error_correction}
\end{equation}
Thus, our remaining task is to establish an upper bound on the right-most term in the above inequality.

Next, we invoke the permutation symmetry of the estimation protocol.
We can show that the statistical behavior of random variables $(\hat{n}_{\rm sift}, \hat{\Xi}_{\rm pub},\hat{\Theta}_{\rm aux}, \hat{n}_{\rm fin}, \hat{\bm{x}}, \hat{\bm{x}}^*)$ determined by a global state $\rho_{\sysA^{n_{\rm tot}} \sysB^{n_{\rm tot}}}$ is equal to that by a permutation-symmetrized version of $\rho_{\sysA^{n_{\rm tot}} \sysB^{n_{\rm tot}}}$. This property stems from the uniformity of the sifting process and the random permutation of the sifted rounds at Step~3. It is formally stated as 
\begin{equation}
    \mathfrak{X}_{X^{m} Q^m}^{(m,\Xi,\Theta)} (\rho_{\sysA^{n_{\rm tot}}\sysB^{n_{\rm tot}}}) = \mathfrak{X}_{X^{m} Q^m}^{(m,\Xi,\Theta)}\left({\cal R}_{\sysA^{n_{\rm tot}} \sysB^{n_{\rm tot}}}(\rho_{\sysA^{n_{\rm tot}}\sysB^{n_{\rm tot}}})\right), \label{eq:insert_random_permutation}
\end{equation}
where ${\cal R}_{\sysA^{n_{\rm tot}}}$ denotes the CPTP map for the random permutation of the $n_{\rm tot}$ systems. The proof of Eq.~\eqref{eq:insert_random_permutation} is given in Appendix~\ref{sec:permutation_symmetry}.

Now, we proceed to bound the right-most term of Eq.~\eqref{eq:failure_phase_error_correction}. From Eqs.~\eqref{eq:failure_estimation_protocol} and~\eqref{eq:insert_random_permutation}, we have 
\begin{equation}
    \begin{split}
    &\probrho_{\rho_{\sysA^{n_{\rm tot}}\sysB^{n_{\rm tot}}}}\left((\hat{\Xi}_{\rm pub}, \hat{\Theta}_{\rm aux})\in \Upsilon_{\epsilon_t}, \hat{n}_{\rm fin}\geq 1, \hat{\bm{x}}\neq\hat{\bm{x}}^*\right) \\
    &= \probrho_{{\cal R}_{\sysA^{n_{\rm tot}} \sysB^{n_{\rm tot}}}(\rho_{\sysA^{n_{\rm tot}}\sysB^{n_{\rm tot}}})}\left((\hat{\Xi}_{\rm pub}, \hat{\Theta}_{\rm aux})\in \Upsilon_{\epsilon_t}, \hat{n}_{\rm fin}\geq 1, \hat{\bm{x}}\neq\hat{\bm{x}}^*\right).
    \end{split} \label{eq:replacement_by_permutation_symmetrized}
\end{equation}
For the right-hand side, we can apply a de-Finetti-type inequality, Lemma~1 and Theorem~1 in Ref.~\cite{Matsuura2024}, to obtain 
\begin{align}
    \begin{split}
    &\probrho_{{\cal R}_{\sysA^{n_{\rm tot}} \sysB^{n_{\rm tot}}}(\rho_{\sysA^{n_{\rm tot}}\sysB^{n_{\rm tot}}})}\!\left((\hat{\Xi}_{\rm pub}, \hat{\Theta}_{\rm aux})\in \Upsilon_{\epsilon_t}, \hat{n}_{\rm fin}\geq 1, \hat{\bm{x}}\neq\hat{\bm{x}}^*\right) \\
    &\leq f_q(n_{\rm tot},d_{\sysA\sysB})\!\max_{\rho_{\sysA\sysB} \in {\cal D}({\cal H}_{\sysA\sysB})} \probrho_{\rho_{\sysA\sysB}^{\otimes n_{\rm tot}}}\!\left((\hat{\Xi}_{\rm pub}, \hat{\Theta}_{\rm aux})\in \Upsilon_{\epsilon_t}, \hat{n}_{\rm fin}\geq 1, \hat{\bm{x}}\neq\hat{\bm{x}}^*\right) 
    \end{split}\label{eq:bound_by_iid}
\end{align} 
with $d_{\sysA\sysB}\coloneqq \dim({\cal H}_{\sysA\sysB})$, where
\begin{equation}
    f_q(n,d)\coloneqq \frac{(n + d - 1)^{\frac{d^2 - 1}{2}}}{\sqrt{2\pi(d/e^2)^d}\prod_{i=0}^{d-1}i!}, \label{eq:def_f_q}
\end{equation}
Thus, if we could show
\begin{equation}
    \max_{\rho_{\sysA\sysB} \in {\cal D}({\cal H}_{\sysA\sysB})} \probrho_{\rho_{\sysA\sysB}^{\otimes n}}\!\left((\hat{\Xi}_{\rm pub}, \hat{\Theta}_{\rm aux})\in \Upsilon_{\epsilon_t}, \hat{n}_{\rm fin}\geq 1, \hat{\bm{x}}\neq\hat{\bm{x}}^*\right)
    \leq \epsilon_{\rm iid}, \label{eq:failure_prob_collective}
\end{equation}
then combining this with Eqs.~\eqref{eq:PEC_bounded_by_cqSW}, \eqref{eq:bound_out_of_Upsilon}, \eqref{eq:replacement_by_permutation_symmetrized}, and~\eqref{eq:bound_by_iid} would imply Eq.~\eqref{eq:phase_error_corrected} with 
\begin{equation}
    \frac{\varepsilon_{\rm sec}^2}{2}=\epsilon_t + \epsilon_{\rm iid}f_q(n_{\rm tot},d_{\sysA\sysB}). \label{eq:correspondence_params}
\end{equation}
This effectively reduces the PEC against general attacks to that against collective attacks at the cost of polynomial overhead on the secrecy parameter. As mentioned at the beginning of this section, the restriction from the protocol itself and the protocol parameters is indirectly reflected through the condition $(\hat{\Xi}_{\rm pub}, \hat{\Theta}_{\rm aux})\in \Upsilon_{\epsilon_t}$.

The necessity of a universal decoder becomes clear at this point. Since the state that achieves the maximum in Eq.~\eqref{eq:failure_prob_collective} is unknown and not uniquely determined by observables in the protocol, Alice and Bob need to construct a decoder that does not depend on the state itself.  To put it differently, we need to upper-bound the worst-case failure probability for a family of fixed decoding strategies $\{M[m, \Xi]\}_{m, \Xi}$ against any i.i.d.~states.  The universal decoder developed in the previous section tells us that the worst-case failure probability is determined solely by the largest conditional R\'enyi entropy over the set of possible states and nothing else.

\begin{remark}
As mentioned earlier, the random variable $\hat{\Theta}_{\rm aux}$ is introduced so that the properties determined solely by the protocol setup, such as the intensity of a light pulse Alice emits, are inherited by i.i.d.~states that upper-bound the given state in the estimation protocol as in Eq.~\eqref{eq:bound_by_iid}~\cite{Zhou2022}.  In the conventional post-selection technique, this has been achieved by showing that a permutation-invariant state with a fixed i.i.d.~marginal can be bounded from above by i.i.d.~quantum states with the same i.i.d.~marginal~\cite{Fawzi2015, Nahar2024}. 
This technique can also be applied to our setup instead of introducing $\hat{\Theta}_{\rm aux}$ for this.  However, the use of the auxiliary random variable $\hat{\Theta}_{\rm aux}$ may be more flexible than relying on the i.i.d.~marginal. For example, a light pulse Alice emits need not be i.i.d.~and can even be correlated; one may still have a nontrivial subset $\Upsilon_{\epsilon_t}$ to satisfy Eq.~\eqref{eq:bound_out_of_Upsilon} in this scenario.  Furthermore, a partial knowledge of Alice's marginal, which may be the case when a light source has imperfections, may suffice to obtain a meaningful subset $\Upsilon_{\epsilon_t}$. For some protocols, there may be a degree of freedom that cannot be controlled by Eve, given a particular realization $\Xi$ of the observable $\hat{\Xi}_{\rm pub}$, which can also be incorporated as a subset $\Upsilon_{\epsilon_t}$.
\end{remark}

\subsection{Estimation of the failure probability} \label{sec:estimation_failure_prob}
In this section, we develop a strategy to obtain an upper bound $\epsilon_{\rm iid}$ on the failure probability as given in Eq.~\eqref{eq:failure_prob_collective}.
First, recall that the random variables $\hat{\Xi}_{\rm pub}$ and $\hat{\Theta}_{\rm aux}$ are the frequencies of $\hat{\bm{\xi}}_{\rm pub}$ and $\hat{\bm{\theta}}_{\rm aux}$, respectively.  
From the definition of the CP maps ${\cal M}_{\sysA\sysB}^{(\theta,\xi)}$ and ${\cal J}_{\sysA\sysB\to Q}^{(\theta,\xi)}$ as in Eqs.~\eqref{eq:failure_map_estimation} and~\eqref{eq:success_map_estimation}, when an i.i.d.~state $\rho_{\sysA\sysB}^{\otimes n_{\rm tot}}$ is input, then the pair $(\hat{\xi}_{\rm pub}^{(i)}, \hat{\theta}_{\rm aux}^{(i)})$ of random variables obeys a joint probability distribution $P_{\rho}$ common to all $i\in\{1,\ldots,n_{\rm tot}\}$, which is defined as 
\begin{equation}
    \forall\xi\in\Omega_{\rm pub}, \forall\theta\in\Omega_{\rm aux},\qquad P_{\rho_{\sysA\sysB}}(\hat{\xi}_{\rm pub}^{(i)}=\xi, \hat{\theta}_{\rm aux}^{(i)}=\theta)\coloneqq \begin{cases} {\cal M}^{(\xi,\theta)}_{\sysA\sysB}(\rho_{\sysA\sysB}), & \text{if } \xi\in\Omega_{\rm pub}^{(0)},\\
        \tr [{\cal J}^{(\xi,\theta)}_{\sysA\sysB\to Q}(\rho_{\sysA\sysB})], & \text{if } \xi\in\Omega_{\rm pub}^{(1)}.
    \end{cases} \label{eq:iid_prob_dist}
\end{equation} 
Thus, $(\hat{\Xi}_{\rm pub}, \hat{\Theta}_{\rm aux})$ obeys a multinomial distribution given through the above probabilities, i.e., 
\begin{equation}
    \probrho_{\rho_{\sysA\sysB}^{\otimes n_{\rm tot}}}\!\left(\hat{\Xi}_{\rm pub}=\Xi, \hat{\Theta}_{\rm aux}=\Theta\right) = P_{\rho_{\sysA\sysB}}^{\times n_{\rm tot}}\!\left(\hat{\Xi}_{\rm pub}=\Xi, \hat{\Theta}_{\rm aux}=\Theta\right).
    \label{eq:xi_theta_follow_iid}
\end{equation}
Then, by applying concentration inequalities for i.i.d.~random variables, we can construct a family $\{U_{\epsilon_v}(\rho_{\sysA\sysB})\}_{\rho_{\sysA\sysB}\in{\cal D}({\cal H}_{\sysA\sysB})}$ of subsets of the range of $(\hat{\Xi}_{\rm pub},\hat{\Theta}_{\rm aux})$ that satisfy
\begin{equation}
    \forall \rho_{\sysA\sysB}\in{\cal D}({\cal H}_{\sysA\sysB}),\qquad  P_{\rho_{\sysA\sysB}}^{\times n_{\rm tot}}\left((\hat{\Xi}_{\rm pub}, \hat{\Theta}_{\rm aux})\notin U_{\epsilon_v}(\rho_{\sysA\sysB})\right) \leq \epsilon_v. \label{eq:ensured_by_concentration}
\end{equation} 
Combining this with the constraint $(\hat{\Xi}_{\rm pub}, \hat{\Theta}_{\rm aux})\in \Upsilon_{\epsilon_t}$ reflecting protocol setups, choose a family $\{V_{\epsilon_t,\epsilon_v}(\Xi)\}_{\Xi}$ of subsets of density operators so that they satisfy
\begin{equation}
    V_{\epsilon_t,\epsilon_v}(\Xi)\supseteq \bigcup_{\Theta:(\Xi,\Theta)\in\Upsilon_{\epsilon_t}}\left\{\rho_{\sysA\sysB}\in{\cal D}({\cal H}_{\sysA\sysB}):(\Xi,\Theta)\in U_{\epsilon_v}(\rho_{\sysA\sysB})\right\}. \label{eq:hull}
\end{equation}
Then, it follows that $(\Xi, \Theta)\in\Upsilon_{\epsilon_t}\cap U_{\epsilon_v}(\rho_{\sysA\sysB})$ implies $\rho_{\sysA\sysB}\in V_{\epsilon_t,\epsilon_v}(\Xi)$.
See also Remark~\ref{rem:optimization_region} and Sec.~\ref{sec:conventional_PEC} for an explicit example of the construction of $\{V_{\epsilon_t,\epsilon_v}(\Xi)\}_{\Xi}$.

Combining Eqs.~\eqref{eq:final_key_as_function}, \eqref{eq:xi_theta_follow_iid}, \eqref{eq:ensured_by_concentration}, and~\eqref{eq:hull}, we have an upper bound on the left-hand side of Eq.~\eqref{eq:failure_prob_collective} as 
\begin{align}
        &\max_{\rho_{\sysA\sysB} \in {\cal D}({\cal H}_{\sysA\sysB})} \probrho_{\rho_{\sysA\sysB}^{\otimes n_{\rm tot}}}\!\left((\hat{\Xi}_{\rm pub}, \hat{\Theta}_{\rm aux})\in \Upsilon_{\epsilon_t}, \hat{n}_{\rm fin}\geq 1, \hat{\bm{x}}\neq\hat{\bm{x}}^*\right)\nonumber\\
        &\leq \epsilon_v + \max_{\rho_{\sysA\sysB} \in {\cal D}({\cal H}_{\sysA\sysB})} \probrho_{\rho_{\sysA\sysB}^{\otimes n_{\rm tot}}}\!\left((\hat{\Xi}_{\rm pub}, \hat{\Theta}_{\rm aux})\in \Upsilon_{\epsilon_t}\cap U_{\epsilon_v}(\rho_{\sysA\sysB}), \hat{n}_{\rm fin}\geq 1, \hat{\bm{x}}\neq\hat{\bm{x}}^*\right) \\
        & \leq \epsilon_v + \max_{\rho_{\sysA\sysB} \in {\cal D}({\cal H}_{\sysA\sysB})} \probrho_{\rho_{\sysA\sysB}^{\otimes n_{\rm tot}}}\!\left( (\hat{n}_{\rm sift},\hat{\Xi}_{\rm pub})\in{\cal W}(\rho_{\sysA\sysB}), \hat{\bm{x}}\neq\hat{\bm{x}}^*\right), \label{eq:connection_to_source_compression}
\end{align}
where ${\cal W}(\rho_{\sysA\sysB})$ is defined as 
\begin{equation}
    {\cal W}(\rho_{\sysA\sysB}) \coloneqq \left\{(m,\Xi):\rho_{\sysA\sysB}\in V_{\epsilon_t,\epsilon_v}(\Xi),  n_{\rm fin}(m,\Xi)\geq 1\right\}, \label{eq:region_W}
\end{equation}
and the last inequality follows from the fact that $(\Xi, \Theta)\in\Upsilon_{\epsilon_t}\cap U_{\epsilon_v}(\rho_{\sysA\sysB})$ implies $\rho_{\sysA\sysB}\in V_{\epsilon_t,\epsilon_v}(\Xi)$ for any $\rho_{\sysA\sysB}\in{\cal D}({\cal H}_{\sysA\sysB})$.

Let us define $\mathfrak{x}_{XQ}(\rho_{\sysA\sysB})$ as 
\begin{align}
    \mathfrak{x}_{XQ}(\rho_{\sysA\sysB}) &\coloneqq \sum_{\xi\in\Omega_{\rm pub}^{(1)}}{\cal P}_{K^r\to X}\circ {\cal C}_{Q\to K^rQ}\circ {\cal J}_{\sysA\sysB\to Q}^{(\xi)}(\rho_{\sysA\sysB}) \label{eq:def_frak_x} \\
    &=\sum_{(\xi,\theta)\in\Omega_{\rm pub}^{(1)}\times\Omega_{\rm aux}}{\cal P}_{K^r\to X}\circ {\cal C}_{Q\to K^rQ}\circ {\cal J}_{\sysA\sysB\to Q}^{(\xi,\theta)}(\rho_{\sysA\sysB}) \label{eq:another_def_frak_x}\\
    &= \sum_{x\in\mathbb{F}_p^r} p^{-r}\ketbra{\widetilde{x}}{\widetilde{x}}_{X} \otimes \sum_{(\xi,\theta)\in\Omega_{\rm pub}^{(1)}\times\Omega_{\rm aux}}\mathfrak{Z}'_Q(x)\, {\cal J}_{\sysA\sysB\to Q}^{(\xi,\theta)}(\rho_{\sysA\sysB})\, \mathfrak{Z}'_Q(x)^{\dagger}, 
\end{align}
where the pinching channel ${\cal P}_{K^r\to X}$ is defined in Eq.~\eqref{eq:pinching_channel}, and $\mathfrak{Z}'_Q(x)$ is defined in Eq.~\eqref{eq:Z_restriction}. 
Then, for a constant $\epsilon_p>0$, the function $K_{\rm PA}(m, \Xi)$ adopted in the actual protocol and also used at Step~4 in the estimation protocol is set to be 
\begin{align}
    \begin{split}
    K_{\rm PA}(m, \Xi) &=  \biggl\lceil\frac{n_{\rm tot}}{\log p} \max_{\sigma_{\sysA\sysB} \in V_{\epsilon_t,\epsilon_v}(\Xi)}  H^{\uparrow,\leq}_{1-\alpha^*}(X|Q)_{\mathfrak{x}_{XQ}(\sigma_{\sysA\sysB})} \\
    &\hspace{4cm} + \frac{p^r(d_{Q}+2)(d_{Q}-1)\log_p(m+1)}{2} + \frac{\log_p(1/\epsilon_p)}{\alpha^*}\biggr\rceil, 
    \end{split} \label{eq:choice_of_final_key} 
\end{align} 
where $d_{Q}\coloneqq \dim({\cal H}_{Q})$, and $\alpha^{*}\in[0,1]$ is arbitrary but should be chosen to minimize the right-hand side to obtain a better key rate. This amounts to determining the set ${\cal W}(\rho_{\sysA\sysB})$ explicitly through the determination of $n_{\rm fin}(m,\Xi)$ as in Eq.~\eqref{eq:final_key_as_function}. 
We will see in the following that, through the use of Proposition~\ref{prop:modified_for_qkd} in Sec.~\ref{sec:construction_universal}, this choice ensures that the right-most term in Eq.~\eqref{eq:connection_to_source_compression} is bounded from above by $\epsilon_p$.

To see this, first recall the definition of $P_{\rm err}$ in Eq.~\eqref{eq:defining_rel_err}. From Eqs.~\eqref{eq:rewritten_by_X}, \eqref{eq:event_prob_estimation}--\eqref{eq:def_M_m_xi}, we have 
\begin{equation}
    \begin{split}
    &\probrho_{\rho_{\sysA\sysB}^{\otimes n_{\rm tot}}}\!\left((\hat{n}_{\rm sift},\hat{\Xi}_{\rm pub})\in{\cal W}(\rho_{\sysA\sysB}), \hat{\bm{x}}\neq\hat{\bm{x}}^* \right) \\
    &= P_{\rm err}\left(\probrho_{\rho_{\sysA\sysB}^{\otimes n_{\rm tot}}}, \left\{\left(\left\{\rho_{Q^m}^{\bm{x},\Xi}\right\}_{\bm{x}\in(\mathbb{F}_p^r)^m}\!\!,{\cal H}(mr, m r-n_{\rm fin}(m, \Xi))\right)\right\}_{(m,\Xi)}\!\!,{\cal W}(\rho_{\sysA\sysB})\right),
    \end{split} \label{eq:reinterpretation_prop}
\end{equation}
where we reinterpret ${\cal X}$ in Eq.~\eqref{eq:defining_rel_err} as $\mathbb{F}_p^r$, $\hat{m}$ as $\hat{n}_{\rm sift}$, $\hat{\omega}$ as $\hat{\Xi}_{\rm pub}$, ${\cal H}_{m,\omega}$ as ${\cal H}(mr, m r-n_{\rm fin}(m, \Xi))$, ${\cal W}$ as ${\cal W}(\rho_{\sysA\sysB})$, $\probP$ as $\probrho_{\rho_{\sysA\sysB}^{\otimes n_{\rm tot}}}$, and $\rho_{Q^m}^{\bm{x},\omega}$ as $\rho_{Q^m}^{\bm{x},\Xi}$, which is defined as 
\begin{equation}
    \rho_{Q^m}^{\bm{x},\Xi}\coloneqq \frac{\sum_{\Theta}\mathfrak{Z}_{Q^m}(\bm{x})\,\mathfrak{J}_{\sysA^{n_{\rm tot}}\sysB^{n_{\rm tot}}\to Q^m}^{(m,\Xi,\Theta)}(\rho_{\sysA^{n_{\rm tot}}\sysB^{n_{\rm tot}}})\,\mathfrak{Z}_{Q^m}(\bm{x})^{\dagger}}{\probrho_{\rho_{\sysA\sysB}^{\otimes n_{\rm tot}}}(\hat{n}_{\rm sift}=m, \hat{\Xi}=\Xi)}.
\end{equation}
Note that $\rho_{Q^m}^{\bm{x},\Xi}$ is normalized, as can be seen from Eqs.~\eqref{eq:rewritten_by_X} and \eqref{eq:event_prob_estimation}. 
Next, let $q_{\rho}$ be defined as 
\begin{equation}
    q_{\rho}\coloneqq \tr[\mathfrak{x}_{XQ}(\rho_{\sysA\sysB})].
\end{equation}
Combining $\sum_{(\xi,\theta)\in\Omega_{\rm pub}^{(0)}\times\Omega_{\rm aux}}{\cal M}_{\sysA\sysB}^{(\xi,\theta)}(\rho_{\sysA\sysB})=1-q_{\rho}$ with Eqs.~\eqref{eq:instrument_global}, \eqref{eq:def_of_projected_state}, \eqref{eq:rewritten_by_X}, and \eqref{eq:another_def_frak_x}, we have 
\begin{align}
    \sum_{\Xi} \!\probrho_{\rho_{\sysA\sysB}^{\otimes n_{\rm tot}}}(\hat{n}_{\rm sift}=m, \hat{\Xi}=\Xi)\!\! \sum_{\bm{x}\in (\mathbb{F}_p^r)^m}\! \! p^{-mr}\ketbra{\widetilde{\bm{x}}}{\widetilde{\bm{x}}}_{X^m}\otimes \rho_{Q^m}^{\bm{x},\Xi} &=\sum_{\Xi,\Theta} \mathfrak{X}_{X^{m}Q^{m}}^{(m,\Xi,\Theta)}(\rho_{\sysA\sysB}^{\otimes n_{\rm tot}}) \\
    &= {n_{\rm tot} \choose m} (1-q_{\rho})^{n_{\rm tot}-m} \mathfrak{x}_{XQ}(\rho_{\sysA\sysB})^{\otimes m} \\
    &={\rm Binom}[n_{\rm tot},q_{\rho}](m) \left(q_{\rho}^{-1}\mathfrak{x}_{XQ}(\rho_{\sysA\sysB}) \right)^{\otimes m}, \label{eq:marginal_binom}
\end{align}
where ${\rm Binom}[n_{\rm tot},q_{\rho}](m)$ is as defined in Eq.~\eqref{eq:def_of_binom}. This implies that the above reinterpretation also fulfills the condition~\eqref{eq:marginal_iid} in Proposition~\ref{prop:modified_for_qkd} with $q$ in Eq.~\eqref{eq:marginal_iid} reinterpreted as $q_{\rho}$ and $\rho_{XQ}$ reinterpreted as $q_{\rho}^{-1}\mathfrak{x}_{XQ}(\rho_{\sysA\sysB})$.
Finally, as has been shown in Sec.~\ref{sec:virtual_protocol}, ${\cal H}(mr, m r-n_{\rm fin}(m, \Xi))= \{G_i\}_{i}$ is a 2-universal family of surjective linear hash functions over $F_p$. This implies that the family of canonically induced maps $\{ G_i:(F_p^r)^{m} \to F_p^{m r - n_{\rm fin}(m,\Xi)} \}_i$ is also a 2-universal family of hash functions of Definition~\ref{def:2-universal_hash}, and hence all the prerequisites of Proposition~\ref{prop:modified_for_qkd} are fulfilled.
By combining Eqs.~\eqref{eq:reinterpretation_prop}--\eqref{eq:marginal_binom}, we have, from Eq.~\eqref{eq:variable_length_renyi} in Proposition~\ref{prop:modified_for_qkd},  
\begin{equation}
    \probrho_{\rho_{\sysA\sysB}^{\otimes n_{\rm tot}}}\!\left((\hat{n}_{\rm sift},\hat{\Xi}_{\rm pub})\in{\cal W}(\rho_{\sysA\sysB}), \hat{\bm{x}}\neq\hat{\bm{x}}^* \right) \leq \overline{P}_{\rm err}(\mathfrak{x}_{XQ}(\rho_{\sysA\sysB}), \{F_p^{m r - n_{\rm fin}(m,\Xi)}\}_{(m,\Xi)}, {\cal W}(\rho_{\sysA\sysB})),\label{eq:bounded_by_single_copy}
\end{equation}
where $\overline{P}_{\rm err}$ is defined in Eq.~\eqref{eq:def_overline_P}. Noticing that $K_{\rm PA}(m, \Xi)=mr-n_{\rm fin}(m,\Xi)$ holds when $n_{\rm fin}(m,\Xi)\geq 1$ from Eq.~\eqref{eq:final_key_as_function}, we have, from Eqs.~\eqref{eq:def_overline_P}, \eqref{eq:region_W}, and \eqref{eq:choice_of_final_key},
\begin{align}
    &\overline{P}_{\rm err}(\mathfrak{x}_{XQ}(\rho_{\sysA\sysB}), \{F_p^{m r - n_{\rm fin}(m,\Xi)}\}_{(m,\Xi)},{\cal W}(\rho_{\sysA\sysB}))\nonumber\\
    &= \min_{\alpha\in[0, 1]} \max_{(m,\Xi)\in {\cal W}(\rho_{\sysA\sysB})} 2^{-\alpha\left(K_{\rm PA}(m, \Xi)\log p - n_{\rm tot}H_{1-\alpha}^{\uparrow,\leq}(X|Q)_{\mathfrak{x}_{XQ}(\rho_{\sysA\sysB})}-[p^r(d_Q+2)(d_Q-1)\log(m+1)]/2\right)} \label{eq:applying_prop_1} \\
    &\leq \max_{(m,\Xi)\in {\cal W}(\rho_{\sysA\sysB})} \epsilon_p 2^{-\alpha^* n_{\rm tot}\left(\max_{\sigma\in V_{\epsilon_t,\epsilon_v}(\Xi)}H_{1-\alpha^*}^{\uparrow,\leq}(X|Q)_{\mathfrak{x}_{XQ}(\sigma_{\sysA\sysB})} - H_{1-\alpha^*}^{\uparrow,\leq}(X|Q)_{\mathfrak{x}_{XQ}(\rho_{\sysA\sysB})}\right)} \\
    &\leq \epsilon_p, \label{eq:remaining_term_bound}
\end{align}
where the last inequality follows from the fact that $(m,\Xi)\in{\cal W}(\rho_{\sysA\sysB})$ implies $\rho_{\sysA\sysB}\in V_{\epsilon_t,\epsilon_v}(\Xi)$.

Combining Eqs.~\eqref{eq:connection_to_source_compression}, \eqref{eq:bounded_by_single_copy}, and~\eqref{eq:remaining_term_bound}, if we can construct $\{V_{\epsilon_t,\epsilon_v}(\Xi)\}_{\Xi}$ that satisfies Eqs.~\eqref{eq:hull}, and $K_{\rm PA}(\hat{n}_{\rm sift},\hat{\Xi}_{\rm pub})$ is set to satisfy Eq.~\eqref{eq:choice_of_final_key}, then we have Eq.~\eqref{eq:failure_prob_collective} with
\begin{equation}
    \epsilon_{\rm iid} = \epsilon_v + \epsilon_p, \label{eq:epsilon_iid}
\end{equation}
which then leads to the secrecy condition with the secrecy parameter $\varepsilon_{\rm sec}$ as 
\begin{equation}
    \varepsilon_{\rm sec} = \sqrt{2}\sqrt{\epsilon_t + (\epsilon_v+\epsilon_p)f_q(n_{\rm tot},d_{\sysA\sysB})},
\end{equation}
from Eq.~\eqref{eq:correspondence_params}, where $f_q(n,d)$ is defined in Eq.~\eqref{eq:def_f_q}. A flow chart of the overall evaluation of the failure probability of PEC given from Sec.~\ref{sec:connection_partially_universal} to \ref{sec:estimation_failure_prob} is shown in Fig.~\ref{fig:flow_chart}.

\begin{figure}[t]
    \centering 
    \includegraphics[width=0.9\linewidth]{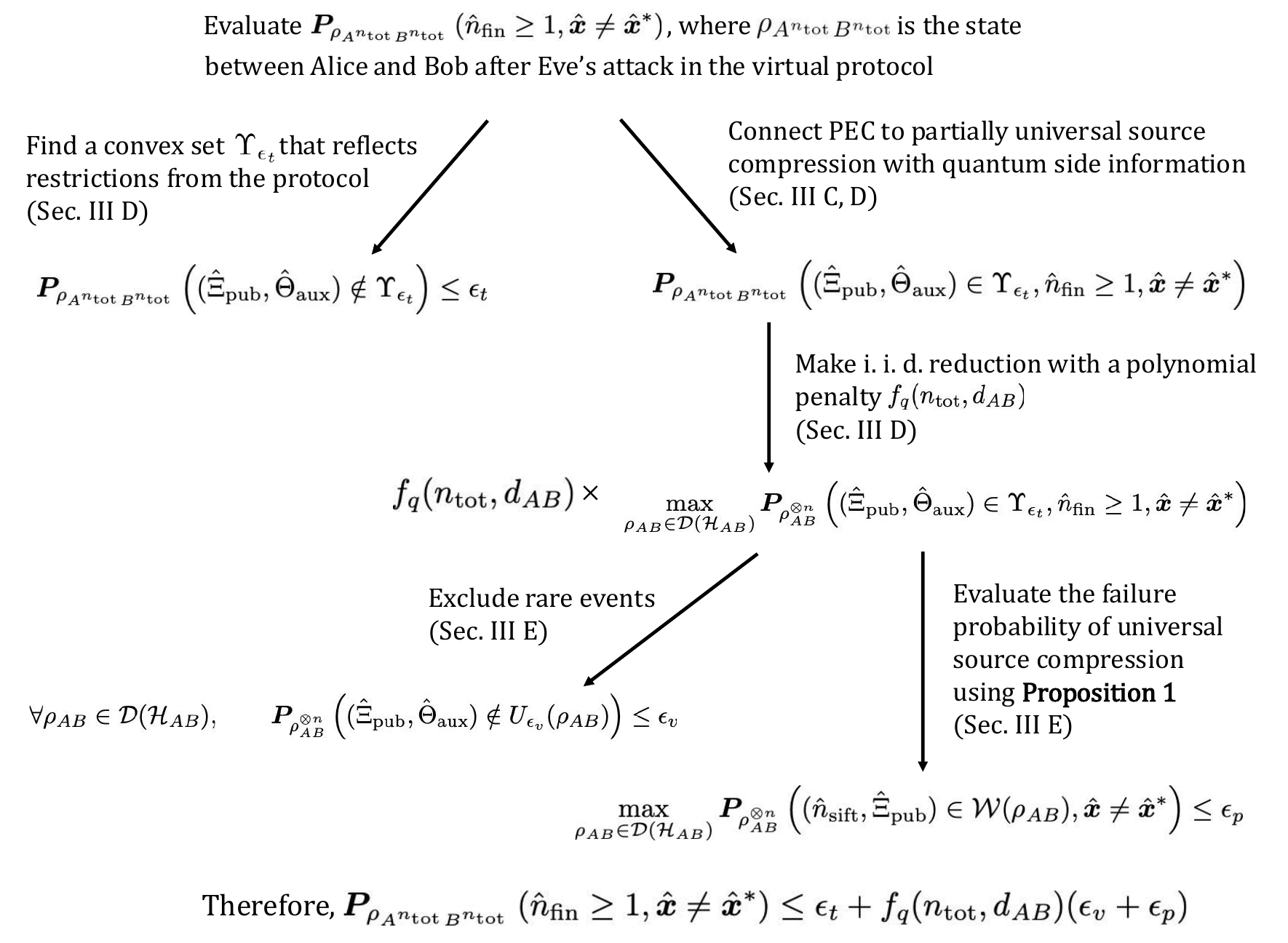}
    \caption{A flow chart of the evaluation of the failure probability of the phase error correction given from Sec.~\ref{sec:connection_partially_universal} to \ref{sec:estimation_failure_prob}. The set $\Upsilon_{\epsilon_t}$ is defined in Eq.~\eqref{eq:bound_out_of_Upsilon}, the factor $f_{q}(n_{\rm tot},d_{\sysA\sysB})$ is in Eq.~\eqref{eq:def_f_q}, and the family of sets $\{U_{\epsilon_v}(\rho_{\sysA\sysB})\}_{\rho_{\sysA\sysB}\in{\cal D}({\cal H}_{\sysA\sysB})}$ is in Eq.~\eqref{eq:ensured_by_concentration}. The set ${\cal W}(\rho_{\sysA\sysB})$ introduced in Eq.~\eqref{eq:region_W} is determined through a family of sets $\{V_{\epsilon_t,\epsilon_v}(\Xi)\}_{\Xi}$ satisfying Eq.~\eqref{eq:hull} and the function $n_{\rm fin}(\cdot,\cdot)$ satisfying Eqs.~\eqref{eq:final_key_as_function} and \eqref{eq:choice_of_final_key}.}
    \label{fig:flow_chart}
\end{figure}

We rewrite the optimization problem involved in Eq.~\eqref{eq:choice_of_final_key}:
\begin{align}
    &\text{Maximize}&   &H^{\uparrow,\leq}_{1-\alpha^*}(X|Q)_{\mathfrak{x}_{XQ}(\rho_{\sysA\sysB})},  \label{eq:optimization_to_be_solved}\\
    &\text{Subject to} & &\rho_{\sysA\sysB}\in V_{\epsilon_t,\epsilon_v}(\hat{\Xi}_{\rm pub}). \label{eq:constraint_test}
\end{align}
If $V_{\sigma_t,\sigma_v}(\Xi)$ is chosen to be a convex set for any $\Xi$, such as the convex hull of the right-hand side of Eq.~\eqref{eq:hull}, the above optimization problem is a (nonlinear) convex semidefinite programming (SDP) problem, since $H_{\alpha}^{\uparrow,\leq}(X|Q)_{\sigma}$ is concave with respect to a subnormalized state $\sigma$ as mentioned in Sec.~\ref{sec:preliminaries}.
Thus, $K_{\rm PA}(\hat{n}_{\rm sift},\hat{\Xi}_{\rm pub})$ can be rigorously computed in numerical optimization.

\begin{remark} \label{rem:optimization_region}
As mentioned, $V_{\sigma_t,\sigma_v}(\Xi)$ in Eq.~\eqref{eq:hull} needs to be convex for the optimization problem Eqs.~\eqref{eq:optimization_to_be_solved} and~\eqref{eq:constraint_test} to be a convex optimization. A sufficient condition for the set in the right-hand side of Eq.~\eqref{eq:hull} to be already convex is 1)~the section ${\cal S}_{\Xi}\coloneqq \{\Theta:(\Xi,\Theta)\in\Upsilon_{\epsilon_t}\}$ of $\Upsilon_{\epsilon_t}$ is a convex set for any $\Xi$, and 2)~for the sections ${\cal T}_{\Xi}(\rho_{\sysA\sysB})\coloneqq\{\Theta:(\Xi,\Theta)\in U_{\epsilon_v}(\rho_{\sysA\sysB})\}$, $t {\cal T}_{\Xi}(\rho_{\sysA\sysB}) + (1-t) {\cal T}_{\Xi}(\sigma_{\sysA\sysB}) \subseteq {\cal T}_{\Xi}(t\rho_{\sysA\sysB}+(1-t)\sigma_{\sysA\sysB})$ holds for any $\Xi$, any $\rho_{\sysA\sysB}, \sigma_{\sysA\sysB}\in{\cal D}({\cal H}_{\sysA\sysB})$, and $t\in[0,1]$~\cite[Chapter~5]{Rockafellar1998}, where the set in the left-hand side denotes that each element is given by the convex combination of elements of ${\cal T}_{\Xi}(\rho_{\sysA\sysB})$ and ${\cal T}_{\Xi}(\sigma_{\sysA\sysB})$. (Recall that the range of $\hat{\Theta}_{\rm aux}$ is a simplex and thus has a vector-space structure.) 
The condition~2) is called the graph convexity of the set-valued map ${\cal T}_{\Xi}(\cdot)$. 

In typical QKD protocols, ${\cal S}_{\Xi}$ takes the form ${\cal S}[a_{\Xi}] \coloneqq \{\Theta:f(\Theta)\leq a_{\Xi}\}$ and ${\cal T}_{\Xi}(\rho)$ takes the form ${\cal T}[g_{\Xi}](\rho) \coloneqq \{\Theta:f(\Theta)\geq g_{\Xi}(\rho)\}$, where $f:\mathbb{R}^{|\Omega_{\rm aux}|}\to\mathbb{R}$ is a linear functional, $a_{\Xi}\in\mathbb{R}$ is a parameter, and $g_{\Xi}:{\cal D}({\cal H}_{\sysA\sysB})\to\mathbb{R}$ is a function, due to the application of a concentration inequality. In this case, a sufficient condition for the set in Eq.~\eqref{eq:hull} to be already convex is loosened; it reduces to the quasiconvexity of the function $g_{\Xi}$ for any $\Xi$.

Practically, it may be convenient to apply different concentration inequalities for different random variables constructed from $\hat{\Xi}_{\rm pub}$ and $\hat{\Theta}_{\rm aux}$ when finding $\{V_{\epsilon_t,\epsilon_v}(\Xi)\}_{\Xi}$.
In this case, one may have several (family of) sets $\{\Upsilon^j_{\epsilon_{t_j}}\}_{j}$ and $\{\{U^j_{\epsilon_{v_j}}(\rho)\}_{\rho}\}_{j}$, each elements of which satisfy Eqs.~\eqref{eq:bound_out_of_Upsilon} and \eqref{eq:ensured_by_concentration}, respectively. Then, we can take $V_{\epsilon_t,\epsilon_v}(\Xi)$ as 
\begin{align}
    V_{\epsilon_t,\epsilon_v}(\Xi)&=\overline{{\rm conv}}\biggl(\bigcap_j V^j_{\epsilon_{t_j},\epsilon_{v_j}}(\Xi)\biggr), \label{eq:split_conditions}\\
    V^j_{\epsilon_{t_j},\epsilon_{v_j}}(\Xi) &= \bigcup_{\Theta:(\Xi,\Theta)\in\Upsilon^j_{\epsilon_{t_j}}}\{\rho\in{\cal D}({\cal H}_{\sysA\sysB}):(\Xi,\Theta)\in U^j_{\epsilon_{v_j}}(\rho)\}, \label{eq:respective_conditions}
\end{align}
where $\epsilon_t=\sum_j \epsilon_{t_j}$ and $\epsilon_v=\sum_j\epsilon_{v_j}$ from the union bound. If every $V^j_{\epsilon_{t_j},\epsilon_{v_j}}$ is a convex set, which may be easier to check, then the convex hull in Eq.~\eqref{eq:split_conditions} can be removed.
Note that the set defined as in Eq.~\eqref{eq:split_conditions} is in general larger than the one constructed from $\Upsilon_{\epsilon_{t}}=\bigcap_j \Upsilon^j_{\epsilon_{t_j}}$ and $U_{\epsilon_v}(\rho)=\bigcap_j U^j_{\epsilon_{v_j}}(\rho)$. They are, however, equal in the case ${\cal S}_j[a^j_{\Xi}]=\{\Theta:f_j(\Theta)\leq a^j_{\Xi}\}$, ${\cal T}_j[g^j_{\Xi}](\rho)=\{\Theta:f_j(\Theta)\geq g^j_{\Xi}(\rho)\}$, and $\{f_j\}_j$ are linearly independent linear functionals, which is a typical situation in QKD.
(See also the later example in Sec.~\ref{sec:numerical_comparison}.)
\end{remark}

\begin{remark} 
While our final key length is adaptively chosen for the value of $\hat{\Xi}_{\rm pub}$, we can also construct a non-adaptive protocol. 
Suppose that $\{V_{\epsilon_t,\epsilon_v}(\Xi)\}_{\Xi}$ is constructed as given in Eq.~\eqref{eq:hull}, where $\Upsilon_{\epsilon_t}$ is assumed to be a convex set, and the set-valued map $U_{\epsilon_v}(\cdot)$ is assumed to be graph-convex in the sense of Ref.~\cite{Rockafellar1998}. Then, the set-valued map $V_{\epsilon_t,\epsilon_v}(\cdot)$ is also graph-convex. Thus, for any convex subset ${\cal S}$ of the range of $\hat{\Xi}_{\rm pub}$, the set $\bigcup_{\Xi\in{\cal S}}V_{\epsilon_t,\epsilon_v}(\Xi)$ is convex. Now, choose ${\cal S}$ so that 1)~the protocol is aborted when $\hat{\Xi}_{\rm pub}\notin{\cal S}$ is observed, and 2)~the probability of the abort for an honest implementation is smaller than $\epsilon_c$ (completeness). Then, $V_{\epsilon_t,\epsilon_{v}}(\Xi)$ in Eq.~\eqref{eq:constraint_test} can be replaced with $\bigcup_{\Xi\in{\cal S}}V_{\epsilon_t,\epsilon_v}(\Xi)$, which makes the optimization problem Eq.~\eqref{eq:optimization_to_be_solved} independent of the value of $\hat{\Xi}_{\rm pub}$.
As in the previous remark, the condition of graph-convexity may be loosened when $\Upsilon_{\epsilon_t}$ and $U_{\epsilon_v}(\cdot)$ are more structured. 
\end{remark}

\subsection{Asymptotic optimality from the entropic uncertainty relation} \label{sec:asymptotic_optimality}
Given the expression of $K_{\rm PA}(\hat{n}_{\rm sift}, \hat{\Xi}_{\rm pub})$ in Eq.~\eqref{eq:choice_of_final_key}, the asymptotic key rate of our new analysis can be given as follows. First, we fix a channel model ${\cal N}_{C\to B}$, which amounts to determining $\rho_{\sysA\sysB}\in{\cal D}({\cal H}_{\sysA\sysB})$ that is shared between Alice and Bob at each round in the estimation protocol by $\rho_{\sysA\sysB} = {\cal N}_{C\to B}(\ketbra{\Psi}{\Psi}_{AC})$, where $\ket{\Psi}_{AC}$ is defined in the virtual protocol in Sec.~\ref{sec:virtual_protocol} and satisfies Eq.~\eqref{eq:purified_state_compatible}. Then, we have
\begin{align}
    &\forall \xi\in\Omega_{\xi},& \lim_{n_{\rm tot}\to \infty} \frac{\hat{\Xi}_{\rm pub}(\xi)}{n_{\rm tot}} &= P_{\rho_{\sysA\sysB}}(\hat{\xi}_{\rm pub}=\xi), \\ 
    &\forall \theta\in\Omega_{\theta}, & \lim_{n_{\rm tot}\to \infty} \frac{\hat{\Theta}_{\rm aux}(\theta)}{n_{\rm tot}} &= P_{\rho_{\sysA\sysB}}(\hat{\theta}_{\rm aux}=\theta).
\end{align}
Through the asymptotic concentrations for $\Upsilon_{\epsilon_t}$ and $U_{\epsilon_v}(\cdot)$, we can find a set $V^{\infty}(P_{\rho_{\sysA\sysB}}(\hat{\xi}_{\rm pub}))$ of density operators that is compatible with observed events in the asymptotic limit as 
\begin{equation}
    V^{\infty}(P_{\rho_{\sysA\sysB}}(\hat{\xi}_{\rm pub})) \supseteq \bigcup_{\epsilon_t,\epsilon_v>0}\bigcap_{n_{\rm tot}\in\mathbb{N}} V_{\epsilon_t,\epsilon_v}(\Xi).
\end{equation}
Since $H_{1-\alpha^*}^{\uparrow,\leq}(X|Q)_{\mathfrak{x}_{XQ}(\rho_{\sysA\sysB})}$ is continuously non-decreasing for the arbitrarily chosen parameter $\alpha^*\in[0,1]$~\cite{Cheng2021}, the asymptotic rate of sacrificing dits in the privacy amplification step can be given by
\begin{equation}
    \lim_{n_{\rm tot}\to \infty} \frac{K_{\rm PA}(\hat{n}_{\rm sift}, \hat{\Xi}_{\rm pub})}{n_{\rm tot}} = \frac{1}{\log p} \max_{\sigma_{\sysA\sysB}\in V^{\infty}(P_{\rho_{\sysA\sysB}}(\hat{\xi}_{\rm pub}))} H(X|Q)_{\mathfrak{x}_{XQ}(\sigma_{AB})},
\end{equation}
where we took the limits $\alpha^*\to 0$ and $n_{\rm tot}\alpha^*\to \infty$ for Eq.~\eqref{eq:choice_of_final_key} as $n_{\rm tot}\to \infty$ and applied Eq.~\eqref{eq:subnormalized_cond_ent}.
This immediately implies that the asymptotic key rate in bits under a channel model ${\cal N}_{C\to B}$ is given by 
\begin{equation}
    q_{\rho}r\log p- \max_{\sigma_{\sysA\sysB}\in V^{\infty}(P_{\rho_{\sysA\sysB}}(\hat{\xi}_{\rm pub}))}H(X|Q)_{\mathfrak{x}_{XQ}(\sigma_{\sysA\sysB})}, \label{eq:rate_achievable_by_new}
\end{equation}
where $q_{\rho} = \tr[\mathfrak{x}_{XQ}(\rho_{\sysA\sysB})] =\tr[\mathfrak{x}_{XQ}\circ {\cal N}_{C\to B}(\ketbra{\Psi}{\Psi}_{AC})]$. In the following, we show that this rate coincides with the optimal key rate called the Devetak-Winter rate~\cite{Devetak2005}, which is defined for the asymptotic limit of i.i.d.~quantum channels ${\cal N}_{C\to B}^{\otimes n}$. 

The Devetak-Winter rate $R_{\rm DW}$ in the context of QKD is defined for a (subnormalized) state $\Phi_{Z_A Z_B E'}$ between Alice's sifted key $Z_A$, Bob's sifted key $Z_B$, and Eve's system $E'=E R$, where $R$ denotes the classical register for the public announcement for each round as 
\begin{align}
    R_{\rm DW} &\coloneqq I(Z_A:Z_B)_{\Phi_{Z_AZ_B}} - I(Z_A:E')_{\Phi_{Z_AE'}} \\
    &= H(Z_A|E')_{\Phi_{Z_A E'}} - H(Z_A|Z_B)_{\Phi_{Z_A Z_B}}.
\end{align}
When the underlying channel in the Actual protocol in Sec.~\ref{sec:virtual_protocol} is ${\cal N}_{C\to B}$, the state $\Phi_{Z_A Z_B E'}$ in the above can explicitly be written, with a Stinespring isometry $L_{C\to BE}$ of the channel ${\cal N}_{C\to B}$ and $q_a$, $\rho_a$, $M^b_B$, $\mathfrak{t}$, $\mathfrak{s}$, and $\mathfrak{s}'$ introduced in the Actual protocol, as 
\begin{equation}
\begin{split}
    \Phi_{Z_A Z_B E'} &\coloneqq \sum_{\xi\in\Omega_{\rm pub}^{(1)}}\sum_{(a, b)\in {\cal X}_{\sysA}\times{\cal X}_{\sysB}:\xi=\mathfrak{t}(a, b)}\ketbra{\mathfrak{s}[\xi](a)}{\mathfrak{s}[\xi](a)}_{Z_A}\otimes \ketbra{\mathfrak{s}'[\xi](b)}{\mathfrak{s}'[\xi](b)}_{Z_B} \\
    &\hspace{5cm}\otimes \tr_{B}[M^{b}_B\, L_{C\to BE} (q_a\rho_a) L_{C\to BE}^{\dagger}]\otimes \ketbra{\xi}{\xi}_R.
    \end{split}
\end{equation}
Since the sifted key is then embedded by an injection $\mathfrak{e}:{\cal X}_{\rm sift} \hookrightarrow \mathbb{F}_p^r$ in the Actual protocol, we have  
\begin{equation}
    R_{\rm DW} = H(Z|E')_{\Phi_{\mathfrak{e}(Z_A) E'}} - H(Z|Z_B)_{\Phi_{\mathfrak{e}(Z_A) Z_B}},
\end{equation}
where we used the fact that the entropy stays invariant under an isometric embedding.
The right-most term in the above is an asymptotic secret-key consumption rate for the information reconciliation. Thus, if the asymptotic final key generation rate Eq.~\eqref{eq:rate_achievable_by_new} in our protocol is equal to the first term $H(Z|E')_{\Phi_{\mathfrak{e}(Z_A) E'}}$ in the right-hand side, then our protocol achieves the Devetak-Winter rate.

We first consider the scenario in which Alice and Bob can completely identify the underlying channel ${\cal N}_{\sysC\to \sysB }$ via $(\hat{\xi}_{\rm pub}, \hat{\theta}_{\rm aux})$, i.e., the set $V^{\infty}(P_{\rho_{AB}}(\hat{\xi}_{\rm pub}))$ is a singleton.
We assume that the Actual protocol in Sec.~\ref{sec:virtual_protocol} further satisfies the following two conditions:
\begin{enumerate}
    \item For any $\xi\in\Omega_{\rm pub}^{(1)}$, the Kraus rank of the CP map ${\cal J}^{(\xi)}_{\sysA\sysB\to Q}$ in Eq.~\eqref{eq:success_map} is one; i.e., ${\cal J}^{(\xi)}_{\sysA\sysB\to Q}(\cdot)$ can be written as ${\cal J}^{(\xi)}_{\sysA\sysB\to Q}(\cdot)=J_{\sysA\sysB\to Q}^{(\xi)}(\cdot)\bigl(J_{\sysA\sysB\to Q}^{(\xi)}\bigr)^{\dagger}$ for an operator $J_{\sysA\sysB\to Q}^{(\xi)}$. This also imposes the following condition on ${\cal J}^{(\xi, \theta)}_{\sysA\sysB\to Q}$ in Eq.~\eqref{eq:success_map_estimation}: 
    \begin{equation}
        \forall \xi\in\Omega_{\rm pub}^{(1)}, \qquad  {\cal J}^{(\xi, \theta)}_{\sysA\sysB\to Q} = p_{\xi}(\theta) {\cal J}^{(\xi)}_{\sysA\sysB\to Q}, \label{eq:trivial_except_one}
    \end{equation}
    where $\{p_{\xi}(\theta)\}$ is non-negative numbers with $\sum_{\theta}p_{\xi}(\theta)=1$.
    \item The CP maps $\{{\cal J}^{(\xi)}_{\sysA\sysB\to Q}\}_{\xi\in\Omega_{\rm pub}^{(1)}}$ additionally satisfy  
    \begin{equation}
        \forall \xi_1,\xi_2,\text{ s.t.\ } \xi_1\neq \xi_2, \qquad {\cal J}^{(\xi_1)}_{\sysA\sysB\to Q}(I_{\sysA\sysB})\; {\cal J}^{(\xi_2)}_{\sysA\sysB\to Q}(I_{\sysA\sysB}) =0, \label{eq:orthogonal_subspace}
    \end{equation}
    which means that the range of the CP maps $\{{\cal J}^{(\xi)}_{\sysA\sysB\to Q}\}_{\xi\in\Omega_{\rm pub}^{(1)}}$ are orthogonal to each other for different values of $\xi\in\Omega_{\rm pub}^{(1)}$. This implies that the Hilbert space ${\cal H}_Q$ can be decomposed as ${\cal H}_Q=\bigoplus_{\xi\in\Omega_{\rm pub}^{(1)}} {\cal H}_Q[\xi]$, where ${\cal J}^{(\xi)}_{\sysA\sysB\to Q}(\cdot)$ acts only on $ {\cal H}_Q[\xi]$ for any input $\cdot$.
\end{enumerate}
It is always possible to adjust the protocol to satisfy these conditions. Note, however, that this adjustment is often at the cost of sacrificing finite-size performance. (For example, one may mitigate the finite-size effect by reducing the dimension of ${\cal H}_Q$, which may violate the condition.)

Now, let $\ket{\phi}_{K^rQ E'}$ be defined as 
\begin{equation}
    \ket{\phi}_{K^rQ E'} \coloneqq \sum_{\xi\in\Omega_{\rm pub}^{(1)}} C_{Q\to K^r Q} J_{AB\to Q}^{(\xi)} L_{C\to BE} \ket{\Psi}_{AC} \otimes \ket{\xi}_{R}.
\end{equation}
Then, Combining Eqs.~\eqref{eq:key_extraction_isometry}, \eqref{eq:def_frak_x}, \eqref{eq:trivial_except_one}, and \eqref{eq:orthogonal_subspace} with the fact that $\{\Pi_Q[z]\}_{z\in{\cal X}_{\rm sift}}$ are orthogonal projections, we have 
\begin{align}
    \tr_{Q}[\ketbra{\phi}{\phi}_{K^r Q E'}] &= \sum_{z\in \mathbb{F}_p^r} \ketbra{z}{z}_{K^r} \tr_{Q}[\ketbra{\phi}{\phi}_{K^r Q E'}]  \ketbra{z}{z}_{K^r}, \label{eq:Z-diagonal_when_traced_out}\\ 
    {\cal P}_{K^r\to Z}\left(\tr_{Q}[\ketbra{\phi}{\phi}_{K^r Q E'}]\right) &= \Phi_{\mathfrak{e}(Z_A) E'}, \label{eq:compatible_Z}\\
    \intertext{and}
    {\cal P}_{K^r\to X}\left(\tr_{E'}[\ketbra{\phi}{\phi}_{K^r Q E'}]\right) &= \mathfrak{x}_{X Q}(\rho_{\sysA\sysB}). \label{eq:compatible_X}
\end{align}
Equations~\eqref{eq:Z-diagonal_when_traced_out} and \eqref{eq:compatible_Z} are exactly the conditions for the standard form defined in Ref.~\cite{Tsurumaru2022}, which implies that the normalized tripartite pure state $q_{\rho}^{-1}\ketbra{\phi}{\phi}_{K^rQE'}$ satisfies the equality condition of the entropic uncertainty relation~\cite{Coles2012, Coles2017, Tsurumaru2022}, noticing that $q_{\rho} = \tr[\mathfrak{x}_{X Q}(\rho_{\sysA\sysB})] = \braket{\phi|\phi}$.
Thus, we have 
\begin{equation}
    H(X|Q)_{{\cal P}_{K^r\to X}(q_{\rho}^{-1}\ketbra{\phi}{\phi}_{K^rQE'})} + H(Z|E')_{{\cal P}_{K^r\to Z}(q_{\rho}^{-1}\ketbra{\phi}{\phi}_{K^rQE'})} = \log p^r.
\end{equation}
Combining this with Eqs.~\eqref{eq:compatible_Z} and~\eqref{eq:compatible_X}, we have
\begin{equation}
    H(Z|E')_{\Phi_{\mathfrak{e}(Z_A) E'}} = q_{\rho} r \log p - H(X|Q)_{\mathfrak{x}_{X Q}(\rho_{\sysA\sysB})},
\end{equation}
where we used $H(A|B)_{\tilde{\rho}} = \tr[\tilde{\rho}] \, H(A|B)_{(\tr[\tilde{\rho}])^{-1}\tilde{\rho}}$ for a subnormalized $\tilde{\rho}$ from Eq.~\eqref{eq:subnormalized_cond_ent}. Thus, our security proof approach achieves the Devetak-Winter rate as long as the protocol is appropriately modified to fulfill the mentioned two conditions.

Next, we consider a more realistic scenario in which Alice and Bob can partially know the underlying channel via $(\hat{\xi}_{\rm pub}, \hat{\theta}_{\rm aux})$. In this case, they amount to have a set $\{{\cal N}^{\zeta}_{\sysC \to \sysB}\}_{\zeta}$ of possible channels, each of which corresponds to an element of $V^{\infty}(P_{\rho_{AB}}(\hat{\xi}_{\rm pub}))$. Then, the asymptotic key rate of our security analysis is given by the smallest among these possible channels. The same applies to the Devetak-Winter rate, where $R_{\rm DW}$ is given as the lowest rate among these channels. Therefore, our proof strategy can achieve the asymptotically optimal key rate in this case as well.

Unlike our new method based on universal classical source compression with quantum side information, the conventional phase error correction approach can only achieve a suboptimal key rate asymptotically, which replaces $H(X|Q)_{\mathfrak{x}_{XQ}(\rho_{\sysA\sysB})}$ with $\inf_{{\cal M}_{Q}} H(X | Q)_{{\cal M}_{Q}(\mathfrak{x}_{XQ}(\rho_{\sysA\sysB}))}$, where ${\cal M}_{Q}$ is a measurement channel on the system $Q$.  
This is because, in the conventional approach, the failure probability of PEC is evaluated through the reduction to the classical statistics of the phase error rate. (The phase error operator that is used to compute the phase error rate is usually defined as $\sum_{\xi\in\Omega_{\rm pub}^{(1)}}({\cal J}_{\sysA\sysB\to Q}^{(\xi)})^{\dagger}(E_Q)$ with a POVM element $E_Q$ of the measurement channel ${\cal M}_Q$.)
It is known that $\inf_{{\cal M}_{Q}} H(X | Q)_{{\cal M}_{Q}(\mathfrak{x}_{XQ}(\rho_{\sysA\sysB}))}$ is in general larger than $H(X|Q)_{\mathfrak{x}_{XQ}(\rho_{\sysA\sysB})}$, and the difference between these two is called the quantum discord~\cite{Ollivier2001} of the c-q state $\mathfrak{x}_{XQ}(\rho_{\sysA\sysB})$.
The necessary and sufficient condition for the quantum discord to be zero for the state $\mathfrak{x}_{XQ}(\rho_{\sysA\sysB})$ is that $\bra{\widetilde{x}}_{X}\mathfrak{x}_{XQ}(\rho_{\sysA\sysB})\ket{\widetilde{x}}_{X} = \sum_{\xi\in\Omega_{\rm pub}^{(1)}}p^{-r} \mathfrak{Z}'_{Q}(x) {\cal J}^{(\xi)}_{\sysA\sysB\to Q}(\rho_{\sysA\sysB})\mathfrak{Z}'_{Q}(x)^{\dagger}$ for every $x\in\mathbb{F}_p^r$ mutually commutes~\cite{Ollivier2001}, where $\mathfrak{Z}'_Q$ is defined in Eq.~\eqref{eq:Z_restriction}.  For the QKD protocols with this condition satisfied, the conventional PEC-type security proof can also asymptotically achieve the Devetak-Winter rate.

\subsection{A choice of the parameter $\alpha^*$ in finite size} \label{sec:parameter_choice}
Although the security proof developed in Sec.~\ref{sec:virtual_protocol}--\ref{sec:estimation_failure_prob} works for any value of the parameter $\alpha^{*}\in[0,1]$ that appears in Eq.~\eqref{eq:choice_of_final_key}, we would like to choose it so that $K_{\rm PA}(\hat{n}_{\rm sift},\hat{\Xi}_{\rm pub})$ is minimized, or equivalently, the final key length $n_{\rm fin}(\hat{n}_{\rm sift},\hat{\Xi}_{\rm pub})$ is maximized.  Since the minimization of the right-hand side of Eq.~\eqref{eq:choice_of_final_key} over $\alpha^*$ may not be a convex problem, there may need a heuristic approach to find a good initial guess of $\alpha^*$ from which $K_{\rm PA}(\hat{n}_{\rm sift},\hat{\Xi}_{\rm pub})$ is numerically minimized.  We will briefly comment on this problem in this section.

For the $\alpha$-R\'enyi divergence defined in Eq.~\eqref{eq:alpha_Renyi_divergence}, the function $\alpha\mapsto D_{\alpha}(\rho\|\sigma)$ is continuously differentiable on $(0,\infty)$ for any positive semidefinite $\rho$ and $\sigma$ with $\tr[\rho]=1$ and $\mathrm{supp}(\rho)\subseteq\mathrm{supp}(\sigma)$~\cite{Lin2015}. In fact, the derivative at $\alpha\neq 1$ is given from Eq.~\eqref{eq:alpha_Renyi_divergence} by 
\begin{equation}
    \partial_{\alpha} D_{\alpha}(\rho\|\sigma) = \frac{\log e}{\alpha - 1}\frac{\tr[\rho^{\alpha} \sigma^{1-\alpha} (\ln\rho - \ln\sigma)]}{\tr[\rho^{\alpha}\sigma^{1-\alpha}]} - \frac{1}{(\alpha - 1)^2}\log \tr[\rho^{\alpha}\sigma^{1-\alpha}], \label{eq:derivative_renyi_divergence_neq_1}
\end{equation}
and the derivative at $\alpha=1$ is given by~\cite[Theorem~5]{Lin2015} 
\begin{equation}
    \frac{\partial}{\partial \alpha} D_{\alpha}(\rho\|\sigma) \Bigr|_{\alpha=1} = \frac{1}{2\log e}V(\rho\|\sigma), \label{eq:derivative_renyi_divergence_at_1}
\end{equation}
where the relative entropy variance $V(\cdot\|\cdot)$ is defined in Eq.~\eqref{eq:relative_entropy_variance}.
From Eqs.~\eqref{eq:derivative_renyi_divergence_neq_1} and \eqref{eq:derivative_renyi_divergence_at_1}, it can be seen that the map $(\alpha, \sigma) \to \partial_{\alpha} D_{\alpha}(\rho_{XQ} \| I_X \otimes \sigma_{Q})$ is jointly continuous on $(0, \infty) \times \{\sigma_Q\in{\cal D}({\cal H}_Q):\mathrm{supp}(\sigma_Q)=\mathrm{supp}(\rho_Q)\}$. Furthermore, the maximizer $\sigma_{\alpha}^*$ in Eq.~\eqref{eq:maximizer_unique} for the conditional $\alpha$-R\'enyi entropy $H^{\uparrow}_{\alpha}(X|Q)_{\rho}$ in Eq.~\eqref{eq:conditional_alpha_Renyi} is unique and continuous for $\alpha\in(0, \infty)$. In particular, $\mathrm{supp}(\rho_Q)=\mathrm{supp}(\sigma_{\alpha}^*)$ is fulfilled.
Combining these facts, one can apply the envelope theorem~\cite[Proposition 2.1]{Oyama2018} at $\alpha=1$ to have 
\begin{align}
    \left.\frac{\partial}{\partial \alpha} H_{\alpha}^{\uparrow}(X|Q)_{\rho}\right|_{\alpha=1} &= -\left.\frac{\partial}{\partial \alpha} D_{\alpha}(\rho_{XQ}\|I_X\otimes\sigma_Q)\right|_{\alpha=1, \sigma=\sigma_1^*} \\
    & = -\frac{1}{2\log e}V(\rho_{XQ}\|I_X\otimes \rho_Q),
\end{align}
where the second equality follows from Eq.~\eqref{eq:derivative_renyi_divergence_at_1}, combined with $\sigma_1^*=\rho$.
From this and the fact that $H^{\uparrow}_{\alpha}(X|Q)_{\rho}$ is monotonically non-increasing for $\alpha\in[0, 1]$~\cite{Cheng2021}, $\alpha=1$ is the unique minimizer of $\min_{\alpha\in[0,1]} H^{\uparrow}_{\alpha}(X|Q)_{\rho}$ as long as $V(\rho_{XQ}\|I_X\otimes \rho_Q)>0$. Thus, $H^{\uparrow}_{\alpha}(X|Q)_{\rho}$ can be expanded around $\alpha = 1$ in this case as
\begin{equation}
    H^{\uparrow}_{\alpha}(X|Q)_{\rho} = H(X|Q)_{\rho} - \frac{\alpha - 1}{2\log e}V(\rho_{XQ}\|I_X\otimes \rho_Q) + o(|\alpha - 1|). \label{eq:expansion_around_1}
\end{equation} 

Now, let $\overline{\rho}_{\sysA\sysB}$ be a density operator that maximizes $\{H(X|Q)_{\mathfrak{x}_{X Q}(\sigma_{\sysA\sysB})}\}_{\sigma_{\sysA\sysB}\in V^{\infty}(P_{\rho_{AB}}(\hat{\xi}_{\rm pub}))}$ in Eq.~\eqref{eq:rate_achievable_by_new}, where $\rho_{AB}$ is determined by a channel model.  The value $H(X|Q)_{\mathfrak{x}_{X Q}(\overline{\rho}_{\sysA\sysB})}$ is thus the asymptotic rate of sacrificing bits at the privacy amplification. Then, from Eq.~\eqref{eq:subnormalized_Renyi}, \eqref{eq:subnormalized_cond_ent}, and \eqref{eq:expansion_around_1}, we have 
\begin{equation}
    H^{\uparrow,\leq}_{1-\alpha^*}(X|Q)_{\mathfrak{x}_{XQ}(\overline{\rho}_{AB})} = H(X|Q)_{\mathfrak{x}_{XQ}(\overline{\rho}_{AB})} + \frac{\alpha^*}{2\log e} V\bigl(\mathfrak{x}_{XQ}(\overline{\rho}_{AB})\|I_X\otimes \tr_X[\mathfrak{x}_{XQ}(\overline{\rho}_{AB})]\bigr) + o(\alpha^*),
\end{equation}
where we have assumed $V\bigl(\mathfrak{x}_{XQ}(\overline{\rho}_{AB})\|I_X\otimes \tr_X[\mathfrak{x}_{XQ}(\overline{\rho}_{AB})]\bigr)>0$. 
Thus, under this assumption, Eq.~\eqref{eq:choice_of_final_key} can be given for $n_{\rm tot}\gg 1$ (and thus $\alpha^* \ll 1$) by 
\begin{align}
    \begin{split}
    K_{\rm PA}(\hat{n}_{\rm sift},\hat{\Xi}_{\rm pub}) 
    &= \frac{n_{\rm tot}}{\log p} \left[ H(X|Q)_{\mathfrak{x}_{XQ}(\overline{\rho}_{\sysA\sysB})} + \frac{\alpha^*}{2\log e}V\!\left(\mathfrak{x}_{XQ}(\overline{\rho}_{\sysA\sysB})\|I_X \otimes \tr_{X}\!\bigl[\mathfrak{x}_{XQ}(\overline{\rho}_{\sysA\sysB})\bigr]\right) + o(\alpha^*)\right] \\
    & \hspace{3cm} +\frac{p^r(d_Q+2)(d_Q-1)\log_p(\hat{n}_{\rm sift}+1)}{2} +\frac{\log_p(1/\epsilon_p)}{\alpha^*}.
    \end{split}
\end{align}
The second dominant terms in the above when $n_{\rm tot}\gg 1$ thus scales as $\sim n_{\rm tot}^{1/2}$, which can be obtained by setting $\alpha^*= O(n_{\rm tot}^{-1/2})$.
In fact, the best choice of $\alpha^*$ to minimize these second dominant terms is 
\begin{equation}
    \alpha^* = \left(\frac{2\log e\log(1/\epsilon_p)}{ n_{\rm tot} V\!\left(\mathfrak{x}_{XQ}(\overline{\rho}_{\sysA\sysB})\|I_X \otimes \tr_{X}[\mathfrak{x}_{XQ}(\overline{\rho}_{\sysA\sysB})]\right)}\right)^{\frac{1}{2}}, \label{eq:good_choice_of_alpha}
\end{equation}
which gives the amount of privacy amplification dits $K_{\rm PA}(\hat{n}_{\rm sift},\hat{\Xi}_{\rm pub})$ as 
\begin{equation}
    \begin{split}
    K_{\rm PA}(\hat{n}_{\rm sift},\hat{\Xi}_{\rm pub}) &= \frac{1}{\log p} \biggl(n_{\rm tot} H(X|Q)_{\mathfrak{x}_{XQ}(\overline{\rho}_{\sysA\sysB})} \\
    &\hspace{2cm} + \sqrt{2n_{\rm tot}\ln(1/\epsilon_p)\,V\!\left(\mathfrak{x}_{XQ}(\overline{\rho}_{\sysA\sysB})\|I_X \otimes \tr_{X}[\mathfrak{x}_{XQ}(\overline{\rho}_{\sysA\sysB})]\right)} + o(n_{\rm tot}^{1/2}) \biggr). 
    \end{split}\label{eq:idealized_amount_syndrome}
\end{equation}
For a finite-size case, the optimal choice of $\alpha^*$ may deviate from the value in Eq.~\eqref{eq:good_choice_of_alpha} due to the $o(n_{\rm tot}^{1/2})$ terms, but the value in Eq.~\eqref{eq:good_choice_of_alpha} should be a good initial guess for a further numerical optimization of $\alpha^*$.

\section{Numerical comparison} \label{sec:numerical_comparison}
In this section, we compare the key rates with the conventional PEC-type security proof and with our new approach applied to an explicit problem.  
In the previous section, we discussed the conditions under which the key rates in the conventional analysis and in our new analysis differ.  
In this perspective, the Bennett 1992 (B92) protocol~\cite{Bennett1992} is particularly insightful since in the B92 protocol, $\bra{\widetilde{0}}_X\mathfrak{x}_{XQ}(\rho_{\sysA\sysB}^{\rm exp})\ket{\widetilde{0}}_X$ and $\bra{\widetilde{1}}_X\mathfrak{x}_{XQ}(\rho_{\sysA\sysB}^{\rm exp})\ket{\widetilde{1}}_X$ ($\dim {\cal H}_X=2$) commute when the channel is ideal and do not commute in general when there is a non-zero bit error rate.  Thus, one can expect that the key rates with the conventional analysis and ours are the same under zero bit errors and different under non-zero bit errors.  As can be shown later, this intuition turns out to be true.

We first define the B92 protocol as follows.

\bigskip 
\noindent --- B92 protocol ---

\noindent Alice and Bob agree on the probabilities $p_{\rm s}$, $p_{\rm t}$, and $p_{\rm d}(=1 - p_{\rm s} - p_{\rm t})$ of choosing the sifting round from which a sifted key bit is probabilistically generated, the test round in which parameters are estimated, and the discard round in which information is discarded, respectively. Define ${\cal X}_{\sysA}\coloneqq \{(a,\text{``s''}), (a,\text{``t''}),(a,\text{``d''})\}_{a\in\{0,1\}}$ and ${\cal X}_B\coloneqq\{0, 1, \text{failure}\}$. 
They also agree on the state $\ket{\psi_a}\coloneqq \beta\ket{\widetilde{0}} + (-1)^a\alpha\ket{\widetilde{1}}$ that Alice sends, where $0<\alpha<1/\sqrt{2}$ and $\beta=\sqrt{1 - \alpha^2}$. 
\begin{enumerate}
    \item Alice generates a uniform random bit $\hat{a}$ and randomly chooses one of ``s'', ``t'', and ``d'' with the probabilities $p_{\rm s}$, $p_{\rm t}$, and $p_{\rm d}$, respectively. According to a generated bit $\hat{a}$, she sends the state $\ket{\psi_{\hat{a}}}$ to Bob.  Bob performs a measurement with a POVM $\{M_{\sysB}^b\}_{b\in{\cal X}_{\sysB}}$ to obtain $\hat{b}$, where each POVM element is defined as 
    \begin{align}
        M^{0(1)}_{\sysB} &\coloneqq \frac{1}{2}\ketbra{\psi_{1(0)}^{\perp}}{\psi_{1(0)}^{\perp}}_{\sysB},\\
        M^{\text{failure}}_{\sysB} & \coloneqq I_{\sysB} - \frac{1}{2}\ketbra{\psi_{0}^{\perp}}{\psi_{0}^{\perp}}_{\sysB} - \frac{1}{2}\ketbra{\psi_{1}^{\perp}}{\psi_{1}^{\perp}}_{\sysB},
    \end{align}
    with $\ket{\psi_a^{\perp}}\coloneqq \alpha\ket{\widetilde{0}} - (-1)^a \beta\ket{\widetilde{1}}$. They repeat this quantum communication $n_{\rm tot}$ times, where Alice may initiate each round before the previous rounds have completed.  
    \item After $n_{\rm tot}$ rounds of quantum communication, Alice announces the label ``s'', ``t'', or ``d'' she initially chose for each of $n_{\rm tot}$ rounds.  Depending on the announced label ``s'', ``t'', or ``d'' at the $i$-th round, Alice and Bob further perform one of the following operations.  
    \begin{itemize}
        \item[``s''] Bob announces $\hat{s}^{(i)}=0$ or $1$ depending on whether his measurement outcome is ``failure'' or the others. The random variable $\hat{\xi}_{\rm pub}^{(i)}$ is then given by $\hat{\xi}_{\rm pub}^{(i)}=(\text{``s''}, \hat{s}^{(i)})$.  
        \item[``t''] Bob announces his measurement outcomes. Unless ``failure'' for which $\hat{\xi}^{(i)}_{\rm pub}=(\text{``t''}, \varnothing)$ is assigned, Alice announces $\hat{a}^{(i)}$, and then the random variable $\hat{\xi}_{\rm pub}^{(i)}$ is given by $\hat{\xi}^{(i)}_{\rm pub}=(\text{``t''}, (\hat{a}^{(i)},\hat{b}^{(i)}))$.
        \item[``d''] Alice and Bob discard their bits and outcomes, where the random variable $\hat{\xi}^{(i)}_{\rm pub}$ is assigned to be $\hat{\xi}^{(i)}_{\rm pub}=(\text{``d''})$.
    \end{itemize} 
    Let $\hat{\Xi}_{\rm pub} = n_{\rm tot}P_{\hat{\bm{\xi}}_{\rm pub}}$.
    Let $\hat{n}_{\rm sift}$ be the length of a key extracted from ``s'' rounds, $\hat{n}_{\rm suc}$ be the number of successful measurement in ``t'' rounds, and $\hat{n}_{\rm err}$ be the number of bit errors among $\hat{n}_{\rm suc}$ successful rounds, which are given explicitly by 
    \begin{align}
        \hat{n}_{\rm sift} &\coloneqq \left|\left\{i\in\{1,\ldots,n_{\rm tot}\}:\hat{\xi}_{\rm pub}^{(i)}=(\text{``s''}, 1)\right\}\right| = \hat{\Xi}_{\rm pub}((\text{``s''}, 1)), \label{eq:def_n_sift}\\
        \hat{n}_{\rm suc} &\coloneqq \left|\left\{i\in\{1,\ldots,n_{\rm tot}\}:\hat{\xi}_{\rm pub}^{(i)}\in \bigcup_{a,b\in\{0,1\}}\{(\text{``t''}, (a, b))\}\right\}\right|= \sum_{a, b\in\{0,1\}}\hat{\Xi}_{\rm pub}((\text{``t''}, (a, b))), \\
        \hat{n}_{\rm bit} &\coloneqq \left|\left\{i\in\{1,\ldots,n_{\rm tot}\}:\hat{\xi}_{\rm pub}^{(i)}\in \bigcup_{a,b\in\{0,1\}:a\neq b}\{(\text{``t''}, (a, b))\}\right\}\right|=\sum_{\substack{a, b\in\{0,1\}\\ :a\neq b}}\hat{\Xi}_{\rm pub}((\text{``t''}, (a, b))). \label{eq:def_n_bit}
    \end{align}
    \item (Permutation symmetrization) 
    Alice and Bob joint-randomly reorder their keys by consuming ${\cal O}(\hat{n}_{\rm sift}\log\hat{n}_{\rm sift})$-bit local randomness and public communication and obtain sifted keys $\hat{\bm{k}}^{\rm sift}_A$ and $\hat{\bm{k}}^{\rm sift}_B$, respectively.
    \item (Information reconciliation)
    From the numbers $\hat{n}_{\rm sift}$, $\hat{n}_{\rm suc}$, and $\hat{n}_{\rm err}$, Alice estimates the bit error rate.
    Depending on the estimated bit error rate, Alice calculated the required number $K_{\rm EC}(\hat{n}_{\rm sift}, \hat{\Xi}_{\rm pub})$ of syndromes for error correction and the number $K_{\rm PA}^{\rm con}(\hat{n}_{\rm sift}, \hat{\Xi}_{\rm pub})$ (resp.~$K_{\rm PA}(\hat{n}_{\rm sift}, \hat{\Xi}_{\rm pub})$) of privacy amplification for the conventional (resp.~new) analysis.  If $\hat{n}_{\rm fin}\coloneqq n_{\rm fin}(\hat{n}_{\rm sift}, \hat{\Xi}_{\rm pub})$ defined in Eq.~\eqref{eq:final_key_as_function} (with $K_{\rm PA}(\hat{n}_{\rm sift}, \hat{\Xi}_{\rm pub})$ replaced with $K_{\rm PA}^{\rm con}(\hat{n}_{\rm sift}, \hat{\Xi}_{\rm pub})$ for the conventional analysis) is nonzero, then Alice sends Bob the error syndrome encrypted by consuming $K_{\rm EC}(\hat{n}_{\rm sift}, \hat{\Xi}_{\rm pub})$-bit of secret key.
    Bob performs a bit-error correction accordingly and obtains a reconciled key $\hat{\bm{k}}^{\rm rec}_B$.
    \item (Privacy amplification) Alice performs the hash function $H_{\hat{i}}$ randomly chosen from a dual 2-universal family ${\cal H}_d(\hat{n}_{\rm sift},\hat{n}_{\rm fin})=\{H_{i}\}_{i}$ of surjective linear hash functions and obtains the final key $\hat{\bm{k}}^{\rm sift}_A H_{\hat{i}}^{\top}$.  Alice sends $\hat{i}$ to Bob, and Bob performs it on his reconciled key as well to obtain the final key $\hat{\bm{k}}^{\rm rec}_B H_{\hat{i}}^{\top}$.
\end{enumerate}
In the protocol, the random variables $\hat{n}_{\rm sift}$, $\hat{n}_{\rm err}$, and $\hat{n}_{\rm suc}$ can be reconstructed from $\hat{\Xi}_{\rm pub}$, but we explicitly wrote them for later use.
The partition $\Omega_{\rm pub}^{(0)}$ and $\Omega_{\rm pub}^{(1)}$ of the range $\Omega_{\rm pub}$ of $\hat{\xi}^{(i)}_{\rm pub}$ are given by 
\begin{align}
    \Omega_{\rm pub}^{(0)} &= \{(\text{``s''}, 0), (\text{``t''}, \varnothing), (\text{``d''})\} \cup \{(\text{``t''}, (a, b))\}_{a,b\in\{0,1\}},  \\
    \Omega_{\rm pub}^{(1)} &= \{(\text{``s''}, 1)\},
\end{align}
and thus the map $\mathfrak{t}:{\cal X}_A\times{\cal X}_B\to \Omega_{\rm pub}$ is defined accordingly.
The collections of maps $\mathfrak{s}[\cdot]$ and $\mathfrak{s}'[\cdot]$ are thus singletons and defined as
\begin{align}
    &\forall a\in\{0,1\}, & \mathfrak{s}[(\text{``s''}, 1)](a)&=a,  \\ 
    &\forall b\in\{0, 1\}, & \mathfrak{s}'[(\text{``s''}, 1)](b)&=b.
\end{align}
For brevity, we omit the dependency on $(\text{``s''}, 1)$ in the following.

There are several options to ensure the correctness of the key in the information reconciliation step.  For simplicity, we assume that Alice and Bob use the error-correcting code that succeeds in unit probability if an upper bound $\hat{r}_{\rm err}\coloneqq r_{\rm err}(\hat{n}_{\rm sift}, \hat{n}_{\rm suc},\hat{n}_{\rm err})$ on the bit error rate is given and the $h(\hat{r}_{\rm err})$-bit syndrome is sent, i.e., 
\begin{equation}
K_{\rm EC}(\hat{n}_{\rm sift}, \hat{\Xi}_{\rm pub}) = \hat{n}_{\rm sift} h(\hat{r}_{\rm err}).
\end{equation}
(Note that decoding such an error-correcting code is typically inefficient. We do not care about the practicality of the protocol here.)
Thus, for the correctness to be satisfied, Alice needs to estimate an upper bound $\hat{r}_{\rm err}$ on the bit error rate with a failure probability no larger than $\varepsilon_{\rm cor}$.
From Corollary~4.2.12 in Ref.~\cite{Matsuura2023}, we have such a function $r_{\rm err}(\hat{n}_{\rm sift}, \hat{n}_{\rm suc},\hat{n}_{\rm err};\varepsilon_{\rm cor})$ satisfying 
\begin{equation}
    D\left(\frac{\hat{n}_{\rm err}}{\hat{n}_{\rm suc}} \middle\| \frac{\hat{n}_{\rm sift}r_{\rm err}(\hat{n}_{\rm sift}, \hat{n}_{\rm suc},\hat{n}_{\rm err};\varepsilon_{\rm cor}) + \hat{n}_{\rm err}}{\hat{n}_{\rm sift} + \hat{n}_{\rm suc}}\right) = -\frac{\log \varepsilon_{\rm cor}}{\hat{n}_{\rm sift} + \hat{n}_{\rm suc}}, \label{eq:defining_r}
\end{equation}
where $D(p\|q)$ denotes the binary relative entropy given by
\begin{equation}
    D(p\|q) \coloneqq p\log p - p\log q + (1 - p)\log(1-p) - (1-q)\log(1-q).
\end{equation} 

While $K_{\rm EC}(\hat{n}_{\rm sift}, \hat{\Xi}_{\rm pub}) $ is the same between the conventional analysis and our new analysis, the function $K_{\rm PA}^{\rm con}(\hat{n}_{\rm sift}, \hat{\Xi}_{\rm pub})$ is different from $K_{\rm PA}(\hat{n}_{\rm sift}, \hat{\Xi}_{\rm pub})$, which thus changes the final key length $n_{\rm fin}(\hat{n}_{\rm sift}, \hat{\Xi}_{\rm pub})$.
To determine the length $K_{\rm PA}^{\rm con}(\hat{n}_{\rm sift}, \hat{\Xi}_{\rm pub})$ or $K_{\rm PA}(\hat{n}_{\rm sift}, \hat{\Xi}_{\rm pub})$ of the required amount of privacy amplification, we introduce a virtual protocol. We give the conventional way of defining the virtual protocol in the following. Our new way of defining a virtual protocol can easily be inferred from it and the procedure described in Sec.~\ref{sec:virtual_protocol}.

\bigskip 
\noindent --- Conventional virtual protocol of the B92 protocol --- 

\noindent For each $H_i\in{\cal H}_d(\hat{n}_{\rm sift}, \hat{n}_{\rm fin})$, $\overline{H}_i$ is chosen in the way described in Sec.~\ref{sec:virtual_protocol}.
\begin{enumerate}
    \item Alice prepares a state $\ket{\Psi}_{\sysA\sysC}$ given by 
    \begin{equation}
    \ket{\Psi}_{\sysA\sysC}\coloneqq \sum_{a\in\{0,1\}}2^{-1/2}\ket{a}_{\sysA}\ket{\psi_{a}}_{\sysC}, \label{eq:initial_state_B92}
    \end{equation}
    and keeps the system $\sysA$ while sending the system $\sysC$ to Bob through the channel, where it is received by Bob as a system $\sysB$. 
    \item After $n_{\rm tot}$ rounds of quantum communication, for every $i\in\{1,\ldots,n_{\rm tot}\}$, Alice and Bob jointly access systems $A$ and $B$ at the $i$-th round to obtain random variables for the announcements while performing a joint operation specified by CP maps ${\cal M}_{\sysA\sysB}^{(\xi)}$ and ${\cal J}_{\sysA\sysB\to Q}$, which are defined as  
    \begin{align}
        {\cal M}_{\sysA\sysB}^{(\xi)}(\rho_{\sysA\sysB})
        & \coloneqq \begin{cases}
        p_{\rm s}\tr[M_{\sysB}^{\rm failure}\rho_{\sysA\sysB}], & \xi=(\text{``s''},0), \\
        p_{\rm t}\tr[M_{\sysB}^{\rm failure}\rho_{\sysA\sysB}],& \xi=(\text{``t''},\varnothing), \\
        p_{\rm t}\tr\bigl[\bigl(\ketbra{a}{a}_{\sysA}\otimes M_{\sysB}^{b}\bigr)\rho_{\sysA\sysB}\bigr], & \xi=(\text{``t''},(a, b)) \text{ for any }(a,b)\in\{0,1\}^2, \\
        p_{\rm d}\tr[\rho_{\sysA\sysB}], & \xi=(\text{``d''}),
        \end{cases} \\
        \begin{split}
        {\cal J}_{\sysA\sysB\to Q}(\rho_{\sysA\sysB}) 
        &\coloneqq \frac{p_{\rm s}}{2}\left(\sum_{j,j'\in\{0,1\}} \ketbra{j}{j}_A \otimes \ket{j\oplus j'}_{B'}\bra{\psi_{1\oplus j'}^{\perp}}_B\right) \rho_{\sysA\sysB} \\
        &\hspace{3cm}\left(\sum_{k,k'\in\{0,1\}} \ketbra{k}{k}_A \otimes \ket{\psi_{1\oplus k'}^{\perp}}_B\bra{k\oplus k'}_{B'}\right), 
        \end{split}\label{eq:filtering_B92}
    \end{align} 
    where $Q\cong A B'$, and $\oplus$ denotes the binary addition. (Since $\Omega_{\rm pub}^{(1)}$ is a singleton, we omit the dependency of ${\cal J}_{\sysA\sysB\to Q}^{(\xi)}$ on $\xi$.)
    The PVM $\{\Pi_Q[0], \Pi_Q[1]\}$ in Eq.~\eqref{eq:success_map} is defined as
    \begin{equation}
        \forall z\in\{0,1\}, \qquad \Pi_Q[z] = \ketbra{z}{z}_A \otimes I_{B'}. 
    \end{equation}
    Alice and Bob announce $\hat{\xi}_{\rm pub}^{(i)}$ for each $i=1,\ldots, n_{\rm tot}$ in the same way as that in the actual protocol.
    Let $\hat{\Xi}_{\rm pub}$, $\hat{n}_{\rm sift}$, $\hat{n}_{\rm err}$, and $\hat{n}_{\rm suc}$ be defined in the same way as those in the actual protocol.
    \item (Permutation symmetrization) 
    Alice and Bob joint-randomly reorder the quantum systems $Q^{\hat{n}_{\rm sift}}$ by consuming ${\cal O}(\hat{n}_{\rm sift}\log\hat{n}_{\rm sift})$-bit local randomness and public communication.
    \item
    From the numbers $\hat{n}_{\rm sift}$, $\hat{n}_{\rm suc}$, and $\hat{n}_{\rm err}$, Alice estimates the bit error rate.
    Depending on the estimated bit error rate, Alice calculated $K_{\rm EC}(\hat{n}_{\rm sift}, \hat{\Xi}_{\rm pub})$, $K_{\rm PA}^{\rm con}(\hat{n}_{\rm sift}, \hat{\Xi}_{\rm pub})$, and $\hat{n}_{\rm fin} = n_{\rm fin}(\hat{n}_{\rm sift}, \hat{\Xi}_{\rm pub})$. If $\hat{n}_{\rm fin}$ is nonzero, Alice sends Bob $K_{\rm EC}(\hat{n}_{\rm sift}, \hat{\Xi}_{\rm pub})$-bit random bits.
    \item Alice randomly chooses $H_{\hat{i}}\in{\cal H}_d(\hat{n}_{\rm sift},\hat{n}_{\rm fin})$, performs the unitary $U(H_{\hat{i}}, \overline{H}_{\hat{i}})$ on the qubits $A^{\hat{n}_{\rm sift}}$ of $Q^{\hat{n}_{\rm sift}}$, and measures the last $\hat{n}_{\rm sift}-\hat{n}_{\rm fin}$ qubits in the $X$ bases to obtain $\hat{\bm{c}}_A\in\mathbb{F}_2^{\hat{n}_{\rm sift}-\hat{n}_{\rm fin}}$. The unmeasured qubits are named as the system $K_A^{\hat{n}_{\rm fin}}$. Bob performs the $Z$-basis measurement on $B'^{\hat{n}_{\rm sift}}$ to obtain $\hat{\bm{z}}_{B}$. Depending on $\hat{\bm{c}}_A$ and $\hat{\bm{z}}_B$, Alice obtains an estimate $\hat{\bm{x}}_A^*\in\mathbb{F}_2^{\hat{n}_{\rm sift}}$. Alice then performs a phase error correction $Z(\hat{\bm{b}}_A)$ with $\hat{\bm{b}}_A=\hat{\bm{x}}_A^*\overline{G}_{\hat{i}}^{\top}\in\mathbb{F}_2^{\hat{n}_{\rm fin}}$ on the system $K_{\sysA}^{\hat{n}_{\rm fin}}$, which leads to the final state $\rho_{K_{\sysA}^{\hat{n}_{\rm fin}} E}^{\rm virt}$.
\end{enumerate}

The CP map ${\cal J}_{\sysA\sysB\to Q}$ in Eq.~\eqref{eq:filtering_B92} can be regarded as a coherent extraction of Bob's sifted value followed by the CNOT action from $A$ to $B'$, which leaves the phase error to the $X$ basis of $A$ and the bit error to the $Z$ basis of $B'$.
In the conventional analysis, the $Z$-basis value of $A^{\hat{n}_{\rm sift}}$ subsystem of $Q^{\hat{n}_{\rm sift}}$ is supposed to be used directly to generate the final key, unlike our new analysis, which uses that of $K^{\hat{n}_{\rm sift}r}$ after the embedding ${\cal C}_{Q\to K^rQ}^{\otimes \hat{n}_{\rm sift}}$. Thus, the PEC in the conventional analysis is done by estimating the $X$-basis value of $A^{\hat{n}_{\rm sift}}$. 
The subscript $A$ in $\hat{\bm{c}}_A$, $\hat{\bm{x}}^*_A$, and $\hat{\bm{b}}_A$ reflects this difference. 

In the early finite-size analysis of the B92 protocol~\cite{Tamaki2003}, the sequence $\hat{\bm{z}}_B$ is not used for the phase-error estimation.  However, since the bit error and the phase error can be defined simultaneously and they may be correlated, the help of the sequence $\hat{\bm{z}}_B$ may improve the estimation of the phase-error pattern. Indeed, such a correlation has been exploited in the PEC-based security analysis~\cite{Tamaki2006} for the Scarani-Acin-Ribordy-Gisin 2004 protocol~\cite{Scarani2004} and improved its performance.
Thus, we compare our new analysis with this slightly improved version of the conventional analysis, which is expected to be nearly optimal in the conventional framework, as the B92 protocol uses only the number of bit errors and sifted keys to estimate the phase errors.

The estimation protocol of the B92 protocol given above is given as follows.

\bigskip 
\noindent --- Conventional estimation protocol of the B92 protocol ---
\begin{enumerate}
    \item The same as that in the virtual protocol.
    \item[2'.] After $n_{\rm tot}$ rounds of quantum communication, for every $i\in\{1,\ldots,n_{\rm tot}\}$, Alice and Bob jointly access the systems $A$ and $B$ at the $i$-th round for obtaining random variables $(\hat{s}^{(i)},\hat{\xi}_{\rm pub}^{(i)},\hat{\theta}_{\rm aux}^{(i)})\in\{0,1\}\times\Omega_{\rm pub}\times\{0,1\}$ while performing a joint operation specified by CP maps ${\cal M}_{\sysA\sysB}^{(\xi,\theta)}$ and ${\cal J}_{\sysA\sysB\to Q}$, where ${\cal M}_{\sysA\sysB}^{(\xi,\theta)}$ is defined as
    \begin{align}
        {\cal M}_{\sysA\sysB}^{(\xi, \theta)} \coloneqq \begin{cases}
            p_{\rm s}\tr[M_{\sysB}^{\rm failure}\rho_{\sysA\sysB}], & (\xi,\theta)=((\text{``s''},0), 0) \\
        p_{\rm t}\tr[M_{\sysB}^{\rm failure}\rho_{\sysA\sysB}],& (\xi,\theta)=((\text{``t''},\varnothing), 0), \\
        p_{\rm t}\tr\bigl[\bigl(\ketbra{a}{a}_{\sysA}\otimes M_{\sysB}^{b}\bigr)\rho_{\sysA\sysB}\bigr], & (\xi,\theta)=((\text{``t''},(a, b)), 0) \text{ for any }(a,b)\in\{0,1\}^2 \\
        p_{\rm d}\tr[(I_A - M_A^{-})\rho_{\sysA\sysB}], & (\xi, \theta)=((\text{``d''}), 0), \\
        p_{\rm d}\tr[M_A^{-}\rho_{\sysA\sysB}], & (\xi, \theta)=((\text{``d''}), 1),
        \end{cases} \label{eq:unfiltered_B92}
    \end{align}
    In the above, $M_A^-$ is defined as 
    \begin{equation}
        M_A^- \coloneqq \ketbra{\widetilde{1}}{\widetilde{1}}_{A}, \label{eq:minus_POVM}
    \end{equation}
    and $\theta$ is regarded as zero for ${\cal J}_{\sysA\sysB\to Q}$. Alice and Bob announce $\hat{\xi}_{\rm pub}^{(i)}$ for each $i=1,\ldots, n_{\rm tot}$ in the same way as that in the actual protocol. Let $\hat{\Xi}_{\rm pub}$, $\hat{n}_{\rm sift}$, $\hat{n}_{\rm suc}$, and $\hat{n}_{\rm err}$ be defined in the same way as those in the actual protocol. Let $\hat{\Theta}_{\rm aux} = n_{\rm tot}P_{\hat{\bm{\theta}}_{\rm aux}}$, and let $\hat{n}_-$ be defined as 
    \begin{equation}
        \hat{n}_- \coloneqq \left|\left\{i\in\{1,\ldots,n_{\rm tot}\}:\hat{\theta}_{\rm aux}^{(i)}=1\right\}\right| = \hat{\Theta}_{\rm aux}(1). \label{eq:def_n_minus}
    \end{equation}
    \item[3--4.] The same as those in the virtual protocol.
    \item[5'.] Alice performs the $X$-basis measurements on $A^{\hat{n}_{\rm sift}}$ to obtain $\hat{\bm{x}}_A$. Bob performs the $Z$-basis measurement
    on $B'^{\hat{n}_{\rm sift}}$ to obtain $\hat{\bm{z}}_B$. Alice randomly chooses $H_{\hat{i}}\in{\cal H}_d(\hat{n}_{\rm sift},\hat{n}_{\rm fin})$, which amounts to determining $(G_{\hat{i}}, \overline{G}_{\hat{i}})$. She then computes $\hat{\bm{c}}_A=\hat{\bm{x}}_A G_{\hat{i}}^{\top}$. Depending on $\hat{\bm{c}}_A$ and $\hat{\bm{z}}_B$, they obtain an estimate $\hat{\bm{x}}^*_A$ of $\hat{\bm{x}}_A$. Alice then sets $\hat{\bm{b}}_A=\hat{\bm{x}}^*_A \overline{G}_{\hat{i}}^{\top}$ and outputs $\hat{\bm{x}}_A\overline{G}_{\hat{i}}^{\top} - \hat{\bm{b}}_A$.
\end{enumerate}
Again, the random variable $\hat{n}_-$ can be reconstructed from $\hat{\Theta}_{\rm aux}$, but we write it for later use.

\subsection{Conventional phase error correction} \label{sec:conventional_PEC}

For the conventional analysis, we aim to obtain an upper bound on the number of phase error patterns.  For this, we define an empirical probability $\hat{\bm{P}} = (\hat{P}_{00}, \hat{P}_{10}, \hat{P}_{01}, \hat{P}_{11}, \hat{P}_{\rm other})$, where $\hat{P}_{jk}$ is defined as 
\begin{equation}
    \hat{P}_{jk} = \frac{\left|\{i\in\{1,\ldots,n_{\rm tot}\}: \hat{\xi}_{\rm pub}^{(i)}=(\text{``s''},1),\hat{x}_A^{\hat{\tau}(i)}=j, \hat{z}_B^{\hat{\tau}(i)}=k\}\right|}{n_{\rm tot}}.
\end{equation}
Here, $\hat{x}_A^{\hat{\tau}(i)}$ and $\hat{z}_B^{\hat{\tau}(i)}$ are $\hat{\tau}(i)$-th elements of $\hat{\bm{x}}_A$ and $\hat{\bm{z}}_B$, respectively, where $\hat{\tau}:\{1,\ldots,n_{\rm tot}\}\to\{1,\ldots,\hat{n}_{\rm sift}\}$ denotes the sifting followed by the random permutation at Step~3. The relative frequency of the other events is thus denoted as $\hat{P}_{\rm other}$.  From the definition, we have 
\begin{align}
    n_{\rm tot}(\hat{P}_{00} + \hat{P}_{10} + \hat{P}_{01} + \hat{P}_{11}) &= \hat{n}_{\rm sift}, \\
    n_{\rm tot}(\hat{P}_{10} + \hat{P}_{11}) &= \mathrm{wt}(\hat{\bm{x}}_A), \\
    n_{\rm tot}(\hat{P}_{01} + \hat{P}_{11}) &= \mathrm{wt}(\hat{\bm{z}}_B),
\end{align}
where $\mathrm{wt}(\cdot)$ denotes the Hamming weight.
The corresponding POVM element $M_{jk}$ on the system $\sysA\sysB$ for $\hat{P}_{jk}$ at every $i$-th round is given by 
\begin{align}
    M_{jk}&={\cal J}_{\sysA\sysB\to Q}^{\dagger}(\ketbra{j}{j}_A\otimes\ketbra{k}{k}_{B'}).
\end{align}
Note that we have ${\cal J}_{\sysA\sysB\to Q}(\ketbra{\Psi}{\Psi}_{\sysA\sysB}) \propto\ketbra{00}{00}_{AB'}$, reflecting the fact that the bit and the phase error are zero when the channel is ideal.
In relation to other relevant POVM elements introduced in the virtual and the estimation protocol, we have 
\begin{align}
    \sum_{a, b: a = b} \ketbra{a}{a}_A\otimes M_B^b &= M_{00} + M_{10}, \\ 
    \sum_{a, b: a \neq b} \ketbra{a}{a}_A\otimes M_B^b &= M_{01} + M_{11}.
\end{align}

Let $\rho_{\sysA^{n_{\rm tot}}\sysB^{n_{\rm tot}}}$ be a state shared between Alice and Bob at the end of Step~1 in the conventional estimation protocol. If the empirical probability $\hat{\bm{P}}$ is contained in a convex set ${\cal A}(\hat{\Xi}_{\rm pub};\epsilon)$ of probability distributions except for a small failure probability $\epsilon$ in the estimation protocol, i.e., 
\begin{equation}
    \probrho_{\rho_{\sysA^{n_{\rm tot}}\sysB^{n_{\rm tot}}}}\!\left(\hat{\bm{P}}\notin {\cal A}(\hat{\Xi}_{\rm pub};\epsilon)\right) \leq \epsilon, \label{eq:empirical_prob_high_prob_set}
\end{equation}
then, using Lemma 1 and Eq.~(51) of Ref.~\cite{Matsuura2019}, the cardinality of the set ${\cal T}_{\epsilon}(\hat{\bm{z}}_B)$ of phase-error patterns conditioned on a bit-error pattern $\hat{\bm{z}}_B$ is bounded from above, with a high probability, by 
\begin{equation}
    \probrho_{\rho_{\sysA^{n_{\rm tot}}\sysB^{n_{\rm tot}}}}\!\left(|{\cal T}_{\epsilon}(\hat{\bm{z}}_B)| \leq \max_{\hat{\bm{P}}\in{\cal A}(\hat{\Xi}_{\rm pub};\epsilon)} 2^{n_{\rm tot} H({\rm ph}|{\rm bit})_{\hat{\bm{P}}}}\right) \geq 1 -\epsilon, \label{eq:set_possible_phase_err_patterns}
\end{equation}
where 
\begin{equation}
    H({\rm ph}|{\rm bit})_{\hat{\bm{P}}} \coloneqq \frac{n_{\rm tot}(\hat{P}_{00} + \hat{P}_{10})}{\hat{n}_{\rm sift}} h\!\left(\frac{\hat{P}_{00}}{\hat{P}_{00} + \hat{P}_{10}}\right) + \frac{n_{\rm tot}(\hat{P}_{01} + \hat{P}_{11})}{\hat{n}_{\rm sift}} h\!\left(\frac{\hat{P}_{01}}{\hat{P}_{01} + \hat{P}_{11}}\right),
\end{equation}
with a binary entropy function $h(p)\coloneqq -p\log p - (1 - p)\log(1 - p)$.
Then, from the conventional argument of the phase error correction, setting $K_{\rm PA}^{\rm con}(\hat{n}_{\rm sift},\hat{\Xi}_{\rm pub})$ to
\begin{equation}
    K_{\rm PA}^{\rm con}(\hat{n}_{\rm sift},\hat{\Xi}_{\rm pub}) = \max_{\hat{\bm{P}}\in{\cal A}(\hat{\Xi}_{\rm pub};\epsilon)} n_{\rm tot}H({\rm ph}|{\rm bit})_{\hat{\bm{P}}} + s \label{eq:amount_conventional_PA}
\end{equation}
for the privacy amplification at Step~5' ensures 
\begin{equation}
    \probrho_{\rho_{\sysA^{n_{\rm tot}}\sysB^{n_{\rm tot}}}}(\hat{n}_{\rm fin}\geq 1, \hat{\bm{x}}_A\neq \hat{\bm{x}}_A^*) \leq \epsilon + 2^{-s}, \label{eq:conventional_PER_bound}
\end{equation}
which then leads to $\sqrt{2(\epsilon+2^{-s})}$-secrecy~\cite{Koashi2009,Matsuura2023}. 
Thus, in the following, we find a convex set ${\cal A}(\hat{\Xi}_{\rm pub};\epsilon)$ that satisfies Eq.~\eqref{eq:empirical_prob_high_prob_set}.
This can be done by finding a family $\{V_{\epsilon_t,\epsilon_v}(\Xi)\}_{\Xi}$ of subsets of density operators as in Eq.~\eqref{eq:hull}, which gives restrictions for the allowed empirical probability in expectation, and applying another concentration inequality to evaluate the deviation of the empirical probability from the expectation in combination with the i.i.d.~reduction.

Since Alice prepares the state $\ket{\Psi}_{\sysA\sysC}$ given in Eq.~\eqref{eq:initial_state_B92} at each round, which satisfies 
\begin{equation} \|\bra{\widetilde{1}}_{\sysA}\ket{\Psi}_{\sysA\sysC}\|^2=\alpha^2,
\end{equation} 
the random variable $\hat{n}_{-}$ defined in the estimation protocol through Eqs.~\eqref{eq:unfiltered_B92} and \eqref{eq:minus_POVM} satisfies, from the Chernoff-Hoeffding bound~\cite{Hoeffding1963}, 
\begin{align}
    \probrho_{\rho_{\sysA^{n_{\rm tot}}\sysB^{n_{\rm tot}}}}\!\left(\hat{n}_{-} > n_{\rm tot}(p_{\rm d}\alpha^2 + \delta_1(p_{\rm d}\alpha^2,n_{\rm tot},\varepsilon_1))\right) 
    & = P_{p_{\rm d}\alpha^2}^{\times n_{\rm tot}}\!\left(\hat{n}_{-} > n_{\rm tot}(p_{\rm d}\alpha^2 + \delta_1(p_{\rm d}\alpha^2,n_{\rm tot},\varepsilon_1))\right) \\
    & \leq \varepsilon_1, \label{eq:n_minus_out}
\end{align}
where the positive function $\delta_1(p,n,\varepsilon)$ is defined to satisfy~\cite{Matsuura2021, Matsuura2023}
\begin{equation}
    \begin{cases}
    -\log\varepsilon = 
    n D(p + \delta_1(p,n,\varepsilon) \| p) & \text{if }\varepsilon \geq p^{n},\\
        \delta_1(p,n,\varepsilon) = 1 - p & \text{Otherwise}.
    \end{cases}
\end{equation}
This allows us to restrict our attention to the case $\hat{n}_{-} \leq n_{\rm tot}(p_{\rm d}\alpha^2 + \delta_1(p_{\rm d}\alpha^2,n_{\rm tot},\varepsilon_1))$ allowing the failure probability $\varepsilon_1$, which amounts to determining $\Upsilon_{\varepsilon_1}$ in Eq.~\eqref{eq:bound_out_of_Upsilon}, i.e.,
\begin{equation}
    \Upsilon_{\varepsilon_1} = \left\{(\Xi,\Theta):\Theta(1) \leq n_{\rm tot}\!\left(p_{\rm d}\alpha^2 + \delta_1(p_{\rm d}\alpha^2,n_{\rm tot},\varepsilon_1)\right)\right\}, \label{eq:Upsilon_B92}
\end{equation}
where we used Eq.~\eqref{eq:def_n_minus}.
Then, for any set ${\cal A}$ of empirical probabilities, we have 
\begin{equation}
    \probrho_{\rho_{\sysA^{n_{\rm tot}}\sysB^{n_{\rm tot}}}}\!\left(\hat{\bm{P}}\notin{\cal A}\right) \leq \varepsilon_1 + \probrho_{\rho_{\sysA^{n_{\rm tot}}\sysB^{n_{\rm tot}}}}\!\left((\hat{\Xi}_{\rm pub},\hat{\Theta}_{\rm aux})\in\Upsilon_{\varepsilon_1},\hat{\bm{P}}\notin{\cal A}\right). \label{eq:in_Upsilon_B92}
\end{equation}

From the permutation symmetry of the protocol, one can apply Eq.~\eqref{eq:bound_by_iid} to replace the state $\rho_{\sysA^{n_{\rm tot}}\sysB^{n_{\rm tot}}}$ that Alice and Bob share with a mixture of i.i.d.~quantum states $\{\rho^{\otimes n_{\rm tot}}\}_{\rho \in {\cal D}({\cal H}_{\sysA\sysB})}$, i.e., 
\begin{equation}
    \begin{split}
    &\probrho_{\rho_{\sysA^{n_{\rm tot}}\sysB^{n_{\rm tot}}}}\!\left((\hat{\Xi}_{\rm pub},\hat{\Theta}_{\rm aux})\in\Upsilon_{\varepsilon_1},\hat{\bm{P}}\notin{\cal A}\right) \\
    &\leq f_q(n_{\rm tot}, 4) \max_{\rho\in{\cal D}({\cal H}_{\sysA\sysB})} \probrho_{\rho^{\otimes n_{\rm tot}}}\!\left((\hat{\Xi}_{\rm pub},\hat{\Theta}_{\rm aux})\in\Upsilon_{\varepsilon_1},\hat{\bm{P}}\notin{\cal A}\right). \label{eq:iid_reduction_B92}
    \end{split}
\end{equation}  
Now, we will construct the family of sets $\{V_{\epsilon_t,\epsilon_v}(\Xi)\}_{\Xi}$ by following the procedures of Eqs.~\eqref{eq:split_conditions} and \eqref{eq:respective_conditions} in Remark~\ref{rem:optimization_region}. First, to construct $(\Xi,\Theta)$-dependent sets $\{\rho\in{\cal D}({\cal H}_{AB}):(\Xi,\Theta)\in U_{\epsilon_{v_j}}^j(\rho)\}$, we use the following lemma. 
\begin{lemma} \label{lem:set_compatible_with_outcome}
    Let $\epsilon$ be a constant with $0<\epsilon<1$ and let $n$ be a natural number.  Let $\{1-M, M\}$ be a POVM and $\hat{X}\in \{0,1\}$ be the measurement outcome.  Let $\{{\cal B}(q)\}_{q > 0}$ be an increasing net of convex sets of density operators under inclusion indexed by a positive real number $q$ defined as 
    \begin{equation}
        {\cal B}(q)\coloneqq \left\{\rho\in{\cal D}({\cal H}): \tr [\rho M] \leq q \right\}. \label{eq:precondition}
    \end{equation}
    Then, there exists a positive monotone increasing function $q_{n,\epsilon}(p)$ of $p\in[0, 1)$ such that $q_{n,\epsilon}(p)>p$ and for any density operator $\rho\notin{\cal B}(q_{n,\epsilon}(p))$, we have 
    \begin{equation}
        P_{\hat{X}\sim \{\tr[\rho (1-M)],\tr[\rho M]\}}^{\times n}\left[\frac{1}{n}\sum_{i=1}^n \hat{X}_i \leq p\right] \leq \epsilon, \label{eq:prob_obtaining_type}
    \end{equation}
    where $\{\hat{X}_i\}_{i=1}^n$ denotes i.i.d.~binary random variables with each probability distribution given by $\{\tr[\rho (1-M)],\tr[\rho M]\}$.  More explicitly, the function $q_{n,\epsilon}(p)$ is defined to satisfy the following:
    \begin{equation}
        q_{n,\epsilon}(p)>p \text{ and } -\log\epsilon =  n D(p\|q_{n,\epsilon}(p)). \label{eq:defining_delta_2}
    \end{equation}
\end{lemma}
\begin{proof}
    Let $q_{n,\epsilon}(p)$ be as defined in Eq.~\eqref{eq:defining_delta_2}.
    Fix a density operator $\rho\notin{\cal B}(q_{n,\epsilon}(p))$.  Then, from a Chernoff-Hoeffding bound~\cite{Hoeffding1963}, we have 
    \begin{equation}
        P_{\hat{X}\sim \{\tr[\rho (1-M)],\tr[\rho M]\}}^{\times n}\left[\frac{1}{n}\sum_{i=1}^n \hat{X}_i \leq p\right] \leq 2^{-n D(p\|\tr[\rho M])}.
    \end{equation}
    Since the binary relative entropy $D(p\|q)$ is monotone increasing for $q$ when $q > p$ with a fixed $p$, we have 
    \begin{equation}
        P_{\hat{X}\sim \{\tr[\rho (1-M)],\tr[\rho M]\}}^{\times n}\left[\frac{1}{n}\sum_{i=1}^n \hat{X}_i \leq p\right] \leq \max_{\rho\notin{\cal B}(q_{n,\epsilon}(p))} 2^{-n D(p\|\tr[\rho M])} \leq 2^{-n D(p\|q_{n,\epsilon}(p))} = \epsilon,
    \end{equation}
    where we used Eq.~\eqref{eq:defining_delta_2} for the last equality.
    Since the above statement holds for any $\rho\notin{\cal B}(q_{n,\epsilon}(p))$, Eq.~\eqref{eq:prob_obtaining_type} holds.

    The fact that $q_{n,\epsilon}(p)$ is monotone increasing for $p$ can be shown via a proof by contradiction. For $0<p_1<p_2$, assume that $q_{n,\epsilon}(p_1) \geq q_{n,\epsilon}(p_2)$ holds. Then, from the definition of the binary relative entropy, $D(p\|q_1) < D(p\|q_2)$ for $p<q_1<q_2$ and $D(p_1\|q) > D(p_2\|q)$ for $p_1 < p_2 < q$ hold. From the chain of inequalities $0<p_1<p_2<q_{n,\epsilon}(p_2) \leq q_{n,\epsilon}(p_1)$, we have 
    \begin{equation}
        -\log\epsilon = n D(p_1\|q_{n,\epsilon}(p_1)) \geq n D(p_1\|q_{n,\epsilon}(p_2)) > nD(p_2 \| q_{n,\epsilon}(p_2)) = -\log\epsilon,
    \end{equation}
    which is the contradiction.
\end{proof}
The above lemma immediately implies that for any $\rho\in{\cal D}({\cal H})$, we have 
\begin{equation}
    P_{\hat{X}\sim \{\tr[\rho (1-M)],\tr[\rho M]\}}^{\times n}\left[\rho \notin {\cal B}\left(q_{n,\epsilon}\Bigl(\frac{1}{n}\sum_{i=1}^n \hat{X}_i\Bigr)\right)\right] \leq \epsilon, \label{eq:prob_out_of_set}
\end{equation}
since ${\cal B}(q_1)\subseteq {\cal B}(q_2)$ for $q_1 \leq q_2$ and $q_{n,\epsilon}(p)$ is monotone increasing for $p$. To apply the above lemma to the B92 protocol, suppose $M$ in the lemma as ${\cal J}_{AB\to Q}^{\dagger}(I_Q)$, $\hat{X}_i$ as $\hat{s}^{(i)}$, and $n$ as $n_{\rm tot}$. Then, we have, for any $\rho\in{\cal D}({\cal H}_{\sysA\sysB})$,
\begin{equation}
    P_{\hat{s}\sim\{\tr[\rho(1-{\cal J}_{AB\to Q}^{\dagger}(I_Q))],\tr[\rho {\cal J}_{AB\to Q}^{\dagger}(I_Q)]\}}^{\times n_{\rm tot}} \left[\rho\notin{\cal B}\left(q_{n_{\rm tot},\varepsilon_2}(\hat{n}_{\rm sift}/n_{\rm tot})\right)\right] \leq \varepsilon_2. \label{eq:example_range_states}
\end{equation}
From Eq.~\eqref{eq:precondition}, this implies that, given $\hat{n}_{\rm sift}$, we can restrict a set of density operators compatible with the observed value of $\hat{n}_{\rm sift}$ as $\{\rho\in{\cal D}({\cal H}_{\sysA\sysB}):\tr[ {\cal J}_{AB\to Q}(\rho)]\leq q_{n_{\rm tot},\varepsilon_2}(\hat{n}_{\rm sift}/n_{\rm tot})\}$, allowing the failure probability $\varepsilon_2$. 
Since this set is $\hat{\Theta}_{\rm aux}$-independent as $\hat{n}_{\rm sift}=\hat{\Xi}_{\rm pub}((\text{``s''}, 1))$, it directly plays the role of $V^j_{\epsilon_{t_j},\epsilon_{v_j}}(\hat{\Xi}_{\rm pub})$ in Eq.~\eqref{eq:respective_conditions} of Remark~\ref{rem:optimization_region} with $\epsilon_{t_j}=0$ and $\epsilon_{v_j}=\varepsilon_2$. We thus define 
\begin{equation}
    V^1_{0,\varepsilon_2}(\hat{\Xi}_{\rm pub}) \coloneqq \{\rho\in{\cal D}({\cal H}_{\sysA\sysB}):\tr[ {\cal J}_{AB\to Q}(\rho)]\leq q_{n_{\rm tot},\varepsilon_2}(\hat{n}_{\rm sift}/n_{\rm tot})\}. \label{eq:V^1}
\end{equation}
By applying similar arguments to $1-\hat{s}^{(i)}$ and $\hat{\xi}_{\rm pub}^{(i)}$, we have 
\begin{align}
    V^2_{0,\varepsilon_2}(\hat{\Xi}_{\rm pub}) &\coloneqq \left\{\rho\in{\cal D}({\cal H}_{\sysA\sysB}):1 - \tr[{\cal J}_{AB\to Q}(\rho)] \leq q_{n_{\rm tot},\varepsilon_2}(1 - \hat{n}_{\rm sift}/n_{\rm tot})\right\}, \\
    \begin{split}
    V^3_{0,\varepsilon_2}(\hat{\Xi}_{\rm pub}) &\coloneqq \left\{\rho\in{\cal D}({\cal H}_{\sysA\sysB}):\left({\cal M}_{\sysA\sysB}^{((\text{``t''},(1,0)),0)} + {\cal M}_{\sysA\sysB}^{((\text{``t''},(0,1)),0)}\right)(\rho) \leq q_{n_{\rm tot},\varepsilon_2}(\hat{n}_{\rm err}/n_{\rm tot})\right\},
    \end{split}
\end{align}
which respectively contains all the possible states except those with at most $\varepsilon_2$ realization probabilities.
We also apply Lemma~\ref{lem:set_compatible_with_outcome} to $\hat{\theta}_{\rm aux}^{(i)}$ to obtain the $\hat{\Theta}_{\rm aux}$-dependent set $\{\rho\in{\cal D}({\cal H}_{\sysA\sysB}): {\cal M}_{\sysA\sysB}^{((\text{``d''}), 1)}(\rho) \leq q_{n_{\rm tot},\varepsilon_2}(\hat{n}_-/n_{\rm tot})\}$, which contains all the relevant states except those that give at most $\varepsilon_2$ realization probability. Combining this with Eqs.~\eqref{eq:def_n_minus} and \eqref{eq:Upsilon_B92}, we can define 
\begin{align}
    V^4_{\varepsilon_1,\varepsilon_2}(\Xi) &\coloneqq \bigcup_{(\Xi,\Theta)\in\Upsilon_{\varepsilon_1}}\{\rho\in{\cal D}({\cal H}_{\sysA\sysB}): {\cal M}_{\sysA\sysB}^{((\text{``d''}), 1)}(\rho) \leq q_{n_{\rm tot},\varepsilon_2}({\Theta}_{\rm aux}(1)/n_{\rm tot})\} \\
    &= \{\rho\in{\cal D}({\cal H}_{\sysA\sysB}): {\cal M}_{\sysA\sysB}^{((\text{``d''}), 1)}(\rho) \leq q_{n_{\rm tot},\varepsilon_2}(p_{\rm d}\alpha^2+\delta_1(p_{\rm d}\alpha^2,n_{\rm tot},\varepsilon_1))\}, 
\end{align}
which is a $\hat{\Xi}_{\rm pub}$-independent convex set. Now, we define $\{V_{\varepsilon_1,4\varepsilon_2}(\Xi)\}_{\Xi}$ as 
\begin{equation}
    V_{\varepsilon_1,4\varepsilon_2}(\Xi)\coloneqq \biggl(\bigcap_{j=1}^3 V_{0,\varepsilon_2}^{j}(\Xi)\biggr) \cap V_{\varepsilon_1,\varepsilon_2}^4(\Xi), \label{eq:intersection_Vs}
\end{equation}
which then gives 
\begin{equation}
    \begin{split}
        &\max_{\rho\in{\cal D}({\cal H}_{\sysA\sysB})} \probrho_{\rho^{\otimes n_{\rm tot}}}\!\left((\hat{\Xi}_{\rm pub},\hat{\Theta}_{\rm aux})\in\Upsilon_{\varepsilon_1},\hat{\bm{P}}\notin{\cal A}\right) \\
        &\leq 4\varepsilon_2 + \max_{\rho\in{\cal D}({\cal H}_{\sysA\sysB})} \probrho_{\rho^{\otimes n_{\rm tot}}}\!\left(\rho\in V_{\varepsilon_1,4\varepsilon_2}(\hat{\Xi}_{\rm pub}),\hat{\bm{P}}\notin{\cal A}\right).
    \end{split} \label{eq:in_V_except_small}
\end{equation}

Now, for $q\in[0,1]$, let ${\cal P}$ be the set of probability vectors with five elements, and let ${\cal A}[\bm{\gamma},q]$ be its subset defined with a vector $\bm{\gamma}=(\gamma_1,\gamma_2,\gamma_3,\gamma_4,\gamma_5)$ as 
\begin{align}
    {\cal A}[\bm{\gamma},q] \coloneqq  \left\{\bm{p}\in{\cal P}: \bm{\gamma}\cdot\bm{p}  \leq 0, p_5 = 1 - q \right\}. \label{eq:def_A}
\end{align}
Then, for a given $\hat{\Xi}_{\rm pub}=\Xi$, an upper bound on the probability that the empirical probability $\hat{\bm{P}}$ is not in ${\cal A}[\bm{\gamma}, \hat{n}_{\rm sift} / n_{\rm tot}]$ is given from Sanov's theorem~\cite{Sanov1957,Csiszar1984} by 
\begin{equation}
    \forall \rho \in V_{\varepsilon_1,4\varepsilon_2}(\Xi), \quad  P_{\rho}^{\times n_{\rm tot}}\!\left(\hat{\bm{P}} \notin {\cal A}[\bm{\gamma},\hat{n}_{\rm sift} / n_{\rm tot}]\right) \leq \sup_{\bm{p}\notin {\cal A}[\bm{\gamma},\hat{n}_{\rm sift}/ n_{\rm tot}]} 2^{- n_{\rm tot}D(\bm{p} \| \bm{q}(\rho))}, \label{eq:bound_empirical_prob}
\end{equation}
where $\bm{q}(\rho)$ is given by 
\begin{equation}
    \begin{split}
    \bm{q}(\rho) = \Bigl(\tr[\rho M_{00}], \tr[\rho M_{10}], \tr[\rho M_{01}], \tr[\rho M_{11}], \tr[\rho(1 - {\cal J}_{\sysA\sysB\to Q}^{\dagger}(I_Q))] \Bigr).
    \end{split} \label{eq:def_q_rho}
\end{equation}
Let $\bm{\gamma}^*(\Xi;\varepsilon_2)$ be a solution that satisfies, for any $\Xi$,
\begin{equation}
    \max_{\rho\in V_{\varepsilon_1,4\varepsilon_2}(\Xi)} \sup_{\bm{p}\notin {\cal A}[\bm{\gamma},\hat{n}_{\rm sift}/ n_{\rm tot}]}  2^{- n_{\rm tot}D(\bm{p} \| \bm{q}(\rho))} \leq \varepsilon_2, \label{eq:solution_sanov}
\end{equation}
which is not unique.  We come back to this non-uniqueness later, but any solution is allowed for the security proof.  The right-hand side of the above equation includes a nonlinear convex SDP, but it can be solved as a nonlinear convex optimization problem due to the small size of the matrix involved.
Combining Eqs.~\eqref{eq:bound_empirical_prob} and \eqref{eq:solution_sanov}, we have
\begin{equation}
    \max_{\rho\in{\cal D}({\cal H}_{AB})}P_{\rho}^{\times n_{\rm tot}}\!\left(\rho \in V_{\varepsilon_1,4\varepsilon_2}(\hat{\Xi}_{\rm pub}),\hat{\bm{P}} \notin {\cal A}[\bm{\gamma}^*(\hat{\Xi}_{\rm pub};\varepsilon_2),\hat{n}_{\rm sift} / n_{\rm tot}]\right) \leq \varepsilon_2. \label{eq:small_by_sanov}
\end{equation}
Thus, by combining Eqs.~\eqref{eq:in_Upsilon_B92}, \eqref{eq:iid_reduction_B92}, \eqref{eq:in_V_except_small}, and \eqref{eq:small_by_sanov}, we have
\begin{equation}
    \probrho_{\rho_{\sysA^{n_{\rm tot}}\sysB^{n_{\rm tot}}}}\!\left(\hat{\bm{P}} \notin {\cal A}[\bm{\gamma}^*(\hat{\Xi}_{\rm pub};\varepsilon_2),\hat{n}_{\rm sift}/n_{\rm tot}]\right) \leq  \varepsilon_1 + 5\varepsilon_2f_q(n_{\rm tot}, 4),
\end{equation}
which then shows that ${\cal A}[\bm{\gamma}^*(\hat{\Xi}_{\rm pub};\varepsilon_2),\hat{n}_{\rm sift}/n_{\rm tot}]$ plays the role of ${\cal A}(\hat{\Xi}_{\rm pub};\epsilon)$ in Eq.~\eqref{eq:empirical_prob_high_prob_set} with $\epsilon=\varepsilon_1 + 5\varepsilon_2f_q(n_{\rm tot}, 4)$. (Notice that $\hat{n}_{\rm sift}$ can be obtained from $\hat{\Xi}_{\rm pub}$.)
Thus, from Eqs.~\eqref{eq:empirical_prob_high_prob_set}, \eqref{eq:amount_conventional_PA}, and~\eqref{eq:conventional_PER_bound}, setting $K_{\rm PA}^{\rm con}(\hat{n}_{\rm sift},\hat{\Xi}_{\rm pub})$ to 
\begin{equation}
    K_{\rm PA}^{\rm con}(\hat{n}_{\rm sift},\hat{\Xi}_{\rm pub}) = \max_{\hat{\bm{P}} \in {\cal A}[\bm{\gamma}^*(\hat{\Xi}_{\rm pub};\varepsilon_2),\hat{n}_{\rm sift}/n_{\rm tot}]} n_{\rm tot}H({\rm ph}|{\rm bit})_{\hat{\bm{P}}} + s \label{eq:key_rate_conventional}
\end{equation}
for the privacy amplification ensures $\sqrt{2(\varepsilon_1 + 5\varepsilon_2f_q(n_{\rm tot}, 4) + 2^{-s})}$-secrecy.  To obtain a better key rate, one needs to optimize the dual parameters $\{\bm{\gamma}^*(\Xi;\varepsilon_2)\}_{\Xi}$, which exploits the arbitrariness of the solution of Eq.~\eqref{eq:solution_sanov}.  One can numerically solve this final optimization problem to minimize $K_{\rm PA}^{\rm con}(\hat{n}_{\rm sift},\hat{\Xi}_{\rm pub})$, which results in maximizing the final key length, but any heuristic choice of $\bm{\gamma}^*(\hat{\Xi}_{\rm pub};\varepsilon_2)$ leads to a lower bound on the secure key rate at least.

\subsection{New phase error correction based on the universal coding} \label{sec:new_PEC}
In our newly developed phase error correction based on the universal decoding of classical source compression with quantum side information, we will follow the same procedure until Step~5 in the virtual protocol and then apply the isometry $C_{Q\to KQ}^{\otimes \hat{n}_{\rm sift}}$ to extract a system $K^{\hat{n}_{\rm sift}}$ with $\dim{\cal H}_K=2$. The estimation protocol should be modified accordingly. Note that the conventional PEC described in the previous section can also be regarded as applying $C_{Q\to KQ}^{\otimes \hat{n}_{\rm sift}}$ first and then performing the $X$-basis and $Z$-basis measurements on the subsystems $A^{\hat{n}_{\rm sift}}$ and $B'^{\hat{n}_{\rm sift}}$ of $Q^{\hat{n}_{\rm sift}}$, respectively, to estimate the $X$-basis values $\hat{\bm{x}}$ of the system $K^{\hat{n}_{\rm sift}}$ in the estimation protocol, since the $X$-basis measurement of the subsystem $A^{\hat{n}_{\rm sift}}$ in this scenario amounts to observing $\hat{\bm{x}}+\hat{\bm{x}}_A$ from Eqs.~\eqref{eq:Z_list}--\eqref{eq:Z_restriction}, and thus estimating $\hat{\bm{x}}_A$ is equivalent to estimating $\hat{\bm{x}}$.

Applying the unitary $U(H_{\hat{i}},\overline{H}_{\hat{i}})$ on $K^{\hat{n}_{\rm sift}}$ and measuring the last $\hat{n}_{\rm sift}-\hat{n}_{\rm fin}$ qubits in the $X$ bases, Alice and Bob perform universal source compression with quantum side information by the measurement outcomes and the quantum state on $Q^{\hat{n}_{\rm sift}}$ to identify $\hat{\bm{x}}$.
We use the same sets $\Upsilon_{\varepsilon_1}$ in Eq.~\eqref{eq:Upsilon_B92} and $\{V_{\varepsilon_1,4\varepsilon_2}(\Xi)\}_{\Xi}$ in Eq.~\eqref{eq:intersection_Vs} for the new analysis.
Then, from Eq.~\eqref{eq:choice_of_final_key}, the function $K_{\rm PA}(\hat{n}_{\rm sift},\hat{\Xi}_{\rm pub})$ is set to be, for an arbitrary $t\in[0,1]$,
\begin{equation}
    \begin{split}
    K_{\rm PA}(\hat{n}_{\rm sift},\hat{\Xi}_{\rm pub}) &= \left\lceil \max_{\rho \in V_{\varepsilon_1,4\varepsilon_2}(\hat{\Xi}_{\rm pub})} n_{\rm tot} H_{1-t}^{\uparrow,\leq}(X|Q)_{\mathfrak{x}_{XQ}(\rho_{\sysA\sysB})}  + 18\log(\hat{n}_{\rm sift}+1) +\frac{\log(1/\varepsilon_2)}{t} \right\rceil,
    \end{split} \label{eq:key_rate_new}
\end{equation}
where we used $p=2$, $r=1$, and $d_{Q}=4$ in this case, and
\begin{equation}
    \mathfrak{x}_{XQ}(\rho_{\sysA\sysB}) = {\cal P}_{K\to X}\circ {\cal C}_{Q\to KQ}\circ {\cal J}_{\sysA\sysB\to Q}(\rho_{\sysA\sysB}).
\end{equation}
Then, we have, from Eqs.~\eqref{eq:reinterpretation_prop}--\eqref{eq:remaining_term_bound},
\begin{equation}
    \max_{\rho\in{\cal D}({\cal H}_{AB})} \probrho_{\rho^{\otimes n}}\left(\rho\in V_{\varepsilon_1,4\varepsilon_2}(\hat{\Xi}_{\rm pub}),\hat{n}_{\rm fin} \geq 1, \hat{\bm{x}}\neq\hat{\bm{x}}^{*}\right) \leq \varepsilon_2. \label{eq:failure_source_compression_B92}
\end{equation}
We heuristically minimize the right-hand side of Eq.~\eqref{eq:key_rate_new} over $t\in[0,1]$ as mentioned earlier.  Finally, combining Eqs.~\eqref{eq:failure_phase_error_correction}, \eqref{eq:bound_by_iid}, \eqref{eq:connection_to_source_compression}, \eqref{eq:n_minus_out}, \eqref{eq:Upsilon_B92}, \eqref{eq:example_range_states}--\eqref{eq:intersection_Vs}, and \eqref{eq:failure_source_compression_B92}, we have
\begin{equation}
    \probrho_{\rho_{A^{n_{\rm tot}}B^{n_{\rm tot}}}}(\hat{n}_{\rm fin}\geq 1, \hat{\bm{x}}\neq\hat{\bm{x}}^*) \leq \varepsilon_1 + 5\varepsilon_2 f_q(n_{\rm tot}, 4),
\end{equation}
which then implies the $\sqrt{2(\varepsilon_1 + 5\varepsilon_2 f_q(n_{\rm tot}, 4))}$-secrecy of the actual protocol.

\subsection{Numerical simulation under depolarizing channel} \label{sec:numerical_simulation}
We compare the key rates obtained through Eqs.~\eqref{eq:key_rate_conventional} and \eqref{eq:key_rate_new} with the conventional and our new analyses, respectively, under the depolarizing channel ${\cal N}_p$ defined as
\begin{equation}
    {\cal N}_p(\rho_B) = (1 - p) \rho_B + p\frac{I_B}{2}. \label{eq:depolarizing_channel}
\end{equation} 
For simplicity, we set $p_{\rm s}=p_{\rm t}=p_{\rm d}=1/3$.  Furthermore, we set $\alpha=0.38$, which is a good choice for a relatively high depolarizing parameter $p\sim 4.5\%$~\cite{Tamaki2003, Renner2005, Matsumoto2013, Sasaki2015}.  The correctness and secrecy parameters $\varepsilon_{\rm cor}$ and $\varepsilon_{\rm sec}$ are both set to be $2^{-50}$.

Although the optimization problems to obtain $K_{\rm PA}^{\rm con}$ and $K_{\rm PA}$ in Eqs.~\eqref{eq:solution_sanov} and \eqref{eq:key_rate_conventional} for the conventional analysis and Eq.~\eqref{eq:key_rate_new} for the new analysis both include nonlinear convex SDP, the problems for the conventional analysis can be solved by the convex optimization package CVXPY~\cite{Steven2016} with a reasonable amount of time by reformulating them to entry-wise convex optimization problems.
To find a good choice for the dual parameters $\bm{\gamma}^*(\hat{\Xi}_{\rm pub};\varepsilon_2)$ in Eqs.~\eqref{eq:key_rate_conventional}, we first solve the following optimization problem: 
\begin{equation}
    \max_{\bm{p}:\bm{p}=\bm{q}(\rho),\exists\rho\in V_{\varepsilon_1,4\varepsilon_2}(\hat{\Xi}_{\rm pub})} H({\rm ph}|{\rm bit})_{\bm{p}},
\end{equation}
where $\bm{q}(\rho)$ is defined in Eq.~\eqref{eq:def_q_rho}. Then, the dual solution for this problem gives coefficients $\bm{\gamma}_{\rm sol}$ as in Eq.~\eqref{eq:def_A} from the equality constraint $\bm{p}=\bm{q}(\rho)$. We shift these coefficients $\bm{\gamma}_{\rm sol}$ by $c$, i.e., we set $\bm{\gamma}^* = \bm{\gamma}_{\rm sol} - (c, c, c, c, c)$, so that Eq.~\eqref{eq:solution_sanov} is satisfied. This shift of coefficients leads to a parallel shift of the affine boundary of Eq.~\eqref{eq:def_A} from the tangent hyperplane $\{\bm{p}:\bm{\gamma}_{\rm sol}\cdot\bm{p}=0\}$ of $\{\bm{p}:\bm{p}=\bm{q}(\rho),\exists\rho\in V_{\varepsilon_1,4\varepsilon_2}(\hat{\Xi}_{\rm pub})\}$ (recall that $\bm{p}\cdot(c,c,c,c,c)=c$).
While this heuristic choice of $\bm{\gamma}^*$ may not be optimal in general, it should be close to the optimal choice when $n_{\rm tot}\gg 1$.

The problem Eq.~\eqref{eq:key_rate_new} for the new analysis, on the other hand, involves nonlinear matrix functions and is not easily handled by convex optimization packages such as CVXPY~\cite{Steven2016}.  We thus use the QICS package~\cite{He2024}, which is developed for optimization problems involving the quantum relative entropy and the quantum R\'enyi divergences.

\begin{figure}[t]
    \centering
    \includegraphics[width=0.98\linewidth]{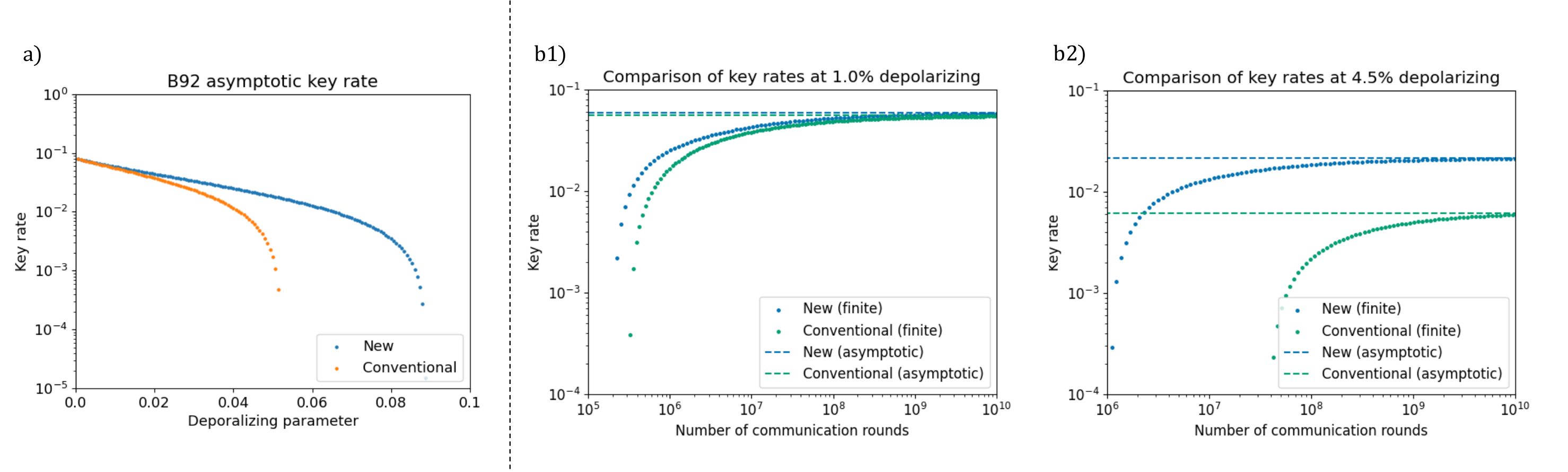}
    \caption{Comparison of the key rates with the conventional PEC-type analysis and our new analysis based on the universal coding for the B92 protocol under the depolarizing channel.  Figure~a) shows the comparison of the asymptotic key rates of the conventional analysis and our new analysis under the depolarizing channel with various depolarizing parameters $p$ in Eq.~\eqref{eq:depolarizing_channel}.  Figure~b1) and~b2) show the comparison of the finite-size key rates between the two with the depolarizing parameter $p$ taken to be 1\% in b1) and 4.5\% in b2). As the figure suggests, the asymptotic key rates differ largely when the depolarizing parameter is large, meaning that the bit-error rate is large.  Furthermore, our new analysis has finite-size key rates superior to the conventional analysis as suggested in Figs.~b1) and b2). For a large depolarizing parameter as in Fig.~b2), we have a two orders of magnitude improvement for the minimum number of rounds to generate a nonzero key.}
    \label{fig:key_rate_comparison}
\end{figure}

Figure~\ref{fig:key_rate_comparison} shows the comparison of the key rates of the conventional analysis and our new analysis in the asymptotic case against the depolarizing parameter $p$ in a), and in the finite-size case against the number of communication rounds with the depolarizing parameter $p$ taken to be 1\% in b1) and 4.5\% in b2). 
As expected from previous works~\cite{Tamaki2003, Renner2005, Matsumoto2013, Sasaki2015}, the asymptotic key rates have a large gap between the two analyses in the high bit-error regime, while the gap closes in the limit of the vanishing bit error.  (Note that a higher depolarizing parameter leads to a higher bit-error rate.) 
In Figs.~\ref{fig:key_rate_comparison} b1) and b2), our new analysis even improves finite-size key rates upon the conventional analysis, and the improvement is larger as the depolarizing parameter gets larger, as Fig.~\ref{fig:key_rate_comparison} b2) suggests.  We should note, however, that the finite-size performance of the conventional analysis has room for improvement since the i.i.d.~reduction performed in Sec.~\ref{sec:conventional_PEC} may not be the best strategy we can take as suggested in Ref.~\cite{Matsuura2024}.

\section{Discussion} \label{sec:discussion}
In this paper, we have constructed a fully as well as a partially universal decoder for classical source compression with quantum side information.  We apply this construction to the security analysis of the QKD and develop a new security-proof strategy.  
In the usual setup of universal source compression, the entropy of a state is given, and one constructs a protocol with it. However, in our new security proof using the universal decoder for classical source compression with quantum side information, we estimate the entropy during the protocol and construct a decoder with the estimated value.  This estimation procedure, combined with source compression, nicely fits the previously developed reduction from a permutation invariant state to (a mixture of) i.i.d.~states~\cite{Matsuura2024} at the cost of the polynomial overhead on the failure probability of the universal decoding. Thus, our new security proof can be applied to a wide range of protocols that can be made permutation symmetric. 

Notably, this new approach can achieve the asymptotically optimal key rate~\cite{Devetak2005}, which cannot be achieved in general by the conventional PEC-based approach as pointed out in several literature~\cite{Furrer2014, Tsurumaru2020, Tsurumaru2022}.  We numerically demonstrated the effectiveness of our new approach with the qubit B92 protocol, in which the conventional PEC-based approach fails to achieve the Devetak-Winter rate under the nonzero bit-error rate.  It has been shown that our new approach not only achieves the asymptotically optimal key rate but also has better finite-size performance than that of the conventional analysis.
This does not immediately imply that our new analysis is nearly optimal in the finite-size regime, since it is widely open whether the optimal error exponent of classical source compression with quantum side information for a known state can be achieved in universal coding as well.  (See Remark~\ref{rem:exponent}.)

The implications of our results are many.  Our result is a major step to unify the two mainstreams of security analysis, i.e., the LHL-based approach~\cite{Renner2008} and the PEC-based approach~\cite{Lo1999, Shor2000, Koashi2009, Matsuura2023}, even at the operational level, unlike the conceptual level as has been shown in Ref.~\cite{Tsurumaru2020, Tsurumaru2022}.  Comparing the performance of the LHL-based approach and our new approach in the finite-size scenario should be the next step.  We also find the condition in which the conventional PEC-based security proof can achieve the asymptotically optimal key rate for a given protocol.

Our approach here to reduce the problem of security against general attacks to that against i.i.d.~collective attacks resembles the post-selection technique~\cite{Christandl2009, Nahar2024}, which is conventionally used in the LHL-based approach. The main difference between LHL plus post-selection technique and our approach is that ours gets reduced to the problem of upper-bounding the ``failure probability'' of an information-theoretic task between Alice and Bob, which may be much more tractable than upper-bounding the distance between permutation-symmetric and i.i.d.~protocols as is done in the post-selection technique.  For example, imposing additional restrictions on Alice's and Bob's ability to carry out this information-theoretic task is allowed and still leads to a secure final key.  As a result, the reduction of the dimension of Alice's or Bob's system may be easier to justify in our new approach. 
This may open up a route for applying our technique developed here to continuous-variable QKD protocols in which we need to tackle the infinite dimensionality.

Our new PEC-based analysis developed here has limited applicability compared to the conventional PEC-based analysis since our new analysis requires permutation symmetry in the protocol.
In this regard, another open problem is to develop a classical (universal) source compression with quantum side information for more general non-i.i.d.~quantum states.  If this could be achieved, then we would not need to use the i.i.d.~reduction in the security analysis anymore, which then means we would not need to impose a permutation symmetry.
Considering a set of quantum Markovian states may be a first step towards this, which may be a counterpart in the PEC-based approach for the entropy-accumulation theorem~\cite{Dupuis2020, Metger2022, Metger2023, Arqand2024} in the LHL-based approach.

\begin{acknowledgments}
    T.~M.~thanks helpful discussion with Toyohiro Tsurumaru and Satoshi Yoshida.  This work was supported by JST, PRESTO Grant Number JPMJPR24FA, Japan; the Ministry of Internal Affairs and Communications (MIC), R\&D of ICT Priority Technology Project (grant number JPMI00316); Council for Science, Technology and Innovation (CSTI), Cross-ministerial Strategic Innovation Promotion Program (SIP), ``Promoting the application of advanced quantum technology platforms to social issues'' (Funding agency: QST).
\end{acknowledgments}

\appendix
\section{Equivalence between the actual protocol and the virtual protocol with a successive measurement} \label{sec:equivalence}
In this appendix, we prove Eq.~\eqref{eq:compatibility} for the actual and the virtual protocol introduced in Sec.~\ref{sec:virtual_protocol}. For this, we introduce two variants of the virtual protocol. The first variant, Variant~I, is given as follows.
\bigskip

\noindent --- Variant~I of the virtual protocol ---
\begin{enumerate}
    \item[1'.] Alice prepares an entangled state $\ket{\Psi}_{\sysA \sysC}$ such that there exists a POVM $\{M_{\sysA}^{a}\}_{a\in{\cal X}_{\sysA}}$ that satisfies Eq.~\eqref{eq:purified_state_compatible}, i.e., the same state preparation as that in the virtual protocol. The system $\sysC$ of $\ket{\Psi}_{\sysA \sysC}$ is sent to a quantum channel and received by Bob as a system $B$. Alice and Bob then perform measurements $\{M_{\sysA}^a\otimes M_{\sysB}^b\}_{a\in{\cal X}_{\sysA},b\in{\cal X}_{\sysB}}$ to obtain a pair $(\hat{a}, \hat{b})\in{\cal X}_{\sysA} \times {\cal X}_{\sysB}$. Alice and Bob repeat this quantum communication for $n_{\rm tot}$ rounds, where Alice may initiate each round before the previous round is completed.
    \item[2--5.] The same as those in the actual protocol.
\end{enumerate}
\medskip

From Eq.~\eqref{eq:purified_state_compatible}, it is clear that the joint states between $\bigl((\hat{a}^{(1)},\hat{b}^{(1)}),\ldots,(\hat{a}^{(n_{\rm tot})},\hat{b}^{(n_{\rm tot})})\bigr)$ and Eve's system are the same between the actual protocol and Variant~I. Thus, the probability distributions over the final key length are the same between these two protocols, and the final states $\{\rho_{K_A^n E|\hat{n}_{\rm fin}=n}^{\rm I, fin}\}_{n\in\mathbb{N}}$ for Variant~I satisfy 
\begin{equation}
    \forall n\geq 1,\qquad \rho_{K_A^n E|\hat{n}_{\rm fin}=n}^{\rm I, fin} = \rho_{K_A^n E|\hat{n}_{\rm fin}=n}^{\rm fin}.  \label{eq:equiv_actual-I}
\end{equation}
The second variant, Variant~II, is given as follows.

\bigskip 
\noindent --- Variant~II of the virtual protocol ---
\begin{enumerate}
    \item[1.] The same as that in the virtual protocol.
    \item[2'.] The same as that in the virtual protocol except that, for each $i=1,\ldots,n_{\rm tot}$, Alice measures the system $Q$ by a PVM $\{\Pi_Q[z]\}_{z\in{\cal X}_{\rm sift}}$ and obtain a sifted value $\hat{z}^{(i)}$ after the action of CP maps ${\cal M}_{\sysA \sysB}^{(\xi)}$ and ${\cal J}_{\sysA \sysB \to Q}^{(\xi)}$ to obtain $\hat{\xi}^{(i)}_{\rm pub}$.
    \item[3--5.] The same as those in the actual protocol.
\end{enumerate}
\medskip
 
It can be seen from Eqs.~\eqref{eq:failure_map} and \eqref{eq:success_map} that the joint state between the sifted values $(\hat{z}^{(1)},\ldots,\hat{z}^{(n_{\rm tot})})$, the announcements $(\hat{\xi}^{(1)}_{\rm pub},\ldots,\hat{\xi}^{(n_{\rm tot})}_{\rm pub})$ and Eve's system at the end of Step~2 in Variant~I is the same as that at the end of Step~2' of Variant~II. Thus, the probability distributions over the final key length are the same between these two protocols, and the final states $\{\rho_{K_A^n E|\hat{n}_{\rm fin}=n}^{\rm II, fin}\}_{n\in\mathbb{N}}$ for Variant~II satisfy 
\begin{equation}
    \forall n\geq 1, \qquad \rho_{K_A^n E|\hat{n}_{\rm fin}=n}^{\rm II, fin} = \rho_{K_A^n E|\hat{n}_{\rm fin}=n}^{\rm fin}. \label{eq:equiv_I-II}
\end{equation}
Finally, combining the facts that $U(H_i,\overline{H}_i)$ transforms the $Z$-basis value as in Eq.~\eqref{eq:action_on_Z_basis} and that Variant~II and the virtual protocol give Eve the same classical information through Steps~3 to~5, we have that the probability distributions over the final key length between these two protocols are the same, and that 
\begin{equation}
    \forall n\geq 1, \qquad \sum_{\bm{k}\in\mathbb{F}_p^{n}}\ketbra{\bm{k}}{\bm{k}}_{K_A^n}\rho_{K_A^nE|\hat{n}_{\rm fin}=n}^{\rm virt}\ketbra{\bm{k}}{\bm{k}}_{K_A^n} = \rho_{K_A^nE|\hat{n}_{\rm fin}=n}^{\rm II, fin}. \label{eq:equiv_II-virtual}
\end{equation}
We can thus prove Eq.~\eqref{eq:compatibility} by combining Eqs.~\eqref{eq:equiv_actual-I}, \eqref{eq:equiv_I-II}, and \eqref{eq:equiv_II-virtual}.

\section{Proof of the permutation invariance of the state in the estimation protocol} \label{sec:permutation_symmetry}
In this appendix, we show Eq.~\eqref{eq:insert_random_permutation}.
For any $\tau\in S_{n_{\rm tot}}$, let $\widetilde{\cal V}_{\tau}$ be a unitary channel representation of $\tau$ on the operators on ${\cal H}_{\sysA\sysB}^{\otimes n_{\rm tot}}$ in the same way as ${\cal V}_{\sigma}$ for $\sigma\in S_m$ on ${\cal H}_Q^{\otimes m}$ in Eq.~\eqref{eq:instrument_global}. Then, there exists $\tau_1\in S_m$ and $\tau_0\in S_{n_{\rm tot} - m}$ such that 
\begin{equation}
    {\cal T}_{{\cal I}\to [m]}\circ \widetilde{\cal V}_{\tau^{-1}} = \widetilde{\cal V}_{\tau_1^{-1}\times\tau_0^{-1}} \circ {\cal T}_{{\tau}({\cal I})\to [m]}, \label{eq:permutation_commute}
\end{equation}
where $\tau({\cal I})$ denotes an ordered set with elements $\tau(i)$ for any $i\in{\cal I}$. Note that given ${\cal T}_{{\cal I}\to [m]}$, the triple $(\tau({\cal I}), \tau_0, \tau_1)$ corresponds one-to-one to $\tau\in S_{n_{\rm tot}}$. The permutations $\tau_0$ and $\tau_1$ depend on $\tau$ and ${\cal I}$, but we omit this dependency for brevity. Next, suppose that the state $\widetilde{\cal V}_{\tau^{-1}}(\rho_{\sysA^{n_{\rm tot}}\sysB^{n_{\rm tot}}})$ is substituted to the map $\mathfrak{J}_{\sysA^{n_{\rm tot}}\sysB^{n_{\rm tot}}\to Q^m}^{(m,\Xi,\Theta)}$ in Eq.~\eqref{eq:instrument_global} instead of $\rho_{\sysA^{n_{\rm tot}}\sysB^{n_{\rm tot}}}$. Then, using Eq.~\eqref{eq:permutation_commute}, we have
\begin{align}
    &\mathfrak{J}_{\sysA^{n_{\rm tot}}\sysB^{n_{\rm tot}}\to Q^m}^{(m,\Xi,\Theta)} \circ\widetilde{\cal V}_{\tau^{-1}}(\rho_{\sysA^{n_{\rm tot}}\sysB^{n_{\rm tot}}}) \nonumber \\
    \begin{split}
        &= \sum_{\substack{{\cal I}\subseteq[n_{\rm tot}]\\ :|{\cal I}|=m}} \frac{1}{|S_{m}|}\sum_{\sigma\in S_{m}} \sum_{(\xi_1,\ldots,\xi_{n_{\rm tot}})\in \Delta_m(\Xi)}\,\sum_{\substack{(\theta_1,\ldots,\theta_{n_{\rm tot}})\\:n_{\rm tot}P_{\bm{\theta}}=\Theta}} \\
        &\qquad {\cal V}_{\sigma}\circ\left({\cal J}^{(\xi_1,\theta_1)}_{\sysA\sysB\to Q}\otimes \cdots \otimes {\cal J}^{(\xi_m,\theta_m)}_{\sysA\sysB\to Q} \otimes {\cal M}^{(\xi_{m+1},\theta_{m+1})}_{\sysA\sysB} \otimes \cdots \otimes {\cal M}^{(\xi_{n_{\rm tot}},\theta_{n_{\rm tot}})}_{\sysA\sysB}\right)\circ \widetilde{\cal V}_{\tau_1^{-1}\times\tau_0^{-1}}\circ {\cal T}_{\tau({\cal I})\to [m]}(\rho_{\sysA^{n_{\rm tot}}\sysB^{n_{\rm tot}}}) 
    \end{split}  \\
    \begin{split}
        &= \sum_{\substack{{\cal I}\subseteq[n_{\rm tot}]\\ :|{\cal I}|=m}} \frac{1}{|S_{m}|}\sum_{\sigma\in S_{m}} \sum_{(\xi_1,\ldots,\xi_{n_{\rm tot}})\in \Delta_m(\Xi)}\,\sum_{\substack{(\theta_1,\ldots,\theta_{n_{\rm tot}})\\:n_{\rm tot}P_{\bm{\theta}}=\Theta}} \\
        &\qquad  {\cal V}_{\sigma}\circ {\cal V}_{\tau_1^{-1}}\circ \left({\cal J}^{(\tau_1(\xi_1,\theta_1))}_{\sysA\sysB\to Q}\otimes \cdots \otimes {\cal J}^{(\tau_1(\xi_m,\theta_m))}_{\sysA\sysB\to Q} \right. \\
        &\hspace{3cm} \left.\otimes {\cal M}^{(\tau_0(\xi_{m+1},\theta_{m+1}))}_{\sysA\sysB} \otimes \cdots \otimes {\cal M}^{(\tau_0(\xi_{n_{\rm tot}},\theta_{n_{\rm tot}}))}_{\sysA\sysB}\right) \circ {\cal T}_{\tau({\cal I})\to [m]}(\rho_{\sysA^{n_{\rm tot}}\sysB^{n_{\rm tot}}}).
    \end{split} \label{eq:inner_permutation}
\end{align}
Noticing that the set $\Delta_m(\Xi)$ is invariant under any permutation in $S_m\times S_{n_{\rm tot}-m}$ and that the type $P_{\bm{\theta}}$ is invariant under any permutation in $S_{n_{\rm tot}}$, we have from Eqs.~\eqref{eq:instrument_global} and \eqref{eq:inner_permutation},
\begin{align}
    &\mathfrak{J}_{\sysA^{n_{\rm tot}}\sysB^{n_{\rm tot}}\to Q^m}^{(m,\Xi,\Theta)} \circ\widetilde{\cal V}_{\tau^{-1}}(\rho_{\sysA^{n_{\rm tot}}\sysB^{n_{\rm tot}}}) \nonumber\\
    \begin{split}
    &= \sum_{\substack{{\cal I}'\subseteq[n_{\rm tot}]\\ :|{\cal I}'|=m}} \frac{1}{|S_{m}| }\sum_{\sigma'\in S_{m}} \sum_{(\xi'_1,\ldots,\xi'_{n_{\rm tot}})\in \Delta_m(\Xi)}\,\sum_{\substack{(\theta'_1,\ldots,\theta'_{n_{\rm tot}})\\:n_{\rm tot}P_{\bm{\theta}}=\Theta}}\\
    &\qquad {\cal V}_{\sigma'}\circ\left({\cal J}^{(\xi'_1,\theta'_1)}_{\sysA\sysB\to Q}\otimes \cdots \otimes {\cal J}^{(\xi'_m,\theta'_m)}_{\sysA\sysB\to Q} \otimes {\cal M}^{(\xi'_{m+1},\theta'_{m+1})}_{\sysA\sysB} \otimes \cdots \otimes {\cal M}^{(\xi'_{n_{\rm tot}},\theta'_{n_{\rm tot}})}_{\sysA\sysB}\right)\circ {\cal T}_{{\cal I}'\to [m]}(\rho_{\sysA^{n_{\rm tot}}\sysB^{n_{\rm tot}}})
    \end{split} \\
    &= \mathfrak{J}_{\sysA^{n_{\rm tot}}\sysB^{n_{\rm tot}}\to Q^m}^{(m,\Xi,\Theta)}(\rho_{\sysA^{n_{\rm tot}}\sysB^{n_{\rm tot}}}),\label{eq:invariant_inner_perm}
\end{align}
and thus the map $\mathfrak{J}_{\sysA^{n_{\rm tot}}\sysB^{n_{\rm tot}}\to Q^m}^{(m,\Xi,\Theta)}$ is invariant under the permutation of an input density operator. 

Now, the random permutation CPTP map ${\cal R}_{\sysA^{n_{\rm tot}} \sysB^{n_{\rm tot}}}$ introduced in Eq.~\eqref{eq:insert_random_permutation} can explicitly be written as 
\begin{equation}
    {\cal R}_{\sysA^{n_{\rm tot}} \sysB^{n_{\rm tot}}}(\rho_{\sysA^{n_{\rm tot}} \sysB^{n_{\rm tot}}})\coloneqq \frac{1}{|S_{n_{\rm tot}}|}\sum_{\tau\in S_{n_{\rm tot}}} \widetilde{\cal V}_{\tau} (\rho_{\sysA^{n_{\rm tot}} \sysB^{n_{\rm tot}}}). \label{eq:permutation_symmetrization}
\end{equation} 
From Eq.~\eqref{eq:invariant_inner_perm}, we have 
\begin{equation}
    \mathfrak{J}_{\sysA^{n_{\rm tot}}\sysB^{n_{\rm tot}}\to Q^m}^{(m,\Xi,\Theta)} \circ {\cal R}_{\sysA^{n_{\rm tot}} \sysB^{n_{\rm tot}}}(\rho_{\sysA^{n_{\rm tot}}\sysB^{n_{\rm tot}}}) = \mathfrak{J}_{\sysA^{n_{\rm tot}}\sysB^{n_{\rm tot}}\to Q^m}^{(m,\Xi,\Theta)}(\rho_{\sysA^{n_{\rm tot}}\sysB^{n_{\rm tot}}}).
\end{equation}
Thus, from the definition of $\mathfrak{X}_{X^{m} Q^m}^{(m,\Xi,\Theta)}$ in Eq.~\eqref{eq:def_of_projected_state}, we have Eq.~\eqref{eq:insert_random_permutation}.

\bibliography{new_strategy_security_proof.bib}
\end{document}